\documentclass[aps,prb,superscriptaddress,twocolumn,preprintnumbers]{revtex4-2}

\usepackage[utf8]{inputenc}
\usepackage{tikz-cd}
\usepackage{dsfont}
\usepackage{bm}
\usepackage{mathtools}
\usepackage{amsthm}
\usepackage{amsfonts,amssymb,amsmath}
\usepackage[hidelinks]{hyperref}
\usepackage{physics}
\usepackage{multirow}
\usepackage{graphicx}
 \newtheorem{theorem}{Theorem}
  \newtheorem{definition}{Definition}
\usepackage{verbatim}

\newcommand{\E}[0]{\mathop{{}\mathbb{E}}}

\begin{document}
\title{Hilbert-Space Ergodicity in  Driven Quantum Systems: \\Obstructions and Designs }
\author{Sa\'ul Pilatowsky-Cameo}
\affiliation{Center for Theoretical Physics, Massachusetts Institute of Technology, Cambridge, Massachusetts 02139, USA}
\author{Iman Marvian}
\affiliation{Departments of Physics \& Electrical and Computer Engineering, Duke University, Durham, North Carolina 27708, USA}
\author{Soonwon Choi}
\email{soonwon@mit.edu}
\affiliation{Center for Theoretical Physics, Massachusetts Institute of Technology, Cambridge, Massachusetts 02139, USA}
\author{Wen Wei Ho}
\email{wenweiho@nus.edu.sg}
\affiliation{Department of Physics, National University of Singapore, Singapore 117551}
\affiliation{Centre for Quantum Technologies, National University of Singapore, 3 Science Drive 2, Singapore 117543}

\preprint{MIT-CTP/5675}

\begin{abstract}
Despite its long history, a canonical formulation of quantum ergodicity that applies to  general classes of quantum dynamics, including driven systems, has  not been fully established.
Here we introduce and study a notion of quantum ergodicity for closed systems with time-dependent Hamiltonians,  defined as statistical randomness  exhibited in their longtime dynamics.
Concretely, we consider the temporal ensemble of quantum states (time-evolution operators) generated by the evolution, and investigate the conditions necessary for them to be statistically indistinguishable from uniformly random states (operators) in the Hilbert space (space of unitaries).
We find that the number of driving frequencies  underlying the Hamiltonian  
needs to be sufficiently large for this to occur. 
Conversely,  we show  that  statistical {\it pseudo}-randomness ---  indistinguishability up to some large but finite moment, 
can already be achieved by a quantum system driven with a single frequency, i.e., a Floquet system, as long as the driving period is sufficiently long. Our work relates the complexity of a time-dependent Hamiltonian and that of the resulting quantum dynamics, and  offers a fresh perspective to the established topics of quantum ergodicity and chaos from the lens of quantum information.
\end{abstract}

\maketitle

\section{Introduction}
Ergodicity in classical  systems is a well established, unambiguous concept: It is the property of dynamics exploring all allowed configurations, irrespective of initial state. Quantum ergodicity, on the other hand, is formulated rather differently, and typically in an inherently nondynamical fashion \cite{Shnirelman1973,Sunada1997}: In systems with a semiclassical limit, it is  taken to be the feature of high-energy  eigenstates having probability densities weakly tending to a uniform distribution in phase space \cite{Berry1977}. This definition though, does not cover all quantum systems, as there are many Hamiltonians without an obvious semiclassical limit, e.g., systems of interacting qubits. Instead, an appeal is often made to statistical similarities of the distribution of energy levels and associated energy eigenstates to those of certain random matrix classes, such as in the eigenstate thermalization hypothesis (ETH) \cite{Srednicki1994} and the Bohigas-Giannoni-Schmit conjecture \cite{Bohigas1984}. Still, such a definition is arguably also not complete, as it presupposes the existence of stationary states in dynamics --- and not all quantum systems exhibit these. These include Hamiltonians with general time dependence, or dynamics arising from (potentially spatiotemporally random) quantum circuits, a class of quantum dynamics that has been the subject of much study recently \cite{Fisher2023}. 
As can be seen, there is no unambiguous,  common notion of ergodicity that applies to all  systems in the quantum setting.  

In this work, we investigate a notion of quantum ergodicity that can be universally attributed to closed quantum dynamics with generic time dependence, which harkens back to ergodicity of classical dynamical systems: whether a quantum system explores all of its ``ambient space" over time. We consider two natural dynamical objects that can capture this behavior, both of which are always present for any closed quantum system undergoing unitary dynamics. First, we consider the temporal ensemble of quantum states $\{ |\psi(t)\rangle\}_t$ beginning from some initial state $|\psi(0)\rangle$, with the natural ambient space  being the entire Hilbert space. Second, we study the temporal ensemble of time-evolution operators $\{ U(t) \}_t$, which propagates the system from the initial time $t = 0$ to a later time $t$, with the ambient space being the manifold of unitary operators acting on the Hilbert space. Quantum ergodicity according to this viewpoint inquires if the temporal ensembles of states or unitaries  {\it uniformly} cover their respective spaces over long times. 

A previous recent work \cite{Pilatowsky2023} had already proposed  the notion of quantum states  uniformly covering the Hilbert space in time, dubbed ``complete Hilbert-space ergodicity" (CHSE), as a novel notion of quantum ergodicity. It also rigorously demonstrated a class of discrete-time driven systems, which despite their simplicity (encapsulated by a notion of having ``low complexity''), surprisingly exhibits such behavior. Here, one of our goals is to further ground this concept, by  identifying general physical principles which allow or forbid CHSE. Additionally, we extend this dynamical version of quantum ergodicity to that of  statistics of the unitary time-evolution operators themselves, a notion we dub ``complete unitary ergodicity" (CUE). CUE is a stronger dynamical version of quantum ergodicity, as it implies CHSE, but not vice versa.

We note that this  generalization of the notion of classical ergodicity to quantum dynamics --- that time averaging equals space averaging --- is ostensibly natural, but yet evidently has not been widely adopted as a standard definition of quantum ergodicity. A moment's thought reveals why this may be so: Under dynamics by a time-independent Hamiltonian, it can immediately be observed that the populations on energy eigenstates are always conserved, leading to an obstruction of coverage of the ambient Hilbert or unitary space. In other words, CHSE or CUE cannot occur for dynamics under any static Hamiltonian $H$,  rendering such dynamical notions of quantum ergodicity ineffectual. However, the key insight of our analyses, as well as those of Ref.~\cite{Pilatowsky2023}, is the realization that these obstructions need not apply in  Hamiltonians  $H(t)$ that have {\it general time dependence}. In this work, we specifically focus on the class of quantum Hamiltonians driven by multiple (rationally independent) frequencies, called quasiperiodically driven systems \cite{Ho1983,Luck1988,Casati1989,Feudel1995,Bambusi2001,Gentile2003,Chu2004,Gommers2006,Chabe2008}, and 
derive how despite potentially having ``quasienergy states," the  analog of stationary states for this class of dynamics, they can under certain conditions already achieve CHSE and/or CUE.

 Concretely, we consider here  $d$-dimensional quantum systems quasiperiodically driven by $m$ rationally independent frequencies, and assume the existence of quasienergy states in dynamics (we note this is a nontrivial assumption and it may not always hold true, see Refs.~\cite{Jauslin1991,Blekher1992}). Equivalently, these can be thought of as quantum systems driven by $m$ external classical harmonic baths with different fundamental frequencies.
 Intuitively, a larger number of drives, i.e., baths,  generates more complex dynamics. For example, one can model a quantum system driven by random white noise in the limit  $m \to \infty$. 
  We might thus expect that the ability of a system to uniformly cover its Hilbert or unitary space, depends on the number of frequencies $m$  of the underlying Hamiltonian. Indeed, in what follows we rigorize such expectation, showing how $m$ governs the possibility or impossibility of CHSE and CUE. The key tool we use is quantum information theoretic: We leverage the concept of {\it state (unitary) designs}, to precisely quantify the statistical indistinguishability of the distribution of the  temporal ensemble of states (unitaries) to the corresponding uniformly random ensemble in their respective spaces. A summary of our main results is as follows:
\begin{itemize}
    \item  Complete Hilbert-space ergodicity (CHSE) cannot be satisfied if $m < 2(d-1)$. That is, a time-quasiperiodic quantum system driven by a limited number of frequencies cannot yield dynamics in which an arbitrary state uniformly explores all of the Hilbert space over time. 
    \item  Complete unitary ergodicity (CUE) cannot be satisfied if $m < d(d-1)$. 
    This is a more restrictive statement that a time-quasiperiodic quantum system driven by too few frequencies cannot generate time-evolution operators  which are uniformly distributed in the unitary space.
    \item Conversely, we explicitly construct families of time-quasiperiodic quantum  Hamiltonians with  $m = d^2-2$  fundamental frequencies, each possessing quasienergy states, which provably exhibit CUE, and therefore CHSE.
\end{itemize}

These three statements are depicted in Fig.~\ref{fig:resultssumm}.\\

It is also possible to relax the condition of full indistinguishability of the  distributions of temporal and spatial ensembles, and demand only indistinguishability of  moments up to some finite order $k \in \mathbb{N}$. This property is called {\it statistical pseudorandomness}. We note that statistical pseudorandomness of states or unitaries has been used as a diagnostic for the presence of quantum information scrambling \cite{Roberts2017,Cotler2017}, and thus our notion of quantum ergodicity is intimately tied to (one version of) quantum chaos. 
Technically, equality of only up to $k$ moments amounts to probing whether the temporal ensemble of states (unitaries) forms a {\it state (unitary) $k$-design}. 
With this, we can also show the following.

\begin{itemize}   
    \item If we  demand a restricted level of quantum ergodicity wherein the temporal ensemble reproduces only the uniform distribution up to a finite $k$th moment, then this can be achieved already by a time-periodic (i.e., Floquet) Hamiltonian. However, the magnitude of the Hamiltonian (or equivalently the length of the Floquet period) necessarily needs to grow with $k$ and $d$ in a quantifiable fashion [Eq.~\eqref{eq:bound_on_B}]. This  captures the intuitive fact that the amount of physical resources required --- strength of the Hamiltonian for a fixed time, or driving duration for a fixed power --- needs to be large in order for a high degree of ergodicity to be achieved. 
\end{itemize}

\begin{figure}[t]
    \centering
    \includegraphics{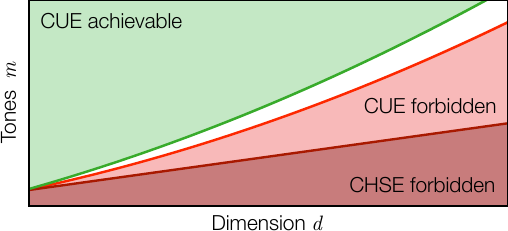}
    \caption{Achievability of complete unitary ergodicity (CUE) and complete Hilbert-space ergodicity (CHSE) in time $m$\protect\nobreakdash-quasiperiodic systems with quasienergy eigenstates. The red regions represent no-go theorems for CHSE and  CUE presented in Sec.~\ref{sec:no-go-theorems}, for $m< 2(d-1)$ (Theorems~\ref{th:01} and~\ref{th:02}) and $m< d^2-d$ (Theorem~\ref{th:03}), respectively. The green region represents an explicit construction of a $(m=d^2-2)$\protect\nobreakdash-quasiperiodic system with QEs that satisfies CUE, presented in Sec.~\ref{sec:CUEModel}.}
    \label{fig:resultssumm}
\end{figure}

Our work represents a step toward a unified understanding of quantum  ergodicity in generic time-dependent quantum systems. Our dynamical notion of ergodicity harmonizes with the notions in classical systems, and further provides a physical understanding of how thermalization arises in these systems, without reference to stationary states of dynamics.

This work is organized as follows.  We begin by introducing the relevant concepts underlying our analysis. In Sec.~\ref{sec:CHSE}, we first introduce our dynamical notion of quantum ergodicity, CHSE and CUE, defined via the tool set of quantum state and unitary designs.
In Sec.~\ref{sec:QP}, we recap  quasiperiodically driven systems and their structure in dynamics and, in particular, a generalization of the Floquet decomposition into windings of quasienergies and quasienergy eigenstates on high-dimensional tori. 
The reader knowledgeable in these topics may elect to skip this section. 
In Sec.~\ref{sec:no-go-theorems}, we present our first results: three no-go theorems establishing conditions under which CUE and CHSE are physically impossible, when the number of frequencies driving the Hamiltonian are not sufficiently large, in relation to the dimension. Section~\ref{sec:numerics} presents a numerical analysis of three toy models, in which we study the consequences of our results at the level of few-body observables. 
In Sec.~\ref{sec:CUEModel}, we demonstrate a converse to our no-go theorems: an explicit construction of a quasiperiodically driven system which satisfies CUE (and hence CHSE), with a sufficiently large number of driving frequencies.
In Sec.~\ref{sec:finiteHSEPeriodicsystems}, we consider relaxing ergodicity  to comparing finite moments. We show that Floquet systems can achieve this relaxed notion of ergodicity by providing examples in both continuous and discrete time.
Lastly, in Sec.~\ref{sec:discussion}, we close with a discussion of connections to previous works and future directions. 

Before proceeding, let us remark that dynamical notions of quantum ergodicity have recently been  discussed in other works \cite{Vikram2023, Kaneko2020,Fava2023,Shou2023,Mark2024}. By borrowing notions from classical ergodic theory, Ref.~\cite{Vikram2023} provides a definition of quantum ergodicity that requires that  certain basis vectors are cyclically transported to each other in a precise sense. Separately, Refs.~\cite{Kaneko2020,Fava2023} build connections between temporal unitary designs and the ETH. Although the conservation of energy prevents the temporal ensemble from forming an exact $k$-design, these references relax the $k$-design condition  in two distinct ways: Reference~\cite{Kaneko2020} introduces a partial unitary design, which restricts to expectation values of some observables, while Ref.~\cite{Fava2023} uses free probability to construct a notion dubbed $k$\nobreakdash-freeness.
Common to these works is the focus on time-independent systems. In contrast, the stronger dynamical version of quantum ergodicity studied in our work requires the absence of any conserved quantity, and is suited for time-dependent systems without energy conservation.  Bridging our work and these other notions of quantum ergodicity is an interesting question.

\section{Dynamical formulation of quantum ergodicity}
Consider a $d$-dimensional quantum system undergoing dynamics under a time-dependent Hamiltonian or a quantum circuit. 
An immediate question arises,  which forms the fundamental motivation behind our work: Is there a sense in which such a system can be termed {\it ergodic}?

In this section, we will introduce a concept of quantum ergodicity defined in terms of statistical similarities of  temporal ensembles of  dynamical objects  --- namely, time-evolved wavefunctions as well as time-evolution operators, to ensembles of such objects distributed unbiasedly (i.e., uniformly) in the respective spaces that they live in. 
In more pedestrian terms, this is  the familiar idea of ``time averaging equals space averaging'' in classical dynamics, applied to the quantum setting. 

\label{sec:CHSE}
\subsection{Hilbert-space ergodicity (HSE)}
We start by discussing quantum ergodicity at the level of  quantum states $\ket{\psi(t)}=U(t)\ket{\psi(0)}$ uniformly covering the Hilbert space over time, a notion first introduced already in  Ref.~\cite{Pilatowsky2023}, dubbed ``complete Hilbert-space ergodicity'' (CHSE). More precisely, since global phases are irrelevant, it was proposed to consider whether the ensemble of time-evolved density matrices $\{ \psi(t) \}_{t \geq 0}$ called the ``temporal ensemble" (if it exists \footnote{It is assumed that there is a well-defined limiting distribution as $t \to \infty$, which may not always hold in certain pathological cases.}), where $\psi(t) = |\psi(t)\rangle \langle \psi(t)|$, is statistically indistinguishable to the ensemble of states $\{ \phi \}_\text{Haar}$ called the ``spatial ensemble." The latter is defined as the set of states randomly sampled without preference to a particular direction in the projective Hilbert space $\mathbb{P}(\mathbb{C}^d)=\{\psi=\dyad{\psi} \,\colon \ket{\psi}\in \mathbb{C}^d,\,\braket{\psi}=1\}$, or in other words, the set where states $\phi$ and $V\phi V^\dagger$ occur equally likely, where $V$ is drawn from the unique, uniform Haar measure on the space of unitaries \cite{Diestel2014Haar}. 
Formally,  we have the following.
\begin{definition} [CHSE]
\label{def:CHSE}
Complete Hilbert-space ergodicity (CHSE) \cite{Pilatowsky2023} is the property of quantum dynamics wherein the temporal and spatial ensembles of quantum states are statistically indistinguishable for any initial state $\psi(0)\in \mathbb{P}(\mathbb{C}^d)$, that is, $\{ \psi(t)\}_{t \geq 0} \sim  \{ \phi \}_\text{Haar}$, where ``$\sim$'' denotes equality in distribution. 
\end{definition}

To make the comparison quantitative, we can consider finite moments of the respective distributions.  For the temporal ensemble, the $k$th moment is defined as 
\begin{align}
    \rho_{\mathrm{time}}^{(k)} : =\E_{t\geq 0} [{\psi(t)}^{\otimes k}]=\lim_{T\to \infty}\frac{1}{T}\int_{0}^T \dd{t} [U(t)\psi(0)U(t)^\dagger]^{\otimes k},
\end{align}
which involves $k$ replicas of the time-evolved state, while the  $k$th moment of the spatial ensemble $\{ \phi\}_\text{Haar}$ is defined as:
\begin{align}
    \rho^{(k)}_\text{Haar}:= \E_{\phi \in \mathbb{P}(\mathbb{C}^d) } [{\phi}^{\otimes k}]=\int \dd{U} (U\phi_0 U^\dagger)^{\otimes k}, 
\end{align} 
where $\dd{U}$ is the Haar measure on the unitary space and $\phi_0$ any fixed reference state. We note that $\rho^{(k)}_\text{Haar}$ have simple, closed-formed expressions as sums of permutation operators over the $k$-replicated Hilbert space [see Eq.~\eqref{eq:rho_Haar_symm_proj}], which can be derived using Schur’s lemma in representation theory \cite{Harrow2013}. As an example, $\rho^{(1)}_\text{Haar} = \mathds{1}/d$ is the maximally entropic state, where $\mathds{1}$ is the identity operator on a single copy of the Hilbert space, while $\rho^{(2)}_\text{Haar} = (\mathds{1} + S)/d(d+1)$, where here $\mathds{1}$ $(S)$ is the identity (swap) operator on the tensor product of two Hilbert spaces.
Using the $k$th moments $\rho^{(k)}_\text{Haar}$, we can  define a less restrictive notion of Hilbert-space ergodicity in terms of statistical indistinguishability of only up to $k$-moments:  

\begin{definition}[$k$\nobreakdash-HSE]
A closed quantum system is said to exhibit Hilbert-space $k$-ergodicity ($k$\nobreakdash-HSE), for $k\in\mathbb{N}$, if for any initial state $\psi(0)\in \mathbb{P}(\mathbb{C}^d)$,
\begin{equation}
\label{eq:kHSE_cond}
    \rho_{\mathrm{time}}^{(k)}= \rho_{\mathrm{Haar}}^{(k)}.
\end{equation}

\end{definition}
\noindent 
Any standard matrix norm can be used to ascertain this equality (captured by vanishing of the norm of $\rho^{(k)}_\text{time} - \rho^{(k)}_\text{Haar}$), but it is conventional to use the trace distance $D(\rho, \sigma) := \frac{1}{2} \| \rho - \sigma\|_1$, where $\|\cdot \|_1$ is the trace norm, given by the sum of the absolute value of the eigenvalues. This is because $ \rho_{\mathrm{time}}^{(k)}$ and $\rho_{\mathrm{Haar}}^{(k)}$ have interpretations of  density operators on the $k$-replicated Hilbert space, and the trace norm operationally  captures the probability of distinguishing these two states under an optimal measurement.

In the parlance of quantum information theory, $k$\nobreakdash-HSE is the statement that the temporal ensemble forms a {\it (state) $k$-design} (see Ref.~\cite{Mele2023} and Appendix~\ref{app:designs}). 
Note that $k$\nobreakdash-HSE implies $k'$-HSE for $k'\leq k$ but not vice versa, and thus forms a hierarchical definition of more and more restricted notions of quantum ergodicity for higher $k$ (see Corollary~\ref{cor:relationsHSEandUE} and Fig.~\ref{fig:hier}).  
CHSE, which is at the top of this hierarchy, is then recovered by demanding equality for all $k$.
\begin{definition} [CHSE; equivalent definition \footnote{This is equivalent to Definition~\ref{def:CHSE} because we are considering finite-dimensional quantum systems. Then knowledge of all moments uniquely determines a distribution (this is known as the moment problem in mathematics). This follows from the   Weierstrass approximation theorem, which states that polynomials are dense under the uniform norm in the space of continuous functions.}]
If a system exhibits $k$\nobreakdash-HSE for all $k$ for any initial state $\psi(0)\in \mathbb{P}(\mathbb{C}^d)$, then it is said to exhibit CHSE. 
\end{definition}

In terms of  physical observables, $k$\nobreakdash-HSE constrains  the behavior of  time-averaged expectation values $\tr(O^{(k)}\psi(t)^{\otimes k})$ of a joint observable $O^{(k)}$ on the $k$-replicated Hilbert space. In the case of a product observable $O^{(k)}=O^{\otimes k}$, this is the time-averaged $k$th power of $\expval{O}{\psi(t)}$. For example, $1$-HSE implies that the time average of  $O$, given by  $\E_{t\geq 0}[\langle\psi(t)|O|\psi(t)\rangle] = \lim_{T \to \infty} \frac{1}{T}\int_0^T dt \langle \psi(t)|O|\psi(t)\rangle 
$,
equals
$
\E_{\phi\in \mathbb{P}(\mathbb{C}^d)}[\expval{O}{\phi}] =\tr(O)/d$  regardless of the initial state $\psi(0)$, i.e., the system over long times reproduces expectation values within the infinite-temperature state. More generally, $k$\nobreakdash-HSE implies $\E_{t\geq 0}[\langle\psi(t)|O|\psi(t)\rangle^k] =  \lim_{T \to \infty}\frac{1}{T} \int_0^T dt \langle \psi(t)|O|\psi(t)\rangle^k $ is equal to $\E_{\phi\in \mathbb{P}(\mathbb{C}^d)}[\expval{O}{\phi}^k]=\text{tr}({O}^{\otimes k}\rho^{(k)}_{\text{Haar}})$ \footnote{This can be seen by multiplying by $O^{\otimes k}$ on both sides of Eq.~\eqref{eq:kHSE_cond} and taking the trace.}, which is independent of $\psi(0)$ and can be calculated using the closed-form expression of $\rho^{(k)}_\text{Haar}$ described above, which physically constrains not only the mean but also temporal fluctuations and beyond to mimic those computed for random states. For instance, for $k=2$, the spatial averaging yields explicitly $[\text{tr}(O^2) + \text{tr}(O)^2]/d(d+1)$.  CHSE is   the strongest statement that the time average of {\it any} (analytic) function $f$ is equal to its spatial average, i.e., $\E_{t\geq 0}[f({\psi(t)})]=\E_{\phi\in\mathbb{P}(\mathbb{C}^d)}[f({\phi})]$. This is the consequent of Birkhoff's ergodic theorem \cite{Cornfeld1982}, applied to a quantum system.

\begin{figure}
    \centering
    \includegraphics[width=0.9\columnwidth]{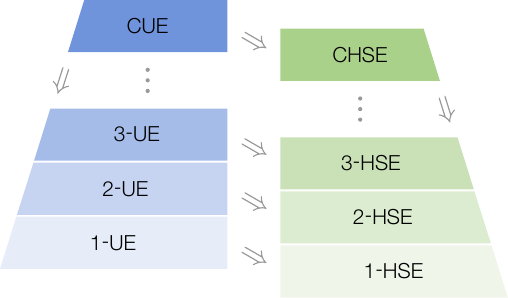}
    \caption{Dynamical notions of quantum ergodicity and their relations. Unitary $k$-ergodicity ($k$-UE) (left) and Hilbert-space $k$-ergodicity ($k$-HSE) (right), with complete unitary ergodicity (CUE) and complete Hilbert-space ergodicity (CHSE) on top, respectively.  Arrows indicate logical implication.}
    \label{fig:hier}
\end{figure}

\subsection{Unitary ergodicity (UE)}
We propose in this work to also consider a different notion of dynamical quantum ergodicity, captured by the equivalence of statistics of the ensemble of time-evolution operators $\{ U(t)\}_{t \geq 0}$ to the uniform ensemble of operators in the space of unitaries. The evolution operators are given by the time-ordered exponentials  $U(t)=\mathcal{T}\text{exp}(-i\int_0^t \dd{\tau} H(\tau))$ which propagate the system from time $0$ to time $t$. More precisely, it is natural to consider the set of unitary
{\it quantum channels} $\{ \mathcal{U}(t) \}_{t \geq 0}$, defined by $\mathcal{U}(t)[\psi]= U(t)\psi U(t)^\dagger$ (we consider channels as opposed to the unitary time-evolution operators themselves, as global phase information is irrelevant). For technical convenience, the map $\mathcal{U}(t)$ can be vectorized into the form $U(t)^* \otimes U(t)$. These are elements of the projective unitary group $\operatorname{PU}(d)=\{V^* \otimes V\, : \, V\in \operatorname{U}(d)\}$  \footnote{A simple example is the qubit $d=2$ case, where $\operatorname{PU}(2)\cong \text{SO}(3)$ is the group of all rotations of three-dimensional space, which acts by rotating states around the Bloch sphere $\mathbb{P}(\mathbb{C}^2)$.} for which there is a notion of a uniform (Haar) ensemble $\{ \mathcal{V} \}_\text{Haar} := \{ V^* \otimes V \}$ describing the distribution of unitary channels obtained from randomly sampling from the Haar measure $dV$ on the space of unitaries \cite{Diestel2014Haar}. Our proposed dynamical notion of quantum ergodicity in this scenario would then amount to asking whether the temporal ensemble is equivalent to the spatial ensemble, which in analogy to CHSE we dub ``complete unitary ergodicity'':
\begin{definition} [CUE]
Complete unitary ergodicity is the property of quantum dynamics wherein the temporal ensemble of unitary time-evolution operators and spatial ensemble of unitary operators are statistically indistinguishable, $\{\mathcal{U}(t) \}_{t \geq 0} \sim  \{ \mathcal{V} \}_\text{Haar}$, where ``$\sim$'' denotes equality in distribution. 
\end{definition}
  
 Such an equality may, once again, be probed by comparing moments of the respective distributions, defined for the $k$th moment for the temporal ensemble as $\E_{t\geq 0} [U(t)^{\otimes k,k}]:=\lim_{T \to \infty} \frac{1}{T}\int_0^T dt [U(t)^*\otimes U(t)]^{\otimes k}$, and for the spatial ensemble as 
 $\E_{V \sim \mathrm{Haar}}[V^{\otimes k,k}]$, where ${V}^{\otimes k,k} :=(V^*\otimes V)^{\otimes k}$. 
 The latter can be exactly computed using   so-called Weingarten calculus and have closed-form expressions \cite{Collins2006}. For example, the first moment
$\E_{V \sim \mathrm{Haar}}[{V^*\otimes V}]$ is equal to the quantum channel $\mathcal{C}^{(1)}[O] =  \tr(O)\mathds{1}/d$, meaning that, under $1$-UE, the time average of any observable $O(t)$ in the Heisenberg picture is $\E_{t\geq 0} [O(t)]=\tr(O)\mathds{1}/d$. We can then define $k$-unitary ergodicity as the statement of  indistinguishability only up to the $k$th moment:
\begin{definition}[UE]
For $k\in\mathbb{N}$, we say that the evolution given by a Hamiltonian $H(t)$ exhibits unitary $k$-ergodicity ($k$\nobreakdash-UE) if the evolution operator $U(t)$  satisfies
\begin{equation}
\label{eq:UEDef}
    \E_{t\geq 0} [U(t)^{\otimes k,k}]=\E_{V \sim \mathrm{Haar}}[{V^{\otimes k,k}}].
\end{equation}
\end{definition}
Again, any vanishing of matrix norm for the difference between the left- and right-hand sides can be used to numerically ascertain $k$\nobreakdash-UE, though  it is common practice to compare the so-called ``frame potentials'' (which is related to the Frobenius norm), viz. asking if 
\begin{equation}
\label{eq:framePotEquality}
    \E_{t\geq 0}\E_{t'\geq 0} \Big[\abs{\tr(U^\dagger(t')U(t))}^{k}\Big] \stackrel{?}{=}\E_{V,W\sim\mathrm{Haar}}\Big[\abs{\tr(W^\dagger V)}^{2k}\Big].
\end{equation}
In the parlance of quantum information theory, $k$\nobreakdash-UE is the statement that the temporal ensemble forms a {\it unitary $k$-design}.

It is straightforward to note that $k$\nobreakdash-UE  implies $k$\nobreakdash-HSE, but the converse is not true (see Appendix~\ref{app:designs}).
Thus, $k$\nobreakdash-UE is an inequivalent, strictly stronger version of quantum ergodicity compared to $k$\nobreakdash-HSE. 
Further, $k$\nobreakdash-UE defines a hierarchical definition of more restricted notions of quantum ergodicity: $k$\nobreakdash-UE implies $k'$-UE for $k' \leq k$ but not vice versa (see  Corollary~\ref{cor:relationsHSEandUE} and Fig.~\ref{fig:hier}). The most restrictive condition is when $k$\nobreakdash-UE is satisfied for all $k$, leading us back to CUE:

\begin{definition} [CUE; equivalent definition]
    If a system exhibits $k$\nobreakdash-UE for all $k$, then it exhibits CUE.
\end{definition}
\noindent Similarly to $k$\nobreakdash-UE and $k$\nobreakdash-HSE, CUE implies CHSE but not vice versa. 

\subsection{ Achievability of HSE or UE and conservation laws}

We briefly comment here on the achievability of HSE or UE in the presence of conservation laws in dynamics. 
As the definition of HSE or UE entails a comparison of the temporal ensemble to the reference {\it uniform} (i.e., unbiased) distribution in the Hilbert space (space of unitaries), it is intuitively clear that any conserved quantities will preclude HSE (UE), since there will be ``bias'' in dynamics toward them (of course, an interesting question, which we do not address here, is how to properly modify the reference distribution in order to account for conserved quantities \cite{Mark2024}). 
For example, in a time-independent quantum system which has energy conservation, not even $1$-HSE can be achieved: If $\ket{\psi}$ is an eigenstate of the Hamiltonian, then its time average remains pure:   $\psi=\E_{t\geq 0}[\psi(t)]$, far off from a maximally mixed state  $\E_{\phi \in \mathbb{P}(\mathbb{C}^d)}[\phi]= \mathds{1}/d$.

Achieving quantum ergodicity defined by HSE or UE therefore necessarily requires considering systems with time dependence, such that there are no conservation laws. A trivial example of dynamics which satisfies CUE is a drive $U(t)$ where at every integer time $t$ an independent Haar-random unitary is applied. Then, the wavefunction undergoes a random walk in the Hilbert space. The time dependence of such a drive is, however, maximally complex: At each time-step we need to  specify a completely new random matrix. A natural question to ask is whether or not CHSE or CUE (or, more generally, different levels of the hierarchy of complete ergodicity) can be achieved with time-dependent systems with more succinct, deterministic, descriptions.
Surprisingly, Ref.~\cite{Pilatowsky2023} gave an explicit example in the affirmative, in terms of a family of simple, deterministic, low-complexity quantum drives, derived from the Fibonacci word and its variants, which provably exhibits CUE (and hence CHSE). However, a more general theory that allows us to systematically determine when CUE or CHSE occurs or not, is at the present time still not fully established. One of the aims of this work is to present a step in this direction. 

 In the next section, we introduce the notion of time quasiperiodicity, which allows us to classify the time dependence of a system in increasing levels of complexity. Using this, we will systematically classify the time complexity  required to achieve the different levels of HSE and UE in the class of quasiperiodically driven systems.

\section{Time-quasiperiodic quantum systems}
\label{sec:QP}
In this section, we give a brief introduction to the class of quantum systems which are quasiperiodically driven by $m$ frequencies, and discuss the structure of the dynamics they generate, in particular, the possibility of decomposing dynamics into quasienergies and quasienergy states. 

Time-quasiperiodic systems are the direct generalization of a Floquet system, i.e., a system driven periodically by a single fundamental frequency \cite{Ho1983,Luck1988,Casati1989,Feudel1995,Bambusi2001,Gentile2003,Chu2004,Gommers2006,Chabe2008}. This class of systems has gained much recent interest \cite{Verdeny2016,Nandy2018, Yang2018,Long2022,Tristan2022,Das2023}, as they may host novel and exotic dynamical phases like time-quasiperiodic topological phases \cite{Martin2017,Crowley2019,Dominic2020} and time quasicrystals \cite{Dumitrescu2018,Autti2018,Giergiel2019,Zhao2021,Zalatel2023}. 

\subsection{Definition}
Floquet Hamiltonians are those that periodically repeat themselves in time, $H(t)=H(t+T)$, where  $\omega$ is the fundamental driving frequency and $T=2\pi/\omega$ the corresponding period. An equivalent way of understanding such Hamiltonians, which allows for an immediate generalization to multifrequency drives, is to define an underlying Hamiltonian $\hat{H}(\theta)$ on the circle $\mathbb{S}^1$, with coordinate $\theta\in [0,2\pi)$. Then, a time-periodic Hamiltonian can be defined via setting $\theta = \omega t + \theta_0 \mod 2\pi$ for some initial phase $\theta_0$ (which we will typically set to be 0),  i.e., $H(t)=\hat{H}(\omega t)$. A multitone, or time-quasiperiodic Hamiltonian then straightforwardly follows by generalizing this concept, by promoting the circle $\mathbb{S}^1$ to the torus $\mathbb{T}^m = \underbrace{\mathbb{S}^1 \times \mathbb{S}^1\times \cdots \times \mathbb{S}^1}_{m\text{ times}} \ni \bm \theta=(\theta_1,\dots,\theta_m)$, and $\omega$ to $\bm \omega=(\omega_1,\dots,\omega_m)$. 
Precisely, we have the following.
\begin{definition}[$m$-time-quasiperiodic Hamiltonian]
\label{def:qp}
    Given a Hamiltonian $H(t)$, with $t\in\mathbb{R}$,  we say that $H$ is time-quasiperiodic with $m$ tones, or $m$-time-quasiperiodic if there exists a so-called parent Hamiltonian $\hat{H}(\bm \theta)$ piecewise-smoothly \footnote{Possibly with countably many pieces.} defined on the $m$-dimensional torus $\mathbb{T}^m\coloneqq\{\bm \theta=(\theta_1,\dots,\theta_m) \,|\, \theta_j\in [0,2\pi)\}$ such that  
\begin{equation}
    H(t)={\hat H}(\bm \omega t),\end{equation}  for some {\it frequency vector } $\bm \omega=(\omega_1,\dots,\omega_m)$, where the winding $\bm \omega t$ is taken modulo $2\pi$ at each entry.
     Furthermore, we require that  $m$ is the smallest integer such that the above decomposition holds.
\end{definition}

\noindent 
As $m$ has to be the smallest possible number of tones, the frequency vector $\bm \omega$ has to be {\it rationally independent}, meaning that the only integer solution $\bm n\in\mathbb{Z}^m$ to the equation
$\bm n\cdot \bm \omega=0$ is $\bm n=\bm 0$, i.e., $\bm\omega$ constitute $m$ independent fundamental tones \footnote{This follows from the minimality of $m$. If we had ${\bm n\cdot\bm \omega}=0$ with some entry $n_j\neq 0$, then we could write $\omega_j=-(\sum_{k\neq j}n_k\omega_k)/n_j$, which would allow us to reduce $m$, by writing $\theta_j$ in terms of the other $\theta_k$ in  $H(\bm \theta)$. }. Henceforth, for simplicity in the notation, we will drop the hat in the parent Hamiltonian $\hat{H}(\bm \theta)$, and simply write $H(\bm \theta)$. This is a standard abuse of notation, as $H(t)$ and $H(\bm \theta)$ are functions technically defined in different domains, but they can easily be distinguished by their arguments \cite{Dominic2020}. A more familiar definition of an $m$-time-quasiperiodic Hamiltonian, which is equivalent for sufficiently well-behaved functions, is the statement that $H(t)$ can be written as a convergent Fourier series with $m$ rationally independent fundamental frequencies,
\begin{align}
    H(t) = \sum_{\bm{n} \in \mathbb{Z}^m} {H}_{\bm{n}} e^{i \bm{n} \cdot \bm{\omega} t},
\end{align}
where $H_{\bm{n}}$ are its Fourier modes (over the torus). In modern quantum simulation experiments, engineering time-quasiperiodic driving with a large number of tones $m$ is readily achievable.

More generally, an $m$-time-quasiperiodic Hamiltonian constitutes an example of an $m$-time-quasiperiodic function $F(t) = \hat{F}(\bm{\omega} t) = \sum_{\bm{n} \in \mathbb{Z}^m} {F}_{\bm{n}} e^{i \bm{n} \cdot \bm{\omega} t}$, where the parent function $\hat{F}$ and frequency vector $\bm\omega$ have all the same properties as that listed in Definition \ref{def:qp}.

\subsection{Generalized Floquet decomposition}
What is understood about the nature of quantum dynamics generated by time-quasiperiodic Hamiltonians? In the case of $m=1$, we recover time-periodic or Floquet drives, for which the Floquet theorem guarantees that there exists a set of quasienergy eigenstates which are also periodic in time \cite{Shirley1965}. This is captured by the statement that the unitary time-evolution operator admits a decomposition 
\begin{align}
    U(t) = P(\omega t) e^{-i Q t},
\end{align}
where $Q$ is the so-called Floquet Hamiltonian whose  $d$ eigenvalues, called quasienergies, and eigenvectors are defined via $Q|\alpha\rangle = q_\alpha |\alpha\rangle$.
$P(\omega t)$ is a periodic unitary with identical period as the driving Hamiltonian and satisfies $P(0) = \mathds{1}$, and thus is descended from a piecewise-smooth parent unitary $P(\theta)$ defined on the circle $\mathbb{S}^1$. One may thus construct  quasienergy eigenstates (QEs) that live on the circle, defined via
\begin{align}
    |\alpha(\theta)\rangle = P(\theta) | \alpha\rangle.
\end{align}
Note that the decomposition into the Floquet Hamiltonian and periodic unitary is not unique: One can shift the quasienergies $q_\alpha + n \omega$ by any integer $n \in \mathbb{Z}$ and redefine the appropriate component of $P(\omega t)$  with a winding phase.  
One sees from this decomposition that if we were to view a Floquet system at stroboscopic times $t = n T$ where $n \in \mathbb{Z}$, then the system can equivalently be thought of as undergoing dynamics under a time-independent  Hamiltonian $Q$, that is, $U(nT) = e^{-i Q n T}$. This property of decomposability of dynamics into that of  a static Hamiltonian, up to a periodic envelope, is known mathematically as {\it reducibility} \cite{Murdock1978,Jorba1992}.

When $m >  1$, it is natural to assume that a generalized Floquet decomposition, or reducibility of dynamics, holds too, namely that 
\begin{equation}
    U(t)=P(\bm \omega t ) e^{-iQ t} , \label{eq:quasienergydecomp}
\end{equation}
where $P(\bm \theta)$ is a piecewise-smooth unitary defined on $\mathbb{T}^m$ which satisfies $P(\bm 0)=\mathds{1}$ \footnote{There is a unitary degree of freedom in this decomposition, as generally one can choose $P(\bm 0)=P_0$ and replace Eq.~\eqref{eq:quasienergydecomp} with $U(t)=P(\bm \omega t)e^{-iQ t} P_0^\dagger$. We set $P_0=\mathds{1}$ for simplicity.}, and $Q$ is the generalized Floquet Hamiltonian with $d$ quasienergies and eigenstates, $Q\ket{\alpha}=q_\alpha \ket{\alpha}$.
Similar to the Floquet case, the generalized Floquet Hamiltonian $Q$ and unitary $P(\bm\theta)$ will not be unique, but this fact will be unimportant in our analysis. 
One may then construct  QEs, now defined as state-valued functions  on the torus:
    $$\ket{\alpha(\bm \theta)}=P(\bm \theta)\ket{\alpha}.$$
Such a decomposition would then entail that if we prepare our system in the initial state $\ket{\alpha}=\ket{\alpha({\bm 0})}$, the resulting dynamics is $m$-time-quasiperiodic up to a global phase:
$$\ket{\alpha(t)}=  U(t)\ket{\alpha}= e^{-iq_\alpha t}\ket{\alpha(\bm \theta=\bm \omega  t)}.$$ More generally, the time dependence of a generic initial state may then be decomposed as a linear combination over QEs: 
\begin{equation}
\label{eq:QEsdecomp}
\ket{\psi(t)}=\sum_{\alpha} c_\alpha e^{-iq_\alpha t}\ket{\alpha(\bm \omega t)},
\end{equation}
with $c_\alpha=\braket{\alpha}{\psi(0)}$, which can be understood as time-quasiperiodic over the torus $\mathbb{T}^n$, with $n\leq (m+d-1)$ (ignoring the global phase). The factor of $m$ comes from the physical driving frequencies $\bm{\omega}$, while there are $d$ additional frequencies coming from the winding phases $e^{-i q_\alpha t}$, minus a global phase.

As appealing as the generalized Floquet decomposition Eq.~\eqref{eq:quasienergydecomp} is, we stress its existence is nontrivial: it is known rigorously that this may not always hold in $(m>1)$-time-quasiperiodic systems  \cite{Jauslin1991,Blekher1992}.  This could come, for example, from topological obstructions in defining a smooth quasienergy state over the torus; see Ref.~\cite{Crowley2019}.
In other words, a generalized Floquet {\it theorem} (i.e., applying to all time-quasiperiodic Hamiltonians) does not hold, though the Floquet decomposition may still be valid in some cases. However, while interesting in its own right, the purpose of this work is not to investigate the conditions for when such a decomposition does or does not hold in time-quasiperiodic systems; rather, we assume that the systems in consideration always admit generalized QEs and study the compatibility of HSE and UE with such structure in dynamics. Note that the existence of QEs guarantees that the infinite-time averages in Eq.~\eqref{eq:kHSE_cond} [Eq.~\eqref{eq:UEDef}] always exist; i.e., the temporal ensemble of states or unitaries is well-defined in the limit $t \to \infty$ \footnote{This stems from the (continuous) Kronecker-Weyl theorem (Ref.~\cite{Beck2017}, p. 9), which states that the infinite-time average of any quasiperiodic function always exists, and it is equal to the average of the parent function over the torus,  $\E_{t\geq 0}[f(t)]=\E_{\bm \theta \in\mathbb{T}^m}[f(\bm \theta)]$}. 

\section{Quasienergy eigenstates limit complete quantum ergodicity}
\label{sec:no-go-theorems}

Having introduced the concepts of Hilbert-space ergodicity and unitary ergodicity, and the class of quantum dynamics (time-quasiperiodic systems) we   consider in this paper, we are now in a position to present our results. 
     Our first finding shows that the existence of QEs in time-periodic ($m=1$) systems precludes them from satisfying CHSE (and hence CUE). That is, Floquet systems cannot achieve full dynamical quantum ergodicity. 

\begin{theorem}
    \label{th:01}
If $H(t)$ is a time-periodic Hamiltonian with period $T$ and a bounded strength in the sense that ${B=\int_0^{T} \dd{t}\norm{H(t)}_\infty<\infty}$ \footnote{ $\norm{\, \cdot\,}_\infty$ is the usual operator norm, corresponding to the Schatten  $\infty$-norm.}, then $H(t)$ does not exhibit CHSE (and thus not CUE). \footnote{Note that Theorem~\ref{th:01} allows the evolution operator $U(t)$ to change discontinuously in time, thus also encompassing  discrete-time dynamics arising, for instance, from a brickwork circuit, in which case the Hamiltonian $H(t)$ is a sequence of  Dirac-$\delta$ pulses.}
\end{theorem}

The quantity $B$ should be understood as a measure of the ``physical resources'' needed to realize the dynamics: it is large for Hamiltonians whose strengths $\| H(t) \|_\infty$ are large or whose driving period is long.  Although $B$ changes upon the substitution $H(t)\rightarrow H(t)+c(t)\mathds{1}$, its minimum value over all $c(t)$ is 
 proportional to the time-integrated bandwidth $B=\tfrac{1}{2}\int_0^T \dd{t} (E_{\max}(t)-E_{\min}(t))$ \footnote{ $E_{\min}(t)$ ($E_{\max}(t)$) is the instantaneous ground (most-excited) state energy.}. As $B$ carries units of energy $\times$ time (recall $\hbar = 1$), it  has also the meaning of an ``action,'' which physically corresponds to the net effect that $H(t)$ has on the system during a single driving period. As we explain further below,  $B < \infty$ is simply the physical requirement of a ``quantum speed limit'': that the length of the trajectory traversed by the wavefunction over a period $T$ cannot be arbitrarily long.

From this point of view, the logic behind the proof of  Theorem~\ref{th:01} can be intuitively explained as an incompatibility of dynamics that traverses a finite ``distance'' to densely cover the continuous space that is the Hilbert space. Indeed, the formal proof proceeds by contradiction: 

\begin{proof}
Assume that the time-periodic Hamiltonian $H(t)$ satisfies CHSE. By  Floquet's theorem, $H(t)$ has a QE $\ket{\alpha(t)}=e^{-iq_\alpha t}\ket{\alpha(\theta=\omega t)}$,  where $\omega=2\pi/T$. Because phases are projected out in $\mathbb{P}(\mathbb{C}^d)$,  dynamics beginning from $\alpha(0)$ is time periodic: $\alpha(t)=\alpha(\theta=\omega t)$. We will reach a contradiction, in three steps.

First,  CHSE implies that the state $\alpha(t)$ uniformly visits  the $2(d-1)$-dimensional projective Hilbert space $\mathbb{P}(\mathbb{C}^d)$. This implies that the map $\theta\mapsto\alpha(\theta)$ is topologically dense, meaning that for any other state $\ket{\phi}$ and arbitrarily small $\varepsilon>0$ there is some angle $\theta$ for which $D(\alpha(\theta),\phi)<\varepsilon$, where
 \begin{equation}
 \label{eq:tracedistancedef}
     D(\psi,\phi)=\tfrac{1}{2}\norm{\psi-\phi}_1=\sqrt{1-\abs{\braket{\psi}{\phi}}^2}
 \end{equation} is the trace distance. This is rigorously proven in Appendix~\ref{app:ergimpliesdens}.

  Second, we appeal to the quantum speed limit $B < \infty$:  The state $\alpha(t)$ can only travel through a finite path in $\mathbb{P}(\mathbb{C}^d)$. Specifically, for any finite partition of the circle  $\theta_0\leq\theta_1\leq\cdots\leq \theta_n=\theta_0+2\pi$,  
    \begin{equation}\label{eq:totaldistandB}\sum_{j=1}^n D(\alpha(\theta_{j-1}),\alpha(\theta_{j}))\leq B. \end{equation}
This is a state-independent variant of the quantum speed limit, which is traditionally phrased in terms of the average energy \cite{Margolus1998} or variance \cite{Mandelstam1945} of a specific state, rather than the Hamiltonian norm \cite{Ng2011, Marvian2015}. Equation~\eqref{eq:totaldistandB} is a straightforward consequence of Schr\"odinger's equation (see Appendix~\ref{app:speed_limit}).
   
   In our final step, we note that the previous two observations are contradictory: By dimensionality arguments, we can find $n\sim \delta^{-2(d-1)}$ different states $\phi_1,\dots \phi_n$ pairwise separated by at least trace distance $\delta$, i.e., $D(\phi_i,\phi_j)\geq \delta$ for $i\neq j$. If the trajectory $\alpha(\theta)$ is dense, at some angles $\theta_1,\dots,\theta_n$ it must come $\varepsilon$-close to these states, $D(\alpha(\theta_i), \phi_i)\leq \varepsilon$. From Eq.~\eqref{eq:totaldistandB} and the triangle inequality we obtain $B\geq \delta^{-2(d-1)}(\delta-2\varepsilon)$, which can be made arbitrarily large by choosing small enough $\varepsilon$ and $\delta$, contradicting the finiteness of $B$. Full details are given in Proposition~\ref{prop:lowerboundonHamStrengthB2}.
   
   This shows the impossibility of CHSE. Lastly, because CHSE implies CUE, then CUE is also not achievable by time-periodic systems.
\end{proof}

In Appendix~\ref{app:periodicsystemsandkHSE} we show a stronger form of Theorem~\ref{th:01}. We prove that if a periodic Hamiltonian satisfies $k$\nobreakdash-HSE for some finite $k$, then $B\coloneqq \int_0^{T} \dd{t}\norm{H(t)}_\infty$ is lower bounded as \begin{align}
\label{eq:bound_on_B}
    B  \geq \max\Big \{\frac{\sqrt{2}}{3k}  &\Bigg(\binom{k+d-1}{k}-1\Bigg),\\ &\quad \frac{8}{(4 d)^d}\left(\frac{k}{\log (k+1)}\right)^{d-3/2}\Big\}. \nonumber
\end{align}
Informally, Eq.~\eqref{eq:bound_on_B} says that time-periodic $k$\nobreakdash-HSE is not achievable for large $k$ or $d$ unless the wavefunction travels for a very long distance within a single Floquet period $T$, in line with our physical intuition. For example, inserting $k=1$ in the first expression in the maximum, we see that $B\geq \frac{\sqrt{2}}{3}(d-1)$, where the linear growth with $d$ is required for a quasienergy eigenstate to come close to $d$ orthogonal states and achieve ${1}$\nobreakdash-HSE. In general, for fixed $k$, $B(d)$ has to grow at least as $d^k$. For fixed $d$, $B(k)$ has to grow at least like ${ (k/\log(k+1))^{d-{3}/{2}}}$, by the second expression in Eq.~\eqref{eq:bound_on_B}, which is obtained from  analyzing the geometrical distribution of a $k$-design in $\mathbb{P}(\mathbb{C}^d)$. In Sec.~\ref{sec:finiteHSEPeriodicsystems}, we provide explicit examples of time-periodic quantum systems with $B$  large enough such that $k$\nobreakdash-HSE is provably achievable.

Our next result is a generalization of Theorem~\ref{th:01} to $m$-time-quasiperiodic Hamiltonians, where we remind the reader our analysis is under the premise of the existence of QEs. 
Like in the Floquet case $(m=1)$, such QEs can lead to an obstruction of the system to achieve CHSE or CUE: Dynamics beginning from a QE is necessarily structured --- specifically time quasiperiodic, or in other words, amounts to winding around an $m$-dimensional torus $\mathbb{T}^m$. It may then be possible this regularity precludes an unbiased exploration of the Hilbert space.  
However, unlike the Floquet case, now there is an interplay between the number of tones $m$ of the drive (its ``complexity'') and the dimension $d$ of the ambient space: Such obstruction is  active only if the torus is small enough, such that the time-evolved state is unable to fully ``wrap'' around the projective Hilbert space.  Indeed, from a dimension-counting argument, we obtain the following.

\begin{theorem}
\label{th:02}

    Let $H(t)$ be a $m$-quasiperiodic Hamiltonian with a piecewise-smooth quasienergy eigenstate. Then $H(t)$ cannot exhibit CHSE if \begin{equation}
        m< 2(d-1).
    \end{equation} 
\end{theorem}
\begin{proof}
    The quasienergy eigenstate $\alpha(t)$ densely visits $\mathbb{P}(\mathbb C^d)$ in time (see Appendix~\ref{app:ergimpliesdens}). By the quasiperiodicity of the time evolution, $\alpha(t)=\alpha(\bm \theta=\bm \omega t)$, we deduce that the map $\bm \theta\mapsto \alpha(\bm \theta)$ is dense, from $\mathbb{T}^m$ to $ \mathbb{P}(\mathbb{C}^d)$. Because $\bm \theta\mapsto \alpha(\bm \theta)$ is piecewise continuous, this map must be  surjective, or entirely covering $ \mathbb{P}(\mathbb{C}^d)$. Intuition suggests that a surjective map from $\mathbb{T}^m$ to $ \mathbb{P}(\mathbb{C}^d)$ requires that the dimension of the codomain, $\text{dim}(\mathbb{P}(\mathbb{C}^d))=2(d-1)$ (the amount of real numbers required to specify a pure density matrix), is not greater than the dimension of the domain $m=\text{dim}(\mathbb{T}^{m})$. This intuition is correct, as long as the map $\bm \theta\mapsto \alpha(\bm \theta)$ is piecewise smooth in the torus, which is required in our definition of quasienergy eigenstate \footnote{The smoothness assumption is necessary, as there are examples of nonsmooth, but continuous and surjective maps which increase dimension (e.g. space-filling curves)}. The technical reason is that a piecewise-smooth map is piecewise Lipschitz continuous, and such maps do not increase Hausdorff dimension (see Proposition~1.7.19 of Ref.~\cite{BuragoBook2001}).  
Thus, CHSE requires $m\geq 2(d-1)$.
\end{proof}

The bound $m<O(d)$ where CHSE is impossible is obtained from the real dimension of $\mathbb{P}(\mathbb{C}^d)$. Similarly, the same idea can be applied  for the consideration of CUE, and we will obtain a bound which is $m<O(d^2)$, coming from the dimension of the projective unitary group.
\begin{theorem}
\label{th:03}
    Let $H(t)$ be an $m$-quasiperiodic Hamiltonian with a basis of piecewise-smooth quasienergy eigenstates. Then the evolution given by $H(t)$ cannot exhibit CUE if \begin{equation}
        m< d(d-1).
    \end{equation}  
\end{theorem}

 We provide the detailed proof in Appendix~\ref{app:proofTH3}. The idea is to note that the generalized Floquet decomposition for $U(t)=P(\bm \omega t ) e^{-iQ t}$ is quasiperiodic, with $m$ tones corresponding to $P(\bm \omega t )$, and (at most) an extra $d-1$ tones corresponding to the winding phases $e^{-iQ t}$, which then implies that
 $m+d-1=\dim(\mathbb{T}^{m+d-1})\geq  \dim(\operatorname{PU}(d))=d^2-1$.

 These three no-go theorems are depicted in Fig.~\ref{fig:resultssumm}.
\section{Implications of no complete quantum ergodicity}
\label{sec:numerics}
\begin{figure*}
    \centering
    \includegraphics[width=\textwidth ]{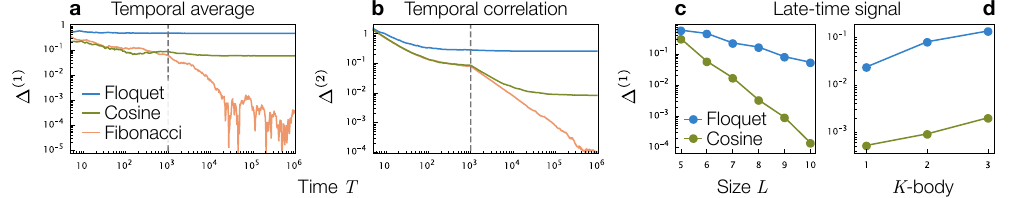}
    \caption{Difference between temporal and Haar moments of $K$-body optimized observable in a spin-$1/2$ chain driven by a Floquet protocol and  two-quasiperiodic drives with QEs (Cosine) and without QEs (Fibonacci). The observable is different for each driving protocol and moment. {\bf a} Difference between temporal average and Haar average $\Delta^{(1)}$ with $K=2$ and system size $L=6$. {\bf b} Difference between temporal and Haar second moments $\Delta^{(2)}$. The observables are selected to maximize the value at $T=10^3$, indicated by a vertical dashed line but are otherwise independent of time. {\bf c} Scaling of the late-time plateaus ($T=10^6$) of Floquet and Cosine drive with system size $L$ for the first moment and $K=2$. {\bf d} Same as {\bf c} but for fixed $L=9$ and varying $K$. The initial state is $\ket{\psi}=\ket{0}^{\otimes L}$ throughout.}
    \label{fig:03}
\end{figure*}
Our results in the previous section show that the evolution of any state under a few-tone quasiperiodic drive which allows for QEs necessarily has to be distinguishable  from a Haar-random state via some (potentially nonlocal) observable. One the one hand, this statement establishes a no-go theorem for CUE and CHSE as we already discussed, but on the other, it implies the existence of observables whose expectation values, temporal correlations, or higher statistical moments in time retain some memory of the initial state. 
For many-body quantum systems, the latter aspect presents us with an exciting possibility: Even at very late times --- when one expects an infinite-temperature, ``featureless'' {\it average state} due to the lack of energy conservation --- there nevertheless still remain nontrivial measurable features  which are different  from those coming from a genuinely featureless underlying {\it distribution}.

We test this idea numerically by focusing on local or few-body correlators which could be measured in a realistic setting. We consider a spin-1/2 chain of length $L$ and two classical Ising Hamiltonians along orthogonal directions,
\begin{align*}
H_0&=\sum_{j=1}^L X_j + \sum_{j=2}^L X_{j-1} X_j + \tfrac{1}{10} X_1 \\
H_1&=\sum_{j=1}^L Z_j + \sum_{j=2}^L Z_{j-1} Z_j + \tfrac{1}{10} Z_1,
\end{align*}
where $X_j$ and $Z_j$ are the Pauli operators acting on site $j$, and the boundary terms are introduced solely to break the spatial-reflection symmetry. 
We construct three driving protocols, consisting of certain alternating
kicks between the Ising Hamiltonians with varying amplitudes. First, we consider a Floquet drive by kicking with $H_0$ at even integer times and with $H_1$ at odd integer times:
\begin{align}
    H_{\textrm{Flo}}(t) = \sum_{n=1}^\infty \delta(2n-t) H_0 +\delta(2n-1 - t) H_1. 
\end{align}
Second, we consider a two-quasiperiodic drive  which we dub the ``Cosine drive'':
\begin{align}
    H_{\textrm{Cos}}(t) = \sum_{n=1}^\infty \delta(n-\omega_1 t)H_{g(\omega_2 t)}, 
\end{align}
where $g(\theta)=(1+\cos(\theta))/2$ and $H_{x} = (1-x)H_0 + x H_1$. We choose $\omega_1=1$ and $\omega_2=\pi (3-\sqrt{5})$ which are rationally independent. Both the Floquet and Cosine drives are expected to posses QEs~\cite{Feudel1995}.
Finally, we consider the two-quasiperiodic Fibonacci drive
\begin{align}
    H_{\textrm{Fib}}(t) = \sum_{n=1}^\infty \delta(n-\omega_1 t)H_{\chi(\omega_2 t)}, 
\end{align}
where $\chi(\theta)= 0$ if $\theta\in[0,2\pi - \omega_2)$ and $1$ if $\theta\in[2\pi - \omega_2,2\pi]$, which was shown to generically satisfy CUE in Ref.~\cite{Pilatowsky2023}. According to Theorems~\ref{th:01}~and~\ref{th:02}, the Floquet and Cosine drive cannot even exhibit CHSE if $L\geq 2$, and by the same results the Fibonacci drive does not admit QEs. We ask whether this difference has a measurable effect.

Our aim is to find a few-body observable whose late-time temporal moments are different from those of the Haar distribution. To this end, we consider a linear combination $O=\frac{1}{M}\sum_S J_S S$ of $K$-body Pauli observables
\begin{equation}
\label{K-body-ops}
    S=\sigma_{j_1}\cdots \sigma_{j_K},
\end{equation}
where each $\sigma_{j_l}\in\{X,Y,Z,\mathds{1}\}$ acts on a distinct site $j_l\in \{1,2,\dots,L\}$. We want to select $O$ as to maximize the difference between its temporal and Haar averages,
\begin{align}
    \Delta^{(1)}(T)&=\E_{0\leq t\leq T} \Big [\expval{O}{\psi(t)}\Big] - \E_{\phi\in \mathbb{P}(\mathbb{C}^d)}[\expval{O}{\phi} ] \nonumber \\& =\tr(O\rho_T^{(1)})-\tr(O\rho_{\textrm{Haar}}^{(1)}),
\end{align}
starting from $\ket{\psi(0)}=\ket{0}^{\otimes L}$, where $\rho_{T}^{(1)}=\E_{0\leq t< T}[\psi(t)]$ and $\rho_{\textrm{Haar}}^{(1)}=\mathds{1}/d$.  
 The maximum is achieved by computing the finite-time average 
$\rho_{T_\textrm{opt}}^{(1)}$ for a fixed $T_{\textrm{opt}}=10^3$~\footnote{Note that average is just a discrete sum because the drives consist purely of kicks at integer times.} and setting each coefficient to be 
  $J_S=\tr(S\rho_{T_\textrm{opt}}^{(1)}) - \text{tr}(S \rho_{\textrm{Haar}}^{(1)} )$~\footnote{We choose the normalization factor $M$ so to set the thermal fluctuations to unity $(\tr(O^2)-\tr(O)^2)/d=1$}. Note that the resulting observable $O$ is different for each driving protocol.

Figure~\ref{fig:03}a shows $\Delta^{(1)}(T)$ for the $2$\nobreakdash-body observable $O$ obtained by the procedure described above. We see that the temporal average remains distinguishable from the Haar average for both the Floquet and Cosine drives for times much beyond the optimization time $T_\textrm{opt}$ (vertical dashed line). In contrast, the corresponding quantity under the Fibonacci drive steadily decays toward the Haar average after the optimized time, as predicted by CHSE. 
This result shows that QEs in a many-body driven system leave a detectable signal at level of few-body expectation values, which is in accordance which our Theorems 1-3 ruling out CHSE is for such dynamics.

In Figure~\ref{fig:03}b we repeat a similar exercise now comparing the temporal correlations
$\E_{0\leq t\leq T}\Big[ \expval{S_1}{\psi(t)}\expval{S_2}{\psi(t)}\Big]$ against those from the Haar distribution $\!\E_{\phi\in \mathbb{P}(\mathbb{C}^d)}\!\Big[ \!\expval{S_1}{\phi}\!\!\expval{S_2}{\phi}\!\Big]$, with each $S_i$ of the form in Eq.~\eqref{K-body-ops}. These correlations can be written as $\tr(\rho^{(2)}_T S_1\otimes S_2)$ and $\tr(\rho^{(2)}_{\textrm{Haar}} S_1\otimes S_2)$, respectively, with $\rho^{(2)}_T=\!\E_{0\leq t< T}[\psi(t)^{\otimes 2}]$. Consequently, we consider a linear combination of correlators $O=\frac{1}{M}\sum_{S_1,S_2} J_{S_1,S_2} S_1\otimes S_2$, which maximizes the difference 
\begin{equation}
    \Delta^{(2)}(T)=\tr(O\rho_T^{(2)})-\tr(O\rho^{(2)}_{\textrm{Haar}})
\end{equation}
at $T_{\textrm{opt}}=10^3$ by picking $J_{S_1,S_2}=\text{tr}(\rho^{(2)}_{T_\textrm{opt}} S_1\otimes S_2) - \text{tr}(\rho_{\textrm{Haar}}^{(2)} S_1\otimes S_2)$ \footnote{Here, we normalize by selecting $M$ such that $(\text{tr}(O^2)-\text{tr}(O)^2)/d^2=1$}. 
Figure~\ref{fig:03}b shows that $\Delta^{(2)}$ for the Cosine and Floquet drives both eventually display a late-time plateau at a finite value, which indicates that there is a long-lived structure distinguishing them from Haar random which can be probed in the temporal correlations of these two-body observables. In contrast, $\Delta^{(2)}$ for the Fibonacci drive shows a steady decay toward the Haar moment.

We finally analyze the scaling of our results in terms of both the system size $L$ and size of the observable $K$. In Fig.~\ref{fig:03}c we show that the late-time signals for the Floquet and Cosine drive displayed in Fig.~\ref{fig:03}a decay exponentially in the system size $L$. This implies that a  measurement of these quantities in practical scenarios would be increasingly challenging, as the signal becomes exponentially weak. However, in Fig.~\ref{fig:03}d we also show an exponential improvement when increasing $K$. These $K$-body correlators can be experimentally probed by various techniques, including randomized-measurement approaches \cite{Elben2022}. It is an interesting future direction of this work to explore if the interplay between $K$ and $L$ could allow for a viable experimental procedure to measure the difference between the temporal and Haar moments of driven systems which violate CHSE.

\section{Many driving frequencies permit complete quantum ergodicity}
\label{sec:CUEModel}
In Sec.~\ref{sec:no-go-theorems}, we identified constraints on a $m$-time-quasiperiodic Hamiltonian's ability to uniformly cover either the Hilbert space (Theorem 2) or unitary space  (Theorem 3), under the assumption of existence of QEs. 
They tell us that a quantum system driven with too few tones cannot exhibit dynamical ergodicity: Namely, if $m <  O(d)$, CHSE is impossible; while if $m < O(d^2)$, CUE is impossible. 
Physically, this is sensible, as when the number of driving frequencies $m$ is small, dynamics will not be ``complex'' enough. However, this leaves open the obvious converse question: suppose $m$ is large enough. Then are there time-quasiperiodic systems that {\it do} exhibit CHSE or CUE?

In this section, we will answer this in the affirmative. We show how to construct explicit $m$-quasiperiodic Hamiltonians with $m=d^2-2$ tones that host QEs, and which provably satisfy CUE (and thus CHSE).
Together with the no-go theorems of the previous section, this leads us to the ``phase diagram'' depicted in Fig.~\ref{fig:resultssumm}.

\subsection{Single-qubit complete unitary ergodicity with $m=2$ driving frequencies}

We start with the case for a single qubit, with $m=2$, which will motivate the generalization for systems of arbitrary dimension.

Our key idea is to construct states $\ket{\alpha(\bm \theta)}$ ($\alpha=0,1$), parametrized by $\bm \theta=(\theta_1,\theta_2)$, that satisfy the CHSE condition, and then reverse engineer a Hamiltonian which has these states as its quasienergy eigenstates.  By imposing Eq.~\eqref{eq:kHSE_cond} on the states $\ket{\alpha(\bm \theta)}$, for all $k$, the resulting Hamiltonian will satisfy CHSE, but further CUE, which will motivate the generalization to $d>2$.

The CHSE condition requires the state $\alpha(t)=\dyad{\alpha(t)}$ to uniformly cover the Bloch sphere  $\mathbb{P}(\mathbb{C}^2)$. Because $\alpha(t)=\alpha(\bm \theta=\bm \omega t)$, we can achieve this by selecting $\bm \omega=(\omega_1,\omega_2)$ to be rationally independent, and ${\alpha(\bm \theta)}$ to be uniformly distributed  on  $\mathbb{P}(\mathbb{C}^2)$, as a function of the angles $\bm \theta$ on the torus $\mathbb{T}^2$.

First, we construct the state $\ket{0({\bm \theta})}$,  parametrized as
$$\ket{0(\bm \theta)}=\sqrt{p(\theta_1)}\ket{0}+\sqrt{1-p(\theta_1)}e^{-i\theta_2}\ket{1}.$$
To uniformly cover the Bloch sphere, the function $p(\theta_1)$ needs to be uniformly distributed in $[0,1]$ when $\theta_1$ is uniformly distributed in $[0,2\pi)$. This is achieved by any surjective function such that $\abs*{\tfrac{\dd p}{\dd \theta_1}}$ is almost-everywhere constant. Here, we consider $p(\theta_1)=\abs{1-\theta_1/\pi},$ which is  continuous on the circle.  The resulting map is depicted in Fig.~\ref{fig:4}. 
\begin{figure}[t]
    \centering
    \includegraphics[width=\columnwidth]{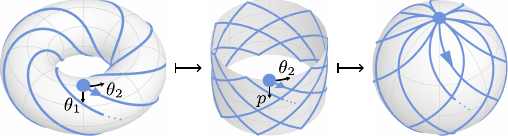}
    \caption{Transformation ${\bm \theta}\mapsto \ket{0(\bm \theta)}$. A pair of angles $\bm\theta=(\theta_1,\theta_2)$ in the torus (left) is mapped to $(p(\theta_1),\theta_2)$ in the cylinder (middle), which is further mapped to the state $\ket{0(\bm \theta)}=\sqrt{p}\ket{0}+e^{-i\theta_2}\sqrt{1-p}\ket{1}$ in the Bloch sphere (right). The blue line displays time evolution $\bm \theta=\bm \omega t$, with a blue disk marking $t=0$.}
    \label{fig:4}
\end{figure}

Having defined $\ket{0(\bm \theta)}$, we set $\ket{1(\bm \theta)}$ to be the orthogonal state
$$\ket{1(\bm \theta)}=\sqrt{p(\theta_1)}\ket{1}-\sqrt{1-p(\theta_1)}e^{i\theta_2}\ket{0}.$$ We now use these two states to construct a quasiperiodic Hamiltonian which has them as QEs.
We can write these two states as the columns of a unitary $P(\bm \theta)$, (i.e., $\ket{\alpha(\bm \theta)}=P(\bm \theta)\ket{\alpha}$), where
\begin{equation}
    P(\theta_1,\theta_2)=\begin{pmatrix}
        \cos  \xi(\theta_1) &&- \sin \xi(\theta_1) e^{i\theta_2}\\
         \sin \xi(\theta_1) e^{-i\theta_2}&& \cos\xi(\theta_1)
    \end{pmatrix},
\end{equation}
and $\xi(\theta_1)=\arccos\sqrt{p(\theta_1)}$. Then we choose any rationally independent driving frequencies $\bm \omega=(\omega_1,\omega_2)$ and quasienergy $q$ to {\it define} the evolution operator to be given by the generalized Floquet decomposition, Eq.~\eqref{eq:quasienergydecomp}, substituting $Q=\text{diag}(-q,q)$, $\theta_1=\omega_1t$, and $\theta_2=\omega_2t$,
\begin{equation}
\label{eq:evolutionopqubit}
    U(t)=\begin{pmatrix}
        \cos  \xi(\omega_1t)e^{iqt} &&- \sin \xi(\omega_1t) e^{i(\omega_2-q)t}\\
         \sin \xi(\omega_1t) e^{-i(\omega_2-q)t}&& \cos\xi(\omega_1t)e^{-iqt}
\end{pmatrix}.
\end{equation}
Finally, we can obtain the two-quasiperiodic Hamiltonian by the Schrödinger equation
$H(t) = i(\partial_t U(t))U(t)^{\dagger}.$ 

It turns out that the evolution given by Eq.~\eqref{eq:evolutionopqubit} not only satisfies CHSE, but the stronger CUE. This is because the transformation
\begin{equation}
    (\xi,\eta,\varphi)\mapsto \begin{pmatrix}
        \cos  \xi\,e^{i\eta } &&- \sin \xi\,e^{i\varphi}\\
         \sin \xi\,e^{-i\varphi}&& \cos\xi\,e^{-i\eta}
         \end{pmatrix}
\end{equation}
 is precisely the Euler-angle parametrization of the group $\operatorname{SU}(2)$, and furthermore the assignment $\xi=\arccos\sqrt{p(\theta_1)}$ makes it measure preserving, i.e. maps the Haar measure of the torus $\mathbb{T}^3\ni (\theta_1,\eta,\varphi)$ to the Haar measure of $\operatorname{SU}(2)$. Thus, upon substituting
\begin{align}\label{eq:SU(2)Euleragnleassignment}
    \xi=\arccos\sqrt{\abs{1-\omega_1 t/\pi}},&&\eta= qt, && \varphi=(\omega_2-q)t,
\end{align}
we guarantee that $U(t)$, in time, explores $\operatorname{SU}(2)$ uniformly. 

In the next section, we explain how to generalize this construction to $\operatorname{SU}(d)$, to obtain a $d$-dimensional quasiperiodic Hamiltonian which has QEs and satisfies CUE. That is, the time evolution operator uniformly explores the entire $\operatorname{SU}(d)$ space (the projective unitary space acting on a qudit of dimension $d$) over time.

\subsection{Qudit complete unitary ergodicity with $m=d^2-2$ driving frequencies}
By considering a specific sequence of rotations of the form 
$$R_j(\xi,\varphi,\eta)=\begin{pmatrix}
        \mathds{1}_{j-1} & & & \\
    & \cos \xi e^{i \eta} &-\sin \xi e^{i \varphi} & \\
    & \sin \xi e^{-i \varphi} & \cos \xi e^{-i \eta} & \\
    & & & \mathds{1}_{d-j-1}
    \end{pmatrix},$$
one can construct Hurwitz's parametrization of $\operatorname{SU}(d)$, in terms of $d^2-1$ Euler angles  \cite{Hurwitz1897,Zyczkowski1994,Diaconis2017}. We utilize this parametrization to construct an $m$-quasiperiodic drive which satisfies CUE and has QEs, with $m=d^2-2$. This is done by explicitly defining the evolution operator $U(t)$ in the generalized Floquet decomposition form [Eq.~\eqref{eq:quasienergydecomp}]. By assigning each Euler angle to a function of the driving frequencies and the quasienergies, we guarantee that $U(t)$ uniformly explores $\operatorname{SU}(d)$ in time. The assignment for the Euler angles is a generalization of Eq.~\eqref{eq:SU(2)Euleragnleassignment}, where the $d^2-1$ Euler angles are written in terms of the $m=d^2-2$ driving angles $\omega_1t,\dots,\omega_mt$, and one of the quasienergies. The details of this construction are left to  Appendix~\ref{app:CUEdrive}. 

 The driving frequencies and quasienergies can be selected so that the corresponding winding in the torus is {\it equidistributed} \cite{Ramamoorthy2008}, and consequently the trace distance between the finite-time temporal moments and the corresponding Haar moments decay like $1/T$. This power law is quadratically faster than the $1/\sqrt{T}$ decay one gets from independent random sampling, so this construction might be useful for producing quasirandom states or unitaries for quasi-Monte Carlo integration \cite{Morokoff1995}.
 
In our construction, only one quasienergy is related to one of the Euler angles, and the remaining quasienergy degrees of freedom are just averaged out in time. We leave as an open question if it is possible to utilize all the $d-1$  quasienergy degrees of freedom. If the answer is positive, this would decrease the required number of driving angles to $m= d^2-d$, saturating the bound given by Theorem~\ref{th:03} and removing the white sliver in the phase diagram in Fig.~\ref{fig:resultssumm}. If the answer is negative, then the bound in Theorem~\ref{th:03} could potentially be strengthened.

\section{Quantum $k$-ergodicity in time-periodic systems}
\label{sec:finiteHSEPeriodicsystems}
Theorems~\ref{th:01}, \ref{th:02}, and \ref{th:03} show that the existence of quasienergy eigenstates forbids the achievability of the most stringent forms of dynamical ergodicity: CHSE and CUE. It is natural to ask if there are similar obstructions to quantum ergodicity if one relaxes to finite moments, as in the notions of $k$\nobreakdash-HSE and $k$\nobreakdash-UE, introduced in Sec.~\ref{sec:CHSE}. Surprisingly, we show here that $k$\nobreakdash-HSE and $k$\nobreakdash-UE can be reached even by time-periodic Hamiltonians, corresponding to the minimal $m=1$ time-periodic or Floquet case.

The achievability of finite $k$\nobreakdash-UE in time-periodic systems can be understood from the existence of finite $k$-unitary designs in quantum information theory \cite{Seymour1984,Scott2008} --- an ensemble of a \emph{finite} number of unitaries which reproduces the Haar measure up to the $k$th statistical moment (see Appendix \ref{app:designs} for more details). 

Utilizing the fact that finite unitary $k$-designs exist, we may construct a periodic sequence of rotations which satisfies $k$\nobreakdash-UE, over discrete time. The construction proceeds as follows: For any $k$, let $\mathcal{D}_k=\{V_0, V_1,\dots,V_{n-1}\}$ be a finite unitary $k$-design with $n$ elements, which can be selected so that $V_0=\mathds{1}$ by otherwise applying $V_0^\dagger$ to all of its elements. We define a periodic drive by applying a sequence of gates such that the evolution operator cycles through $\mathcal{D}_k$. 

At time every integer time $t= j \mod n$, we apply the unitary $V_{j}V_{j-1}^\dagger$. Then, the evolution operator satisfies $U(t)=V_{j}$. In this case, the integral in the left-hand side of Eq.~\eqref{eq:UEDef} which defines  $k$\nobreakdash-UE, can be rewritten in terms of the series 
$$\lim_{N\to \infty}\frac{1}{N+1}\sum_{t=0}^N U(t)^{\otimes k,k} =\E_{\mathrm{Haar}}[{W^{\otimes k,k}}].$$
Note that this evolution has period $T=n$, and it can achieve  the $k$\nobreakdash-UE condition with $B={\int_0^T dt \|H(t)\|_\infty \le n \pi}$, where the time-periodic Hamiltonian $H(t)$ consists of a sequence of infinite-strength kicks $H_j(t)=i\delta(t-j)\mathrm{log}(V_{j}V_{j-1}^\dagger)$ which satisfy $\int dt \|H_j(t)\|_\infty\le \pi $. This is consistent with the bound on $B$ given by Eq.~\eqref{eq:bound_on_B}, as $n$ has to be sufficiently large in order for $\mathcal{D}_k$ to form a unitary $k$-design.

In the construction above $k$\nobreakdash-UE is achieved by a periodic sequence of gates, in discrete time. The Hamiltonian discontinuously drives the state around $\mathbb{P}(\mathbb{C}^d)$.  It is, however, interesting to ask if the same level of ergodicity can be achieved when the evolution is continuous, or even smooth. In what is left of this section, we present some examples to show that the answer is positive.

We provide examples of continuous time-periodic systems that satisfy $k$\nobreakdash-HSE and $k$\nobreakdash-UE. We start with a qubit, $d=2$. In this case, $k$\nobreakdash-UE is completely characterized by the time trajectory of a single quasienergy eigenstate since the trajectory of the remaining state is determined by their orthogonality. In a single-qubit Hamiltonian $H(t)$, if one quasienergy eigenstate $\ket{\alpha(\theta)}$ ($\alpha=0$ or $\alpha=1$) satisfies the $k$\nobreakdash-HSE condition \begin{equation}\label{eq:condkHSEeigenstate}
    \E_{\theta\in\mathbb{T}}[\alpha(\theta)^{\otimes k}]=\E_{\phi\in\mathbb{P}(\mathbb{C}^2)}[\phi^{\otimes k}],
\end{equation} and the corresponding quasienergy and driving frequency are rationally independent, then $H(t)$ satisfies $k$\nobreakdash-UE (see Corollary~\ref{cor:singleeigenstateimpliesUE}). Thus, constructing a single-qubit time-periodic drive which satisfies $k$\nobreakdash-UE reduces to designing a closed curve $\alpha(\theta)$ in $\mathbb{P}(\mathbb{C}^2)$ which satisfies Eq.~\eqref{eq:condkHSEeigenstate}, from which one can construct the evolution operator by the Floquet decomposition $U(t)=P(\omega t)e^{-i \text{diag}(-q,q)}$ with $P(\theta)=\sum_{\alpha=0}^1 \dyad{\alpha(\theta)}{\alpha}$, where $q,\omega$ are chosen to be rationally independent.

\begin{figure}
    \centering
    \includegraphics{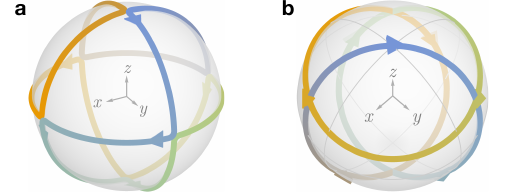}
    \caption{Single-qubit quasienergy eigenstate of a time-periodic drive that satisfies $k$-UE, for {\bf a} $k=3$ and {\bf b} $k=2$. In {\bf a}, the right-angle corners form  six-state $3$-design,  and, in {\bf b}, arrows mark a seven-state $2$-design.}
    \label{fig:5}
\end{figure} 
We use two approaches to find curves $0(\theta)\in \mathbb{P}(\mathbb{C}^d)$ that satisfy Eq.~\eqref{eq:condkHSEeigenstate}. In Fig.~\ref{fig:5}a, we show a continuous curve constructed to interpolate through the six-state $3$-design $\{\ket{0},\ket{1}, (\ket{0}\pm \ket{1})/\sqrt{2},(\ket{0}\pm i\ket{1})/\sqrt{2}\}$, via the great circles of the Bloch sphere. Equation~\eqref{eq:condkHSEeigenstate} is easily shown to hold for $k=3$ by explicit integration. This curve is not entirely differentiable, resulting in a Hamiltonian with discontinuous time dependence which satisfies $3$-UE. Alternatively, the curve shown Fig.~\ref{fig:5}b is obtained by solving Eq.~\eqref{eq:condkHSEeigenstate} for $k=2$ in Fourier space (see Appendix~\ref{app:Fourier}), which yields an analytic curve that turns out to interpolate through a seven-state $2$-design. The corresponding Hamiltonian has analytic time dependence, but satisfies only $2$-UE. The values of $B=\int_0^T \dd t\norm{H(t)}_\infty$ for the drives shown in Figs.~\ref{fig:5}a and b satisfy $B\geq3\pi$ and $B\geq 8.296$, respectively \footnote{These lower bounds are given by the trace distance traversed by the trajectories in Figs.~\ref{fig:5}a and b (see Corollary \ref{cor:dis_traveled_bound}).}. 

In Appendix~\ref{app:Fourier} we construct a time-periodic analytic Hamiltonian which satisfies $1$-HSE in any dimension. This is done again by going into Fourier space.  We believe the Fourier approach may generalize to arbitrary $k$, and might allow to explicitly find  time-periodic Hamiltonians which satisfy $k$\nobreakdash-HSE, or even $k$\nobreakdash-UE, where the time dependence is smooth, although more analytical understanding is required in this direction, which we leave open (see Appendix~\ref{app:Fourier} for more details).

\section{Summary and discussion}
\label{sec:discussion}
In this work, we have introduced and studied novel dynamical notions of quantum ergodicity defined in terms of statistical similarities of the temporal ensemble of states or unitaries to their respective uniform spatial ensembles. These are dubbed Hilbert-space ergodicity (HSE) and unitary ergodicity (UE),  and define a hierarchical tower of quantum ergodicities based on equivalence at different levels of moments $k$. 
In the limit of $k\to \infty$, we obtain complete Hilbert-space ergodicity (CHSE) and complete unitary ergodicity (CUE), in which the temporal distribution of initial states and  evolution operators, respectively, are exactly equal to the respective uniform Haar distribution. 
We studied the achievability of HSE and UE in the class of quasiperiodically driven systems driven by $m$ fundamental tones assuming the existence of quasienergy eigenstates, and proved that CHSE and CUE are not achievable in Floquet systems, as well as in quasiperiodically driven $d$-dimensional systems if $m< 2(d-1)$ and $m< d(d-1)$, respectively. Conversely, we provided examples of drives satisfying CUE (and hence CHSE) with $m=d^2-2$. We finally showed that a more relaxed form of quantum ergodicity, $k$\nobreakdash-HSE and $k$\nobreakdash-UE for some fixed $k$, can be achieved even by Floquet systems with driving periods that are long enough. 

Besides representing an important step toward a   unifying notion of quantum ergodicity and chaos applicable across different classes of quantum dynamics, 
our work has  several conceptual and technical implications. For one, our dynamical notions of ergodicity provide a   framework to understand the emergence of thermalization in extended driven systems, without reference to eigenstates like in the eigenstate thermalization hypothesis. For example, a system exhibiting $1$-HSE is such that the infinite-time average of any observable is equal to its expectation value at infinite temperature. Moreover, the higher levels of $k$\nobreakdash-HSE and $k$\nobreakdash-UE imply that the system at almost all times is locally maximally mixed, and furthermore the ensemble of pure quantum states which make up a local subsystem itself forms a quantum state $k'$-design, for some moment $k'$ related to $k$ \cite{Cotler2023,Pilatowsky2023}, a recently uncovered stronger form of quantum thermalization called ``deep thermalization'' \cite{Claeys2022,Wilming2022,Ho2022,Ippoliti2022,Choi2023,Cotler2023,Ippoliti2023,Lucas2023,Shrotriya2023}.

Our  results also provide an avenue to partially answer the open question whether a quasiperiodically driven system exhibits quasienergy eigenstates (QEs) or not, which is the mathematical question of reducibility of quantum dynamics. Physically, it corresponds to the question of localization versus delocalization of a driven system when mapped to the so-called extended Hilbert space \cite{Sambe1973,Ho1983} (frequently referred to as the ``frequency lattice"). 
As we have seen, the presence of QEs is incompatible with CHSE and CUE in large-dimensional systems, and so a demonstration of CHSE or CUE would preclude the existence of QEs within a given model. For instance, in Ref.~\cite{Pilatowsky2023}, it was shown that the family of $(m=2)$\nobreakdash-quasiperiodically driven systems called Fibonacci drives provably satisfies CUE in any dimension. This result, compounded with our Theorem~\ref{th:03}, implies that these  drives cannot be reducible, a nontrivial mathematical statement, and further suggests the computational complexity required to describe such a system grows unboundedly with time, owing to the lack of regular structure of quantum dynamics. 

There are several open questions arising from our work. First, our work relates two notions of complexity of quantum dynamics: (i) the number of driving frequencies $m$ underlying a driven Hamiltonian, and (ii) the degree of  ergodicity exhibited by dynamics, captured by the moment $k$ in HSE or UE. An immediate interesting question is the connection of these notions of complexity to other existing notions, such as the Krylov \cite{Parker2019} or circuit \cite{Brown2018,Haferkamp2022} complexities of quantum dynamics. These have been recently studied in periodically driven systems~\cite{Giancarlo2020,Suchsland2023}. 
Second, the question of typicality deserves to be addressed: While we have provided explicit constructions of quasiperiodically driven Hamiltonians provably exhibiting HSE and UE,  is such ergodicity expected to hold more in {\it generic} quasiperiodically driven  systems?
Relatedly, beginning from a system that does exhibit $k$\nobreakdash-HSE and $k$\nobreakdash-UE, are these properties robust against noise and perturbations to the driving Hamiltonian, i.e., can we  define universality classes of ergodic behavior? 
We leave the exploration of such interesting questions to future work.

\begin{acknowledgments}
We thank C.~Dag, D.~Mark, D.~Long, and A.~Chandran for insightful conversations. This work is partly supported by the Center for Ultracold Atoms (an NSF Physics Frontiers Center) PHY-1734011, NSF CAREER (DMR-2237244), NSF STAQ (PHY-1818914), the DARPA ONISQ program (W911NF201002), NSF PHY-2046195, and NSF QLCI grant OMA-2120757. W.~W.~H.~is supported by the NRF Fellowship, NRF-NRFF15-2023-0008, and the CQT Bridging Fund. 
\end{acknowledgments}

\appendix
  \newtheorem{apptheorem}{Theorem}
    \newtheorem{corollary}{Corollary}
    \newtheorem{lemma}{Lemma}
    \newtheorem{prop}{Proposition}
    \newtheorem{appdefinition}{Definition}

\let\oldsection\section
\renewcommand{\section}[1]{\oldsection{#1}\setcounter{apptheorem}{0}
\renewcommand\theHapptheorem{\theapptheorem}%hyperref
\setcounter{corollary}{0}
\renewcommand\theHcorollary{\thecorollary}%hyperref
\setcounter{lemma}{0}
\renewcommand\theHlemma{\thelemma}%hyperref
\setcounter{prop}{0}
\renewcommand\theHprop{\theprop}%hyperref
\setcounter{appdefinition}{0}
\renewcommand\theHappdefinition{\theappdefinition}%hyperref
}
\renewcommand{\theapptheorem}{\Alph{section}\arabic{apptheorem}}
\renewcommand{\thecorollary}{\Alph{section}\arabic{corollary}}
\renewcommand{\thelemma}{\Alph{section}\arabic{lemma}}
\renewcommand{\theprop}{\Alph{section}\arabic{prop}}
\renewcommand{\theappdefinition}{\Alph{section}\arabic{appdefinition}}

\section{Quantum ergodicity by design}
In this appendix, we introduce the notions of state and unitary $k$-designs from quantum information  theory. We state some of their properties and utilize them to prove the relations between the different levels of HSE and UE described in the main text.

We begin with the notion of state $k$-design, which underpins HSE.
\label{app:designs}
\begin{appdefinition}[State $k$-design]
    A probability measure $\mu$ over $\mathbb{P}{(\mathbb{C}}^d)$ is a (state) $k$-design if
    \begin{equation}
        \E_{\psi \sim \mu}[\psi^{\otimes k}]=\E_{\phi\in \mathbb{P}{(\mathbb{C}}^d)}[\phi^{\otimes k}].
    \end{equation}
\end{appdefinition}
\noindent The right-hand side is to be understood as the  expectation value with respect to the invariant measure induced by the Haar measure of the unitary group. It can be calculated explicitly using Schur’s lemma of representation theory, \begin{equation}
\label{eq:rho_Haar_symm_proj}
\rho_{\text{Haar}}^{(k)}\coloneqq\E_{\phi\in\mathbb{P}(\mathbb{C}^d)}[\phi^{\otimes{k}}]=\Pi_{\mathrm{sym}}^{(k)}\frac{(d-1)!\, k!}{(d+k-1)!},
\end{equation}  where $\Pi_{\mathrm{sym}}^{(k)}$ is the orthogonal projector into  the symmetric subspace of $(\mathbb{C}^d)^{\otimes k}$, obtained by averaging the operators $V_\pi$ which permute the $k$ tensors, $V_\pi \bigotimes_{j=1}^k \ket{\psi_j}  =\bigotimes_{j=1}^k \ket*{\psi_{\pi(j)}}$, over all permutations $\pi$ in the symmetric group of $k$ elements $\mathcal{S}_k$  \cite{Harrow2013}, \begin{equation}
\label{eq:symm_proj}
    \Pi_{\mathrm{sym}}^{(k)}=\frac{1}{k!}\sum_{\pi \in \mathcal{S}_k} V_\pi.
\end{equation} 

For HSE, we are interested in the case where $\mu$ is the state temporal ensemble, in continuous time,  $\mu_{\text{time}}=\lim_{T\to \infty}\frac{1}{T}\int_0^T\dd{t} \delta_{\psi(t)}$, with $\delta_{\psi(t)}$ the Dirac measure centered at $\psi(t)$. Simply, $k$\nobreakdash-HSE is  the statement that $\mu_{\text{time}}$ forms a $k$-design. 

Note that the assumption that the limit ${T\to \infty}$ exists is implicit in the definition of  $k$\nobreakdash-HSE. There are examples of dynamics where this average may fail to converge. Nevertheless, if the Hamiltonian is quasiperiodic and has quasienergy eigenstates, then $\mu_{\textrm{time}}$ is guaranteed to exist. This is because $\ket{\psi(t)}$, when expanded in the quasienergy-eigenstate basis, is seen to be $n$-quasiperiodic, for some integer $n$ that depends on the rational dependence of the quasienergy and driving frequencies, and then by the Kronecker-Weyl theorem, $\mu_{\textrm{time}}=(2\pi)^{-n}\int_{\mathbb{T}^n}\dd{\bm \theta} \delta_{\psi(\bm \theta)}$.

One simple way to verify if a probability measure $\mu$ forms a $k$-design is via the so-called frame potential
\begin{equation}
    \mathcal{F}^{(k)}_{\mu}\coloneqq \E_{\psi,\phi\sim \mu}[\abs{\braket{\psi}{\phi}}^{2k}]. \label{eq:state_frame_pot}
\end{equation}
It can be shown that $\mu$ forms a $k$-design if and only if $\mathcal{F}_{\mu}^{(k)}=\mathcal{F}_{\text{Haar}}^{(k)}=\frac{(d-1)!\, k!}{(d+k-1)!}$ (see Proposition~38 in Ref.~\cite{Mele2023}). In the particular case of the temporal ensemble $\mu_{\text{time}}$ with initial state $\psi(0)$, the frame potential is given by
$$\mathcal{F}_{\text{time}}^{(k)}= \E_{t,t'\geq 0}[\abs{\braket{\psi(t')}{\psi(t)}}^{2k}],
$$
which is equal to the Haar frame potential if and only if the system satisfies $k$\nobreakdash-HSE.

Now we introduce unitary $k$-designs, which provide the framework of UE.
\begin{appdefinition}[Unitary $k$-design]
    A probability measure $\nu$ over $\operatorname{U}(d)$ is a unitary $k$-design 
    \begin{equation}
        \E_{U\sim \nu}[U^{\otimes k,k}]=\E_{\mathrm{Haar}}[V^{\otimes k,k}].
    \end{equation}
\end{appdefinition}
\noindent The right-hand side denotes average over the Haar measure of $\operatorname{PU}(d)$, which can be constructed by sampling Haar $V$ from $\operatorname{U}(d)$ or $\operatorname{SU}(d)$, and projecting into $\operatorname{PU}(d)$ by taking the tensor product $V^*\otimes V$.

Unitary $k$-ergodicity ($k$\nobreakdash-UE) is the statement that the unitary operator temporal ensemble $\nu_{\text{time}}=\lim_{T\to \infty}\frac{1}{T}\int_0^T\dd{t} \delta_{U(t)}$ forms a unitary $k$-design. As before, this ensemble is guaranteed to converge under a quasiperiodic Hamiltonian with quasienergy eigenstates.

A probability measure $\nu$ is a unitary $k$-design if the frame potential 
$$\mathcal{F}^{(k)}_{\nu}\coloneqq \E_{U,V \sim \nu}[\abs{\tr(U^\dagger V)}^{2k}]$$ is equal to the Haar frame potential (Lemma 33 in Ref.~\cite{Mele2023})
$$\mathcal{F}^{(k)}_{\text{Haar}}\coloneqq \E_{W,V\sim \text{Haar}}[\abs{\tr(W^\dagger V)}^{2k}]=\E_{V\sim \text{Haar}}[\abs{\tr(V)}^{2k}].$$ The unitary frame potential for the temporal ensemble is given by 
$\mathcal{F}^{(k)}_{\text{time}}= \E_{t\geq 0,t'\geq 0}[\abs{\tr(U(t')^\dagger U(t))}^{2k}]$, which is equal to the unitary Haar frame potential if and only if the system satisfies $k$\nobreakdash-UE.

There are two basic properties of designs, which we state below, which allow us to prove the relations between the different levels of the hierarchies of quantum ergodicity. 

\begin{prop}[{$k$-designs are $k'$-designs if $k'\leq k$}] 
\label{prop:kdesignsarekm1designs}
Let $\mu$ be a state (unitary) $k$-design. Then $\mu$ is a state (unitary) $k'$-design for all $k'\leq k$ (Observation 29 in Ref.~\cite{Mele2023}).
\end{prop}
\begin{prop}[A unitary $k$-design acted on a state forms a state $k$-design]
\label{prop:unitarydesignsgeneratestatedesigns}
    Let $\nu$ be a unitary $k$-design. For a fixed state $\ket{\psi}$ let $\nu_\psi$ be  the probability distribution on $\mathbb{P}(\mathbb{C}^d)$ that results from applying a $\nu$-distributed unitary to  $\ket{\psi}$.
     Then $\nu_\psi$ is a state $k$-design (Ref.~\cite{Mele2023}, p.~25).
\end{prop}

\begin{corollary}[Arrows in Fig.~\ref{fig:hier}]\label{cor:relationsHSEandUE}
In any time-dependent system, the following implications hold.
\begin{enumerate}
    \item[{\normalfont(a)}] $\forall k\geq k'\colon\,\, k$-HSE $\implies$ $k'$-HSE,
    \item[{\normalfont(b)}] $\forall k\geq k'\colon \,\,k$-UE $\,\,\,\implies$ $k'$-UE,
    \item[{\normalfont(c)}] $\forall k\in\mathbb{N}\,\colon\,\, k$-UE $\,\,\,\implies$ $k$\nobreakdash-HSE,
        \item[{\normalfont(d)}] CUE $\,\,\,\implies$ CHSE.
\end{enumerate}
\end{corollary}

\begin{proof}
Properties~(a) and (b) follow from Proposition~\ref{prop:kdesignsarekm1designs}, applied to the state and unitary operator temporal ensembles. Property (c) follows directly from Proposition~\ref{prop:unitarydesignsgeneratestatedesigns}. Property (d) is an immediate consequence of property (c).
\end{proof}

Corollary~\ref{cor:relationsHSEandUE} tells us that $k$\nobreakdash-UE (CUE) is a stronger property than $k$\nobreakdash-HSE (CHSE). It is natural to ask if it is strictly stronger. In the particular case of a qubit, $k$\nobreakdash-HSE and $k$\nobreakdash-UE are equivalent. The reason is the following property of $k$-designs in qubits, which is a converse for Proposition~\ref{prop:unitarydesignsgeneratestatedesigns}.

\begin{apptheorem}
\label{th:qubitallstatekdesignisunitarydesign}
    Let $\nu$ be a probability measure on $ \operatorname{U}(2)$ such that for any state $\psi\in \mathbb{P}(\mathbb{C}^2)$, the state distribution $\nu_\psi$ on $\mathbb{P}(\mathbb{C}^2)$ that results from applying a $\nu$-distributed unitary to  $\psi$ forms a state $k$-design. Then $\nu$ is a unitary $k$-design.  
\end{apptheorem}

Taking $\nu$ to be the unitary temporal ensemble we immediately deduce the following.
\begin{corollary}
\label{cor:qubitHSEimpliesUE}
    In a qubit ($d=2$),  $k$\nobreakdash-HSE (CHSE) is equivalent to $k$\nobreakdash-UE (CUE).
\end{corollary}
\noindent The proof of Theorem~\ref{th:qubitallstatekdesignisunitarydesign} relies on the representation theory of $\mathrm{SU}(2)$. One can understand the central argument physically, in terms of spin addition: Adding $2k$ spin\nobreakdash-$1/2$ particles generates the same total spin subspaces as adding two spin\nobreakdash-$k/2$ particles (ignoring multiplicities).
\begin{proof}[Proof of Theorem~\ref{th:qubitallstatekdesignisunitarydesign}]
      We first transform the assumption that $\nu_\psi$ is a state design for all $\psi$ into a single convenient equality. By the definition of $\nu_\psi$, we have that  $\E_{\phi\sim \nu_\psi}[\phi^{\otimes k}]=\E_{U\sim \nu}[(U\psi U)^{\otimes k}]$. Then, that the distribution $\nu_\psi$ forms a state $k$-design means that $\E_{U\sim \nu}[(U\psi U)^{\otimes k}]=\E_{U\sim \textrm{Haar}}[(U\psi U)^{\otimes k}],$ which is vectorized to 
    \begin{equation}
\label{eq:eqofHaarandacteddistvec}\E_{U\sim \nu}[U^{\otimes k,k}]\text{vec}(\psi^{\otimes k})=\E_{U\sim \textrm{Haar}}[U^{\otimes k,k}]\text{vec}(\psi^{\otimes k}).
    \end{equation} 
    The subspace spanned by $\{\psi^{\otimes k}\}_{\psi}$ is the space of operators in the symmetric subspace of $(\mathbb{C}^2)^{\otimes k}$. Consequently Eq.~\eqref{eq:eqofHaarandacteddistvec} holds for all $\psi$ if and only if 
    \begin{equation}
\label{eq:eqofHaarandacteddistvecsymproj}\E_{U\sim \nu}[U^{\otimes k,k}]\Pi_{\mathrm{sym}}^{(k)}\otimes \Pi_{\mathrm{sym}}^{(k)}=\E_{U\sim \textrm{Haar}}[U^{\otimes k,k}]\Pi_{\mathrm{sym}}^{(k)}\otimes \Pi_{\mathrm{sym}}^{(k)},
    \end{equation} 
    where $\Pi_{\mathrm{sym}}^{(k)}$ is the projector into the symmetric subspace given by Eq.~\eqref{eq:symm_proj}.

    Now, in order to use representation-theoretic results, it is convenient to rewrite Eq.~\eqref{eq:eqofHaarandacteddistvecsymproj} in terms of the representation of $\mathrm{SU}(2)$ on the symmetric subspace of  $(\mathbb{C}^2)^{\otimes k}$, which we denote by $V_{\text{sym}}^{(k)}(U)\coloneqq U^{\otimes k}\Pi_{\mathrm{sym}}^{(k)}$. We have
    \begin{equation*}
\E_{U\sim \nu}[V_{\text{sym}}^{(k)*}(U)\otimes V_{\text{sym}}^{(k)}(U)]=\E_{U\sim \textrm{Haar}}[V_{\text{sym}}^{(k)*}(U)\otimes V_{\text{sym}}^{(k)}(U)].
    \end{equation*} 
       This allows us to appeal to the following general result from representation theory, which is a straightforward consequence of the Peter-Weyl theorem \cite{Diestel2014Haar}.
    \begin{lemma}
\label{lemm:equalityofexpvalsiffequalityinirreps}
        Let $G$ be a compact topological group, $g\mapsto V(g)$ a finite-dimensional unitary representation, and $\mu_1,\mu_2$  probability measures on $G$. Then
        $$\E_{g\sim \mu_1}[V(g)]=\E_{g\sim \mu_2}[V(g)]$$ if and only if 
        $$\E_{g\sim \mu_1}[W(g)]=\E_{g\sim \mu_2}[W(g)],$$
        for each  irreducible subrepresentation $W$ of $V$.
    \end{lemma}

    We apply Lemma~\ref{lemm:equalityofexpvalsiffequalityinirreps} to the representation $V=V_{\text{sym}}^{(k)*}\otimes  V_{\text{sym}}^{(k)}$ and the probability distributions $\mu_1=\nu$ and $\mu_2=\textrm{Haar}[\textrm{SU}(2)]$. We obtain that for each  irreducible subrepresentations $W$ of $V_{\text{sym}}^{(k)*}\otimes  V_{\text{sym}}^{(k)}$, $\E_{U\sim \nu}[W(U)]=\E_{U\sim \mathrm{Haar}}[W(U)]$. However, observe that, because we are working in $d=2$, and $\textrm{SU}(2)$ is self-dual, $V_{\text{sym}}^{(k)}$ and $V_{\text{sym}}^{(k)*}$ are just the ($k$+1)-dimensional representations of $\textrm{SU}(2)$, acting in the Hilbert space of a spin-$k/2$ particle. Thus, the irreducible subrepresentations of $V_{\text{sym}}^{(k)*}\otimes  V_{\text{sym}}^{(k)}$ are just the $(1,3,\dots, 2k+1$)-dimensional representations, labeled by the total spin $j=0,1,\dots, k$, obtained by adding such spins. These are the same irreps obtained from the addition of $2k$ spin-$1/2$ particles, which are also the irreducible subrepresentations of $ U^{\otimes k, k}$, where again we utilized the self-duality of $\textrm{SU}(2)$. Thus, we can apply the converse implication of Lemma~\ref{lemm:equalityofexpvalsiffequalityinirreps}, and we find that $\E_{U\sim \nu}[U^{\otimes k, k}]=\E_{U\sim \textrm{Haar}}[U^{\otimes k, k}]$, which says that $\nu$ is a unitary $k$-design. 
\end{proof}

It is worth noting that Theorem~\ref{th:qubitallstatekdesignisunitarydesign}  only holds for qubits. The underlying reason is that, if $d\geq 3$, there are irreps which appear in the representation $U\mapsto U^{\otimes k,k}$ that do not appear in the symmetric representation $U\mapsto V_{\text{sym}}^{(k)*}(U)\otimes V_{\text{sym}}^{(k)}(U)$.

\section{Ergodicity implies density}
In this appendix, we show that our notions of quantum ergodicity imply density over time, in two ways. First, if the system satisfies CHSE, then any state visits the projective Hilbert space densely in time, meaning that it eventually  comes arbitrarily close to any other state. Second, if the system satisfies CUE, then the unitary operator visits the projective unitary group densely in time, meaning that it eventually comes arbitrarily close to any other unitary.

\label{app:ergimpliesdens}
\subsection{Complete ergodicity implies density in the projective Hilbert space}
\label{app:Hergimpliesdens}
We will show that a state $\ket{\psi(t)}$ undergoing evolution which satisfies $k$\nobreakdash-HSE uniformly covers the the projective Hilbert space. To precisely quantify by we mean by uniformity, we introduce the following concept.
\begin{appdefinition}[$\varepsilon$-net and dense set]
\label{def:enetdense}
For $\varepsilon> 0$, a set of states $S\subseteq \mathbb{P}(\mathbb{C}^d)$ is an $\varepsilon$-net if for any state $\phi\in \mathbb{P}(\mathbb{C}^d)$ there exists $\psi\in S$ such that $D(\phi,\psi)\leq\varepsilon$, where $D(\phi,\psi)$ is the trace distance given by Eq.~\eqref{eq:tracedistancedef}. If $S$ forms an $\varepsilon$-net for any $\varepsilon>0$, it is said that $S$ is dense.
\end{appdefinition}
We show that, under $k$\nobreakdash-HSE, for any initial state $\psi$, its evolution is an $\varepsilon$-net, for $\varepsilon$ that grows smaller with increasing $k$ and, consequently, under CHSE, the evolution of $\psi$ is dense in  $\mathbb{P}(\mathbb{C}^d)$. To that end, we prove the following result about state $k$-designs.
\begin{lemma}
\label{lemm:kdesignsformepsilonnets}
    Let $\nu$ be a state $k$-design, and define \begin{equation}
        \gamma=\sqrt{1-\left(\frac{(d-1)!\, k!}{(d+k-1)!}\right)^{1/k}}.
    \end{equation} For any $\varepsilon\geq \gamma$, the support of $\nu$ forms an $\varepsilon$-net. 

    \begin{proof}
        Let $\phi\in \mathbb{P}(\mathbb{C}^d)$ remain fixed. We consider the quantity $\mathcal{F}\coloneqq\E_{\psi\sim \nu}\left[\abs{\braket{\phi}{\psi}}^{2k}\right]$. This is a modified frame potential [Eq.~\eqref{eq:state_frame_pot}] in which, instead of a double average, we keep one state fixed and only perform one average. We will verify that 
        $\mathcal{F}$ is lower bounded by $(1-\varepsilon^2)^k$, which implies that there is some $\psi\in \text{supp}(\nu)$ such that $\abs{\braket{\phi}{\psi}}^{2k}\geq (1-\varepsilon^2)^k$ and $D(\phi,\psi)\leq\varepsilon$. Because $\nu$ is a $k$-design,
        $$\mathcal{F}=\tr(\phi^{\otimes k}\,\E_{\psi\sim \nu}\big[{\psi}^{\otimes k}\big])=\tr({\phi}^{\otimes k}\!\!\!\!\E_{\phi\in\mathbb{P}(\mathbb{C}^d)}[\phi^{\otimes{k}}]).$$
        From Eq.~\eqref{eq:rho_Haar_symm_proj} and the fact that $\phi^{\otimes k}$ has support only in the symmetric subspace we readily obtain
\begin{equation*}
    \mathcal{F}=\frac{(d-1)!\, k!}{(d+k-1)!}\geq(1-\varepsilon^2)^k.\qedhere
\end{equation*}
    \end{proof}
\end{lemma}
Applying Lemma~\ref{lemm:kdesignsformepsilonnets} to the temporal ensemble generated by an initial state $\psi$, whose support is $S_\psi =\{\psi(t)\,|\, t\in[0,\infty)\}$, we see that $S_\psi$ is an $\varepsilon$-net under $k$\nobreakdash-HSE, as long as $\varepsilon\geq \gamma$. Now, because $\lim_{k\to \infty} \gamma=0$, we have that, under CHSE,  $S_\psi$ is dense, meaning that  for any other state $\ket{\phi}$ and $\varepsilon>0$,  there is a time at which the trace distance satisfies \begin{equation}
\label{eq:compergdensity}
    D(\phi,\psi(t))<\varepsilon.
\end{equation}

\subsection{CUE implies density in the projective unitary group}
\label{app:Uergimpliesdens}
A similar result to the previous section holds for unitary complete ergodicity. If the evolution given by $U(t)$ satisfies CUE, then $U(t)$ densely visits the projective unitary group $\operatorname{PU}(d)$, meaning that for any other unitary $V$ and $\varepsilon>0$, \begin{equation}
    D_{\operatorname{PU}(d)}(U,V)\coloneqq\sqrt{1-\frac{1}{d^2}\abs{\tr(U(t)^\dagger V)}^2}<\varepsilon
    \label{eq:unitaryergdensity}
\end{equation} at some time $t$. The quantity $\frac{1}{d^2}\abs{\tr(U^\dagger V)}^2$ is a matrix analog of the fidelity between states, as it equals $1$ if and only if $U$ equals $V$ up to some global phase, so $D_{\operatorname{PU}(d)}(U,V)$ may be understood as a matrix analog of the trace distance. 

The proof is very similar to the one for CHSE, but now using tools of unitary designs instead of state designs. We define
 \begin{align*}
\mathcal{F}\coloneqq\E_{t\geq 0}\Big[\abs{\tr U(t)^\dagger V}^{2k}\Big]=\tr(\E_{t\geq0}[U(t)^{\otimes k,k}]^\dagger V^{\otimes k,k}),
 \end{align*}
 which is a modified unitary frame potential, in which we perform only one average while keeping the unitary $V$ fixed.
 Under CUE, $$\mathcal{F}=\tr(\E_{W\in \operatorname{SU}(d)}[W^{\otimes k,k}]^\dagger V^{\otimes k,k})=\E_{W\in \operatorname{SU}(d)}[\abs{\tr W}^{2k}].$$
where the second equality holds because of the right-invariance of the Haar measure.

The quantity $\mathcal{F}_{\text{Haar}}^{(k)}=\E_{W\in \operatorname{SU}(d)}[\abs{\tr W}^{2k}]$ is  the unitary frame potential of the Haar measure, and it is well known to be $k!$ for $k\leq d$ \cite{Roberts2017,Mele2023}, but for $k>d$, which is the relevant case here, this is not longer true. In general, $\mathcal{F}=\mathcal{F}_{\text{Haar}}^{(k)}$ can be shown to be equal to the number of permutations of $\{1,\dots, k\}$ satisfying a specific subsequence-length constraint \cite{Rains1998}. This number cannot be written as a simple expression, but it can be shown to satisfy \cite{Regev1981,Stanley2006} $$\lim_{k\to \infty} {\mathcal{F}}^{1/k}=d^2.$$  This means that for any $\varepsilon$ there exits $k$ such that $d^{-2k}\mathcal{F}> (1-\varepsilon^2)^k$, which implies that at some time, $d^{-2k}\abs*{\text{tr}(U(t)^\dagger V)}{}^{2k}>(1-\varepsilon^2)^k$. Taking the $k$th root yields Eq.~\eqref{eq:unitaryergdensity}.

As a remark, although not used in this article, it is in fact true that if the system is $k$\nobreakdash-UE, for only finite $k$, then the unitary operator, in time, forms an $\varepsilon$-net over the projective unitary group. This follows directly from a unitary-design analog to Lemma~\ref{lemm:kdesignsformepsilonnets} \cite{Oszmaniec2020}.

\section{A quantum speed limit}
\label{app:speed_limit}
We show a type of quantum speed limit, in which the distance traveled by any state in the projective Hilbert space is upper bounded by the time integral of the norm of the Hamiltonian  \footnote{Proposition \ref{prop:speed_limit_prop} is a special case of Lemma~2 of Ref.~\cite{Ng2011}, taking one of the Hamiltonians to be zero}.

\begin{prop}
\label{prop:speed_limit_prop}
 Consider any state $\ket{\psi(t)}$ evolving under the unitary dynamics generated by $H(t)$. For any pair of times $t_0\leq t_1$, the trace distance between the state at time $t_0$ and the state at $t_1$ is upper bounded as follows:
$$D(\psi(t_0),\psi(t_1))\coloneqq  \sqrt{1-\abs{\braket{\psi(t_0)}{\psi(t_1)}}^2} \leq \int_{t_0}^{t_1}\!\!\!\dd{t} \norm{H(t)}_\infty$$ 
\end{prop}
\begin{proof}
 We have the following chain of inequalities:
    \begin{align*}
    \int_{t_0}^{t_1}\dd{t} \norm{H(t)}_\infty &\geq \int_{t_0}^{t_1}\dd{t} \norm{H(t)\ket{\psi(t)}} \\ &=\int_{t_0}^{t_1}\dd{t} \norm{\partial_t\ket{\psi(t)}} \\ &\geq \norm{\int_{t_0}^{t_1}\dd{t} \partial_t\ket{\psi(t)}}
    \\ &=\norm{\ket{\psi(t_1)}-\ket{\psi(t_0)}}
    \\ &=\sqrt{2-2\Re\braket{\psi(t_0)}{\psi(t_1)}}
      \\ &\geq\sqrt{1-\abs{\braket{\psi(t_0)}{\psi(t_1)}}^2}.
    \end{align*}
    The first line holds by the definition of operator norm $\norm{\, \cdot\,}_\infty$, the second is Schr\"odinger's equation, the third is the integral triangle inequality, and the fourth is the fundamental theorem of calculus.  The last inequality holds because  $2-2\text{Re}(z)\geq 1-\abs{z}^2$ for any  $z\in\mathbb{C}$.
\end{proof}

By applying the result above to a finite sequence of times, we can bound the length of the path traversed by state.
\begin{corollary}
\label{cor:dis_traveled_bound}
For times $t_0\leq t_1\leq \dots\leq t_M=T$,
$$\sum_{j=1}^M D(\psi(t_{j-1}),\psi(t_j)) \leq \int_{t_0}^{T}\dd{t} \norm{H(t)}_\infty.$$
\end{corollary}
\noindent As we take the sequence of times to have finer spacings, the left-hand side approaches the total length of the path traveled by the state $\psi$ in $\mathbb{P}(\mathbb{C}^d)$, which is seen to be upper bounded by the right-hand side, which depends only on the endpoints $t_0$ and $T$.
\section{Time-periodic systems and $k$\nobreakdash-HSE}
\label{app:periodicsystemsandkHSE}
In this appendix, we show that $k$\nobreakdash-HSE in a time-periodic Hamiltonian with period $T$ requires a Hamiltonian strength $B\coloneqq\int_0^T \dd{t}\norm{H(t)}_\infty$ which grows with $k$ and $d$.  Specifically, $k$\nobreakdash-HSE implies that $B\geq\max\{B_1,B_2\}$, with
\begin{align}
B_1&=\frac{\sqrt{2}}{3k}  \Bigg(\binom{k+d-1}{k}-1\Bigg),\label{eq:boundB1}\\
    B_2&=C\gamma^{3-2d}\geq\frac{8}{(4 d)^d}\left(\frac{k}{\log (k+1)}\right)^{d-\frac{3}{2}}, \label{eq:boundB2}
\end{align}
 $C=2^{5-4 d} (d-1)^{2-2 d}(2 d-3)^{2 d-3}$, and $\gamma=\gamma(k,d)$ as defined as in Lemma \ref{lemm:kdesignsformepsilonnets}. Theorem~\ref{th:01} follows, upon taking $k\to\infty$. The bound $B_1$ is better when $k\lesssim    d$, and is surpassed by $B_2$   when $k\gg d$. We derive each bound separately, as they require different techniques.

To obtain the bound $B_1$ in Eq.~\eqref{eq:boundB1}, we utilize the following simple combinatorial result.
\begin{lemma}
Let $(j_l)_{l\in\{0,1,\dots d-1\}}$ be a permutation of the set of integers $\{0,1,\dots , d-1\}$, where $d\geq 2$. Then,
$$\sum_{l=1}^{d-1}  \sqrt{1-\frac{\max(j_l,j_{l-1})}{d}}\geq \frac{\sqrt{2}}{3}  (d-1).$$
\label{lemm:permboundedtracedist}
\end{lemma}
\begin{proof}
Let us minimize over all possible permutations $j_l$, 
    \begin{equation}
    \label{eq:optLmaxmin}
L=\min_{j_l}\sum_{l=1}^{d-1}  \sqrt{1-\frac{\max(j_l,j_{l-1})}{d}}.
    \end{equation}
 The minimum on Eq.~\eqref{eq:optLmaxmin} is achieved for the permutation $(j_l)=(0,d-1,1,d-2,2,d-3,\dots),$ 
as the alternation between large and small numbers maximizes the values $\max(j_l,j_{l-1})$. For this permutation, $L=\sum_{j=1}^{d-1} \sqrt{\min(j,d-j)}/\sqrt{d}$ (this is easier to verify by separating the cases where $d$ is even or odd). We can lower bound this sum by the integral
\begin{equation*}
    L\geq \frac{1}{\sqrt{d}}\int_{0}^{ d-1 } \dd{x}\sqrt{\min(x,d-x)}
\geq \frac{\sqrt{2}}{3}  (d-1).\qedhere 
\end{equation*}
\end{proof}

\begin{prop}(First lower bound on Hamiltonian strength under periodic $k$\nobreakdash-HSE).
\label{prop:lowerboundonHamStrengthB1}
Let $H(t)$ be a periodic Hamiltonian with period $T$. If $H(t)$ satisfies $k$\nobreakdash-HSE, then $ B\coloneqq\int_0^T \dd{t}\norm{H(t)}_\infty\geq B_1$, as defined by Eq.~\eqref{eq:boundB1}.
\end{prop}
\begin{proof}
    We begin with the case $k=1$, where $B_1=\frac{\sqrt 2}{3}(d-1)$.  Consider a quasienergy eigenstate $\alpha(\omega t)\in \mathbb{P}(\mathbb{C}^d)$, whose existence is guaranteed by Floquet's theorem. We will apply Corollary~\ref{cor:dis_traveled_bound} on the state $\alpha(\theta)$ by finding a list of angles $\theta_0,\theta_1,\dots,\theta_{d-1}$ such the trace distance between the states $\alpha(\theta_j)$ is lower bounded as 
    \begin{equation}
D(\alpha(\theta_i),\alpha(\theta_j))\geq \sqrt{1-\frac{\max(i,j)}{d}}\label{eq:boundondistji}
    \end{equation} for $i\neq j$.
     As $\alpha(\theta)$ touches all the states $\alpha(\theta_j)$ in some order $\theta_{j_0}\leq \theta_{j_1}\leq \cdots \leq \theta_{j_{d-1}}$ (where $j_l$ is a permutation of the indices), Corollary \ref{cor:dis_traveled_bound} guarantees that
    $$B\geq  \sum_{l=1}^{d-1}  \sqrt{1-\frac{\max(j_l,j_{l-1})}{d}}\geq \frac{\sqrt{2}}{3}(d-1),$$
    where the second inequality is Lemma~\ref{lemm:permboundedtracedist}. 
    
    To construct the angles $\theta_j$ satisfying Eq.~\eqref{eq:boundondistji}, begin by setting  $\theta_0=0$. Now, inductively, assume we have already found the first $j<d$ angles $\theta_0,\theta_1,\dots,\theta_{j-1}$. We set $\Pi_j$ to be the orthogonal projector into  $\text{span}\{\ket{\alpha(\theta_0)},\dots,\ket{\alpha(\theta_{j-1})}\}$. By $1$-HSE, $\E_{\theta}[\tr(\Pi_j \alpha(\theta))]=\tr(\Pi_j)/d=j/d$, so there must exist $\theta_{j}$ such that $\tr(\Pi_j \alpha(\theta_{j}))\leq j/d$. For any $i< j$, we have $\abs{\braket{\alpha(\theta_i)}{\alpha(\theta_{j})}}^2\leq \tr(\Pi_j \alpha(\theta_{j}))\leq j/d,$
    so  $D(\alpha(\theta_i),\alpha(\theta_j))\geq \sqrt{1-j/d}$, yielding Eq.~\eqref{eq:boundondistji} and proving the bound for $k=1$.

    For the case where $k>1$, observe that the Hamiltonian $$H_{\text{sym}}(t)=\sum_{j=1}^k \mathds{1}\otimes \cdots \otimes \underbrace{H(t)}_{\mathclap{j\text{-th entry}}}\otimes \cdots \otimes \mathds{1}$$ acting on the symmetric subspace of $(\mathbb{C}^d)^{\otimes k}$ satisfies ${1}$\nobreakdash-HSE if and only if $H(t)$ satisfies $k$-HSE, because $\{\dyad{\phi}^{\otimes k}\}_\phi$  spans the space of  operators in the symmetric subspace [Eq.~(11b), Ref.~\cite{Harrow2013}]. Consequently, we can apply the case $k=1$  on $H_{\text{sym}}(t)$, which gives
    $\int_0^T \dd t\norm{H_\text{sym}(t)}_{\infty}\geq \frac{\sqrt 2}{3}(d_{\text{sym}}-1)$, where $d_{\text{sym}}=\binom{k+d-1}{k}$ is the dimension of the symmetric subspace. Finally, note that $\norm{H_\text{sym}(t)}_{\infty}=k\norm{H(t)}_{\infty}$, which yields the desired result.
\end{proof}

Now we prove the bound $B_2$ given by Eq.~\eqref{eq:boundB2}, for which we require the following lemma regarding the geometry of $\mathbb{P}({\mathbb{C}^d})$.

\begin{lemma}(Lower bound on the packing number of complex projective space).
\label{lemm:packingprojspace}
For any $\varepsilon\in(0,1)$, we can pack inside $\mathbb{P}(\mathbb{C}^d)$ at least $n=\lceil\varepsilon^{-2(d-1)}\rceil$ disjoint balls \footnote{Balls are taken with respect to the trace distance, i.e., $B({\phi,r})=\{\psi \in \mathbb{P}(\mathbb{C}^d) | D(\psi,\phi)<r\}$.} of radius $\varepsilon/2$, where $\lceil\,\cdot\,\rceil$ denotes the ceiling function.
\label{lemm:D2}
\begin{proof}
The following is a standard argument in covering and packing theory, which we include here for completeness.

 Let $n_{\text{max}}$ be largest number of disjoint balls of radius $\varepsilon/2$ that we can pack inside $\mathbb{P}(\mathbb{C}^d)$. Take $\mathcal{C}=\{\phi_1,\dots \phi_{n_\text{max}}\}$ to be the centers of the balls forming such a maximal packing.  The set $\mathcal{C}$ forms an $\varepsilon$-net, for, if it did not, there would exist a state which is more than $\varepsilon$ trace distance away from any state in $\mathcal{C}$, which would mean that we can add another ball of radius $\varepsilon/2$ to the packing, contradicting the maximality of $n_\text{max}$. Then, all the balls of radius $\varepsilon$ centered at points in $\mathcal{C}$ together cover the whole $\mathbb{P}(\mathbb{C}^d)$. Each ball has volume $V(\varepsilon)$, and normalizing the total volume of $\mathbb{P}(\mathbb{C}^d)$ to unity, we must have  $n_\text{max}V(\varepsilon)\geq 1$. Thus we can pack at least $n=\lceil(1/V(\varepsilon)\rceil\leq n_{\text{max}}$ balls of radius $\varepsilon$ inside   $\mathbb{P}(\mathbb{C}^d)$.

 To finish the proof, we explicitly compute the volume  $V(\varepsilon)$. The distribution of the overlap $x=\abs{\braket{\psi}{\phi_0}}^2=1-D(\psi,\phi_0)^2$ of a Haar-random state $\psi\in \mathbb{P}(\mathbb{C}^d)$ with a fixed state $\phi_0$ is given by 
$p(x)=(d-1)(1-x)^{d-2}$ [Eq.~(14) of Ref.~\cite{Kus1988}], from which we can derive the volume  \cite{Oszmaniec2024}
\begin{equation}
    V(\varepsilon)=\int_{D(\psi,\phi_0)<\varepsilon} \!\!\!\!\!\dd{\psi}=\int_{1-\varepsilon^2}^1 \dd{x}p(x) =\varepsilon^{2(d-1)},
\end{equation}
which gives $n=\lceil\varepsilon^{-2(d-1)}\rceil$, as claimed.
\end{proof}
\end{lemma}

\begin{prop}(Second lower bound on Hamiltonian strength under periodic $k$\nobreakdash-HSE).
\label{prop:lowerboundonHamStrengthB2}
Let $H(t)$ be a periodic Hamiltonian with period $T$. If $H(t)$ satisfies $k$\nobreakdash-HSE, then $ B\coloneqq\int_0^T \dd{t}\norm{H(t)}_\infty\geq B_2$, as defined by Eq.~\eqref{eq:boundB2}
\end{prop}
\begin{proof}
Again, consider a quasienergy eigenstate $\alpha(\omega t)\in \mathbb{P}(\mathbb{C}^d)$. By $k$\nobreakdash-HSE, the curve $\{\alpha(\theta)\,|\,\theta\in\mathbb{T}\}$ forms a $\gamma$-net, taking $\gamma$ as in Lemma \ref{lemm:kdesignsformepsilonnets}.  For any $\delta\in(0,1)$, we can pack at least $n=\lceil\delta^{-2(d-1)}\rceil$ balls of radius $\delta/2$ inside $\mathbb{P}(\mathbb{C}^d)$ (Lemma~\ref{lemm:D2}). That is, there exists a set of states $\{\phi_j\}_{j=1}^{n}$ (the centers of the balls) whose pairwise trace distances are lower bounded, $D(\phi_i,\phi_j)>\delta$ for $i\neq j$.  By the $\gamma$-net property, we can find angles $\theta_j$ so that $D(\alpha(\theta_j),\phi_j)<\gamma$. The angles may be assumed to be sorted, relabeling the $\phi_j$ otherwise. By Corollary~\ref{cor:dis_traveled_bound} and the triangle inequality,
      \begin{align}
          B&\geq \sum_{j=1}^{n} D(\alpha(\theta_{j}),\alpha(\theta_{j+1 \mod n}))\geq n\left(\delta-2\gamma\right)\nonumber \\&\geq \delta^{-2(d-1)}\left(\delta-2\gamma\right). \label{eq:Bgeqval}
      \end{align}
Maximizing Eq.~\eqref{eq:Bgeqval} over $\delta$, we get $B\geq B_2$ for $\delta=4 (d-1) \gamma /(2 d-3)<1$. 
One can verify that $B_2\geq \frac{8}{(4 d)^d}(\frac{k}{\log (k+1)})^{d-{3}/{2}}$ by applying Stirling's approximation to the binomial in $\gamma$.
\end{proof}

\section{Proof of Theorem \ref{th:03}}
\label{app:proofTH3}
Let $H(t)$ be an $m$-quasiperiodic Hamiltonian with a basis of piecewise smooth quasienergy eigenstates, i.e. such that the generalized Floquet decomposition given by Eq.~\eqref{eq:quasienergydecomp} holds. We will show that, under CUE, necessarily $m\geq d(d-1)$.

In Eq.~\eqref{eq:quasienergydecomp}, we may write $Q=\text{diag}(q_0,\dots,q_{d-1})$ as a diagonal matrix in the basis of QEs. Because global phases are irrelevant, we may assume that $Q$ is traceless, so that $q_{0}=-\sum_{\alpha=1}^{d-1}q_\alpha$, giving a total of $d-1$ rationally independent quasienergies $\bm{q}=(q_1,\dots, q_{d-1})$.  
The exponential $$e^{-iQt}=\text{diag}\big(e^{i\sum_{\alpha=1}^{d-1}q_\alpha t}, e^{-iq_1t},\dots, e^{-iq_{d-1}t}\big)$$ is a quasiperiodic function, with frequency vector contained in $\bm q$, and $U(t)=P(\bm \omega t ) e^{-iQ t}$ overall is a quasiperiodic function, with frequency vector contained in $(\bm q,\bm \omega)$. We say ``contained in," and not ``equal to," because the driving frequencies may be reducible (e.g. if there is rational dependence), but, regardless, we are guaranteed that the map $t\mapsto U(t)\otimes U(t)^*$ is $n$-quasiperiodic, for some $n\leq m+d-1$.

      Furthermore, if the evolution satisfies CUE,  the $n$-quasiperiodic map $t\mapsto U(t)\otimes U(t)^*$ densely visits the projective unitary group $\operatorname{PU}(d)$ (see Appendix~\ref{app:Uergimpliesdens}). Then the parent function $\mathbb{T}^{n}\to \operatorname{PU}(d)$ is also dense. By assumption, this map is piecewise smooth, so  $$m+d-1\geq n=\text{dim}(\mathbb{T}^{n})\geq \text{dim}(\operatorname{PU}(d))=d^2-1,$$
      which gives $m\geq d(d-1)$. 
\qed

\section{Complete unitary ergodicity with  $m=d^2-2$ tones}
\label{app:CUEdrive}

In this appendix, we explain how to construct a measure preserving surjective function from the $(d^2-1)$-dimensional torus to $\operatorname{SU}(d)$. We then utilize this map to construct a $(d^2-2)$-quasiperiodic Hamiltonian which has QEs and satisfies CUE.

We consider Hurwitz's Euler-angle parametrization of $\operatorname{SU}(d)$ \cite{Hurwitz1897,Zyczkowski1994,Diaconis2017}, which is constructed as follows.
    For $j\in\{1,2,\dots, d-1\}$, define the two-level unitary rotation matrices  $$R_j(\xi,\varphi,\eta)=\begin{pmatrix}
        \mathds{1}_{j-1} & & & \\
    & \cos \xi e^{i \eta} &-\sin \xi e^{i \varphi} & \\
    & \sin \xi e^{-i \varphi} & \cos \xi e^{-i \eta} & \\
    & & & \mathds{1}_{d-j-1}
    \end{pmatrix},$$
    and for $r\in\{1,2,\dots, j\}$ define the Euler angles,
    \begin{align}
        \xi_{r,j}\in [0,\tfrac{\pi}{2})&&\varphi_{r,j}\in [0,2\pi)&&\eta_{j}\in [0,2\pi).
    \end{align}
    which are, in total, $d^2-1$.
    Consider the matrices
       \begin{align*}
       E_1&=R_1(\xi_{1,1},\varphi_{1,1},\eta_{1})\\
       E_2&=R_2(\xi_{2,2},\varphi_{2,2},0)R_1(\xi_{1,2},\varphi_{1,2},\eta_{2})\\
       E_3&=R_3(\xi_{3,3},\varphi_{3,3},0)R_2(\xi_{2,3},\varphi_{2,3},0)R_1(\xi_{1,3},\varphi_{1,3},\eta_{3})\\
       \vdots \\
       E_j&=R_j(\xi_{j,j},\varphi_{j,j},0)R_{j-1}(\xi_{j-1,j},\varphi_{j-1,j},0)\cdots \\&R_2(\xi_{2,j},\varphi_{2,j},0) R_1(\xi_{1,j},\varphi_{1,j},\eta_{j}),
          \end{align*}
          for $j\leq  d-1$,
       and multiply them all together, to obtain   
$$V(\xi_{r,j},\varphi_{r,j},\eta_{j})=E_1E_2\cdots E_{d-1},$$
        which yields a parametrization of $\operatorname{SU}(d)$. One can compute the Haar measure of $\textrm{SU}(d)$ to be \cite{Diaconis2017}
        $$\ \dd{U}_{\text{Haar}} \propto \prod_{1\leq j \leq d-1} \dd \eta_{j} \prod_{1 \leq r\leq j}  \dd{\big[\left(\sin \xi_{r,j}\right)^{2r}\big]} \dd \varphi_{r,j},$$
        which means that this parametrization is not measure preserving, because of that $\sin$ term. However, we can make it measure preserving by considering a change of variables $\theta_{r,j}\mapsto \xi_{r,j}$ given by $$\xi_{r,j}(\theta_{r,j})=\arcsin(\abs{1-\theta_{r,j}/\pi}^{1/2r}),$$ for $\theta_{r,j}\in [0,2\pi)$, which gives $\dd{\big[\left(\sin \xi_{r,j}\right)^{2r}\big]}=\dd{\abs{1-\theta_{r,j}/\pi}}\propto \dd{\theta_{r,j}}$. Then the map $$(\theta_{r,j},\varphi_{r,j},\eta_{j})\mapsto V(\xi_{r,j}(\theta_{r,j}),\varphi_{r,j},\eta_{j})$$ is measure preserving from $\mathbb{T}^{d^2-1}$ to $\operatorname{SU}(d)$.

     To construct a drive that satisfies CUE with QEs, consider $q_1,\dots,q_{d-1},\omega_1,\omega_2,\dots, \omega_{d^2-2}$ to be rationally independent frequencies. We assign two of the Euler angles as $\eta_{d-1}=q_{1}t$,  $\varphi_{1,d-1}=(\omega_{d^2-2}-q_{1})t$. The remaining $d^2-3$ angles $\theta_{r,j},\varphi_{r,j},\eta_{j}$ are set to be equal to $\omega_1t,\dots \omega_{d^2-3}t$, respectively. By this assignment, the parametrization $V$ is a function of $\bm \omega t=(\omega_1,\omega_2,\dots, \omega_{d^2-2})t$ and $q_1t$. Now, using the fact that \begin{align*}
        R_1&\big(\xi,\varphi=(\omega+q)t,\eta=qt\big)=\\&R_1\big(\xi,\varphi=\omega t,\eta=0\big)\exp(-i\,\text{diag}(-q,q,0,\dots,0)t),
     \end{align*}
     it is seen that
     \begin{align}
         P(\bm \omega t)\coloneqq  V(\bm \omega t, q_1 t ) \exp(i\,\text{diag}(-q_1,q_1,0,\dots,0)t)V(\bm 0)^\dagger
     \end{align}
     depends only on $\bm \omega t$ and not on $q_1t$. Then, we may define the evolution operator by the generalized Floquet decomposition
    \begin{equation}
        \label{eq:utfloquetCUE} U(t)\coloneqq  P(\bm \omega t)e^{-iQt}= V(\bm \omega t, q_1 t ) e^{-i\widetilde{Q}t} V(\bm 0)^\dagger,
    \end{equation}
    where the matrix $Q=\text{diag}(-\sum_{\alpha=1}^{d-1} q_\alpha,q_1,q_2,\dots,q_{d-1})$ is the diagonal matrix of quasienergies, and $\widetilde{Q}=\text{ diag}(-\sum_{\alpha=2}^{d-1} q_\alpha,0,q_{2},\dots,q_{d-1})$.

     Because of the rational independence of  $(\bm \omega,\bm q)$, the map $t\mapsto (\bm{\omega},\bm q)t$ uniformly covers the ${(d^2+d-3)}$\nobreakdash-dimensional torus. Using that $V$ forms a measure preserving map, from $\mathbb{T}^{d^2-1}$ to $\operatorname{SU}(d)$, we obtain 
     \begin{align}
         \E_{t\geq 0}\Big[U(t)^{\otimes k,k}\Big]&=\E_{t\geq 0}\Big[ \E_{W\in\operatorname{SU}(d)}[(We^{-i\widetilde{Q}t})^{\otimes k,k}]\Big]\\&=\E_{W\in\operatorname{SU}(d)}\Big[W^{\otimes k,k}\Big],
     \end{align} where the second equality holds by the right invariance of the Haar measure. This proves that  the $(d^2-d)$-quasiperiodic Hamiltonian $H(t) = i(\partial_t U(t))U(t)^{\dagger}$ satisfies CUE and has, by construction, QEs with the preselected quasienergies $q_0,\dots,q_{d-1}$.

\section{Sufficient and necessary conditions for $k$\nobreakdash-HSE with quasienergy eigenstates}
In this appendix, we assume that an $m$-quasiperiodic Hamiltonian $H(t)$ has a basis of QE $\ket{\alpha(t)}=e^{-i q_\alpha t}\ket{\alpha(\bm \theta={\bm \omega t})}$, with $\alpha\in\{0,\dots,d-1\}$. We derive a property on $\ket{\alpha(\bm \theta={\bm \omega})}$ which is equivalent to $k$\nobreakdash-HSE. Specifically, we see that $k$\nobreakdash-HSE is equivalent to requiring that all tensor product combinations (of length $k$) of the states $\alpha(\bm \theta)$, averaged over the torus and symmetrized, are equal to $\rho_{\mathrm{Haar}}^{(k)}.$ To show this, we require first to compute the time average of an arbitrary state expanded in the basis of QE.  

\begin{lemma}
    Let $H(t)$ be a quasiperiodic Hamiltonian with QEs, such that the quasienergies (possibly excluding one) and  $\bm \omega$ are rationally independent and a state $\ket{\psi(t)}=\sum_{\alpha}c_\alpha e^{-iq_\alpha t}\ket{\alpha(\bm\theta=\bm\omega t)}$. Then \begin{align}
\label{eq:time_average_of_state_reduced_sim}
    \E_{t\geq0}[\psi(t)^{\otimes k}]= \sum_{{\alpha_1,\dots, \alpha_k=0}}^{d-1} \abs{c_{\alpha_1}}^2\cdots \abs{c_{\alpha_{k}}}^2\rho_{\mathrm{sym}(\bm \alpha)},
\end{align}
where $\rho_{\mathrm{sym}(\bm \alpha)}\coloneqq\mathcal{P}_{\bm \alpha}\Pi_{\mathrm{sym}}^{(k)}\E_{\bm\theta \in \mathbb{T}^m} \Big[\bigotimes_{j=1}^k{\alpha_j(\bm \theta)}\Big]\Pi_{\mathrm{sym}}^{(k)}$, with $\Pi_{\mathrm{sym}}^{(k)}$  the orthogonal projector into the symmetric subspace and $\mathcal{P}_{\bm \alpha}$ a normalization factor, equal to the total number of different permutations of $\bm \alpha=(\alpha_1,\dots,\alpha_k)$.
\label{lemm:G01}
\end{lemma}

\noindent To gain intuition, it is useful to first understand the time-independent version of Lemma \ref{lemm:G01}, derived in Ref. \cite{Mark2024}. If the Hamiltonian has no time dependence, it has proper eigenstates $\ket{\alpha}$, and $\rho_{\mathrm{sym}(\bm \alpha)}=\mathcal{P}_{\bm \alpha}\Pi_{\mathrm{sym}}^{(k)}\bigotimes_{j=1}^k{\alpha_j}\Pi_{\mathrm{sym}}^{(k)}$ reduces to a symmetrized product of $\alpha_1,\dots,\alpha_k$.  
We generalize this result to quasiperiodic systems, where the only difference is an additional average over the torus.

\begin{proof}
    In the statement, we allow one quasienergy to not be rationally independent but, in fact, up to an irrelevant global phase, we can shift all quasienergies by adding a constant multiple of the identity to $H(t)$. This constant can be chosen to ensure that all quasienergies and $\bm \omega$ are rationally independent, which we henceforth assume. Moreover, note that although the quasienergies are only defined up to a shift $\bm n\cdot \bm \omega$, this condition is preserved upon substituting $q_\alpha\rightarrow q_\alpha+\bm n_\alpha\cdot \bm \omega$, so it is a well-defined condition on the quasienergy spectrum.

By the rational independence and the quasiperiodicity of the states $\alpha_j(t)$, we can split the time average in two separate averages, one corresponding to the winding quasienergy phases, and the other to the quasienergy eigenstates, defined over the torus,
\begin{align}
\label{eq:timeavestate}
    \E_{t\geq0}[\psi(t)^{\otimes k}]= &\sum_{\substack{\bm \alpha \bm  \beta}}\Bigg(\prod_{j=1}^k c_{\alpha_j}c_{\beta_j}^*\E_{t\geq 0}\Big[e^{-i\sum_{j=1}^k (q_{\alpha_j}-q_{\alpha_j'})t}\Big]\nonumber\\ & \quad\times\E_{\bm\theta \in \mathbb{T}^m}\Big[\bigotimes_{j=1}^k\dyad{\alpha_j(\bm \theta)}{\beta_j(\bm \theta)}\Big]\Bigg),
\end{align}
where the sum runs over all possible pairs of tuples of indices $\bm \alpha=(\alpha_1,\dots,\alpha_k)$, $\bm  \beta=(\beta_1,\dots,\beta_k)$.
 The time average of the exponential in Eq.~\eqref{eq:timeavestate} is 
 \begin{equation}
 \label{eq:time_integral_phases}
\E_{t\geq 0}[e^{-i\sum_{j=1}^k (q_{\alpha_j}-q_{\beta_j})t}]=\begin{cases}1\text{ if } \bm \beta \in \text{Perms}(\bm \alpha)\\ 0\text{ else,}\end{cases}
 \end{equation}
 where $\text{Perms}(\bm \alpha)$ is the set of all permutations of $\bm \alpha$. This is a consequence of the rational independence, which only allows the linear combination $\sum_{j=1}^k (q_{\alpha_j}-q_{\beta_j})$ to be zero if $\bm \beta$ is a permutation of $\bm \alpha$. This last statement is  called the no $k$-resonance condition in Ref.~\cite{Mark2024}. 
 
 By writing the symmetric projector explicitly in terms of permutation operators [Eq.~\eqref{eq:symm_proj}], we see that
 \begin{equation}
     \label{eq:permsysmm}
 \sum_{\bm \beta\in  \mathrm{Perms}(\bm \alpha)}\bigotimes_{j=1}^k\dyad{\alpha_j(\bm \theta)}{\beta_j(\bm \theta)}=\mathcal{P}_{\bm \alpha}\bigotimes_{j=1}^k{\alpha_j(\bm \theta)}\Pi_{\mathrm{sym}}^{(k)}.
  \end{equation}
 Inserting Eqs.~\eqref{eq:time_integral_phases} and \eqref{eq:permsysmm} into Eq.~\eqref{eq:timeavestate}, and overall left-multiplying by $\Pi_{\mathrm{sym}}^{(k)}$, we obtain Eq.~\eqref{eq:time_average_of_state_reduced_sim}.
\end{proof}
\begin{apptheorem}
\label{th:HSEconditionintermsofeigenstates}
Let $H(t)$ be a quasiperiodic Hamiltonian with QEs, such that the quasienergies (possibly excluding one) and  $\bm \omega$ are rationally independent. Then $H(t)$ satisfies $k$\nobreakdash-HSE if and only if for every $\bm \alpha=(\alpha_1,\dots,\alpha_k)\in \{0,\dots,d-1\}^k$,
\begin{equation}
\label{eq:thH2statmentform}
\rho_{\mathrm{sym}(\bm \alpha)}=\rho_{\mathrm{Haar}}^{(k)},
\end{equation}
where $\rho_{\mathrm{sym}(\bm \alpha)}$ is defined in Lemma \ref{lemm:G01} and $\rho_{\mathrm{Haar}}^{(k)}\coloneqq\E_{\phi\in\mathbb{P}(\mathbb{C}^d)}[\phi^{\otimes{k}}]$
.\end{apptheorem}
\begin{proof}

This result follows entirely from Eq.~\eqref{eq:time_average_of_state_reduced_sim}. If we assume Eq.~\eqref{eq:thH2statmentform}, then Eq.~\eqref{eq:time_average_of_state_reduced_sim} reduces to
$k$\nobreakdash-HSE, by noting that $\sum_{\bm \alpha} \abs{c_{\alpha_1}}^2\cdots \abs{c_{\alpha_{k}}}^2=1$. Conversely, if we assume $k$\nobreakdash-HSE, then from Eq.~\eqref{eq:time_average_of_state_reduced_sim}, we see that the polynomials  defined over all $\mathbb{R}^d$, \begin{align}
\label{eq:polexpr}
    P(x_0,\dots,x_{d-1})&= \sum_{{\bm \alpha} } \rho_{\mathrm{sym}({\bm \alpha}) } x_0^{n_0}x_1^{n_1}\cdots x_{d-1}^{n_{d-1}}, \\Q(x_0,\dots,x_{d-1})&= \sum_{\bm \alpha}\,\,\,\,\rho_{\text{Haar}}^{(k)}\,\,\,\, x_0^{n_0}x_1^{n_1}\cdots x_{d-1}^{n_{d-1}},\nonumber
\end{align}
 coincide for values $(x_1,\dots,x_k)\in[0,1]^k$ that satisfy $\sum_{j=1}^k {x_k}=1$, where $n_\alpha$ counts the number of times $\alpha$ appears in the tuple $\bm \alpha=(\alpha_1,\dots,\alpha_k)$. This is seen by taking the initial state $\ket{\psi}$  to have coefficients $c_{\alpha}=\sqrt{x_{\alpha}}$. It follows that $P$ and $Q$ must be equal everywhere, and thus equal as polynomials, meaning that each of their coefficients is the same. Note that there may be repeated terms in the expressions \eqref{eq:polexpr}, due to the existence permutations of $(\alpha_1,\dots,\alpha_k)$ that produce the same values of $n_\alpha$. However, by the symmetry of $\rho_{\mathrm{sym}(\bm\alpha)}$, the coefficients for the repeated terms are the same, 
 guaranteeing that $\rho_{\mathrm{sym}(\bm\alpha)}=\rho^{(k)}_{\text{Haar}}$ for all $\bm \alpha$.
\end{proof}

Theorem~\ref{th:HSEconditionintermsofeigenstates} provides a set of conditions to verify $k$\nobreakdash-HSE in quasiperiodic systems which feature QEs. Moreover, as we prove below, when applied to a single-qubit Hamiltonian, these conditions simplify greatly: One just needs to analyze a single quasienergy eigenstate to guarantee that the whole system is $k$\nobreakdash-HSE (and further $k$\nobreakdash-UE  by Corollary~\ref{cor:qubitHSEimpliesUE}).
\begin{corollary}
\label{cor:singleeigenstateimpliesUE}
    If $H(t)$ is a single-qubit quasiperiodic Hamiltonian with a quasienergy eigenstate that satisfies the $k$\nobreakdash-HSE (CHSE) condition [Eq.~\eqref{eq:kHSE_cond}] and a  quasienergy that is rationally independent from the driving frequencies, then $H(t)$ satisfies $k$\nobreakdash-UE (CUE).
\end{corollary}
\begin{proof}
   Assume that $0(\bm \theta)$ satisfies the $k$\nobreakdash-HSE condition. We will show that this implies that $H(t)$ satisfies $k$\nobreakdash-UE. 
   
   The second quasienergy eigenstate is guaranteed to exist (see Corollary 3.4 in Ref.~\cite{Blekher1992}), determined by the resolution of the identity $1(\bm \theta)=\mathds{1}-0(\bm \theta)$. We compute $\Pi_{\mathrm{sym}}^{(k)}{\mathbb{E}}_{\bm \theta}\big[\bigotimes_{j=1}^k {\alpha_j}(\bm \theta)\big]\Pi_{\mathrm{sym}}^{(k)}$ for arbitrary $\bm \alpha\in \{0,1\}^k$ by noting that, in between the projectors to the symmetric subspace $\Pi_{\mathrm{sym}}^{(k)}$, the tensor product $\otimes$ becomes commutative, allowing for algebraic manipulation,
\begin{align*}
    &\frac{\rho_{\mathrm{sym}(\bm \alpha)}}{\mathcal{P}_{\bm \alpha}}=  \Pi_{\mathrm{sym}}^{(k)} \mathbb{E}_{\bm \theta}\Big[0(\bm \theta)^{\otimes k- \abs{\bm \alpha}}\otimes\left(\mathds{1}-0(\bm \theta)\right)^{\otimes \abs{\bm \alpha}}\Big] \Pi_{\mathrm{sym}}^{(k)} \\
    &=\Pi_{\mathrm{sym}}^{(k)} \sum_{j=0}^\abs{\bm \alpha}\binom{\abs{\bm \alpha}}{j}(-1)^{\abs{\bm \alpha}-j}\Big(\mathds{1}^{\otimes j}\otimes \mathbb{E}_{\bm \theta}\Big[0(\bm \theta)^{\otimes k-j}\Big] \Big)\Pi_{\mathrm{sym}}^{(k)} 
    \end{align*}
where we used the binomial theorem in the second equality.

Using that the state $0(\bm\theta)$ satisfies the $k$\nobreakdash-HSE condition, and in consequence the $(k-j)$-HSE condition [see Corollary~\ref{cor:relationsHSEandUE} (a)], we have $\mathbb{E}_{\bm \theta}\big[0(\bm \theta)^{\otimes k-j}\big]=\rho_{\text{Haar}}^{(k-j)}={\Pi_{\mathrm{sym}}^{(k-j)}}/{(1-j+k)}$, by Eq.~\eqref{eq:rho_Haar_symm_proj} with $d=2$. Further noting that 
$ \Pi_{\mathrm{sym}}^{(k)}(\mathds{1}^{\otimes j}\otimes \Pi_{\mathrm{sym}}^{(k-j)} )\Pi_{\mathrm{sym}}^{(k)}= \Pi_{\mathrm{sym}}^{(k)}$, we may compute
    \begin{align}
    {\rho_{\mathrm{sym}(\bm \alpha)}}
      &=  \binom{k}{{\abs{\bm{\alpha}}}}\sum_{j=0}^\abs{\bm \alpha}\binom{\abs{\bm \alpha}}{j}\frac{(-1)^{\abs{\bm \alpha}-j}}{1 - j + k} \Pi_{\mathrm{sym}}^{(k)}\\&=\frac{\Pi_{\mathrm{sym}}^{(k)}}{1+k}= \rho_{\text{Haar}}^{(k)}.
\end{align}
By Theorem \ref{th:HSEconditionintermsofeigenstates}, the Hamiltonian satisfies $k$\nobreakdash-HSE, and by Corollary~\ref{cor:qubitHSEimpliesUE}, this further implies $k$\nobreakdash-UE. 
\end{proof}
\section{$k$\nobreakdash-HSE in the frequency lattice}
\label{app:Fourier}
In this appendix we derive a set of equations in Fourier space, which are sufficient and necessary for the system to satisfy $k$\nobreakdash-HSE. We consider the case where the quasienergy eigenstates  $\ket{\alpha(t)}=e^{-i q_\alpha t}\ket{\alpha(\bm \theta={\bm \omega t})}$ exist and allow for a Fourier decomposition
\begin{equation}
\label{eq:Fourier_transform_QE}
    \ket{\alpha(\bm \theta)}=\sum_{\bm n\in\mathbb{Z}^m} \ket{\alpha_{\bm n}}e^{-i\bm n\cdot \bm\theta}.
\end{equation}
The Fourier components $\ket{\alpha_{\bm n}}$ do not need to be normalized. They can be understood as the partial components of the eigenstates of a time-independent Hamiltonian defined over a so-called frequency lattice \cite{Ho1983,Verdeny2016,Long2022}.

All the information about the dynamics is encoded in the Fourier components $\ket{\alpha_{\bm n}}$, allowing us to write $k$\nobreakdash-HSE as a condition in terms of them. By Fourier transforming the matrices $\rho_{\mathrm{sym}({\bm \alpha})}$ in  Theorem~\ref{th:HSEconditionintermsofeigenstates} and assuming the rational independence hypothesis, $k$\nobreakdash-HSE can be recast as
\begin{equation}
\label{eq:kHSEconstraints}
    \mathcal{P}_{\bm \alpha}\Pi_{\mathrm{sym}}^{(k)}\sum_{\bm n_j,\bm n'_j\in\mathcal{K}}\bigotimes_{j=1}^k\dyad*{{\alpha_j}_{\bm n_j}}{{\alpha_j}_{\bm n'_j}}\Pi_{\mathrm{sym}}^{(k)}=\rho_{\mathrm{Haar}}^{(k)},
\end{equation}
for all $\bm \alpha =(\alpha_1,\dots,\alpha_k)$,
where the sum runs over $$\mathcal{K}=\Big\{(\substack{\bm n_1,\dots, \bm n_k,\\ \bm n'_1,\dots, \bm n'_k})\in(\mathbb{Z}^m)^{2k}\,\Big|\, \sum_{j=1}^k\bm n_j={\sum_{j=1}^k \bm n'_j}\Big\}.$$ The Fourier components must satisfy an additional orthonormality constraint, due to the unitarity of the dynamics: The orthonormality condition of the quasienergy eigenstates $\forall \bm{\theta}\in \mathbb{T}^m\colon\braket{\alpha(\bm \theta)}{\alpha'(\bm \theta)}=\delta_{\alpha\alpha'}$ is Fourier transformed, via the convolution theorem, to
\begin{align}
\label{eq:orthonormalityconstraints}
 \forall{\bm n'}\in\mathbb{Z}^m:\sum_{\bm n\in\mathbb{Z}^m} \braket{{\alpha_{ \bm n}}}{{\alpha'}_{\bm n' + \bm n}}=\delta_{\alpha\alpha'} \delta_{\bm n'\bm 0}.
\end{align}

Equations \eqref{eq:orthonormalityconstraints} and \eqref{eq:kHSEconstraints} completely characterize the Fourier components of the QEs under $k$\nobreakdash-HSE, in the sense that if one constructs a family of vectors satisfying them, it is possible to then reconstruct  an  $m$-quasiperiodic Hamiltonian that satisfies $k$\nobreakdash-HSE. This can be done by constructing the quasienergy eigenstates $\ket{\alpha(\bm \theta)}$ via Eq.~\eqref{eq:Fourier_transform_QE}, and from them the evolution operator via the generalized Floquet decomposition Eq.~\eqref{eq:quasienergydecomp}, where the (rationally independent) quasienergies and driving frequencies can be chosen freely.

For brevity, we say that a set of vectors $\ket{\alpha_{\bm n}}\in\mathbb{C}^d$, with $\alpha\in\{0,\dots,d-1\}$ and $\bm n\in\mathbb{Z}^m$, is an  $(m,k)$-ergodic lattice [$(m,k)$-EL] if Eqs.~\eqref{eq:kHSEconstraints} and \eqref{eq:orthonormalityconstraints} are satisfied. In what is left of this appendix, we provide examples of finite $(m,k)$-ELs, where finite means that there is only a finite number of nonzero vectors $\ket{\alpha_{\bm n}}$. Finite $(m,k)$-ELs give rise to $m$-quasiperiodic Hamiltonians with analytic time dependence, which have QEs and satisfy $k$\nobreakdash-HSE.

A $(1,1)$-EL yields a periodic Hamiltonian that satisfies $1$-HSE. For $m=1$, $k=1$,  Eq.~\eqref{eq:kHSEconstraints} reduces to $\sum_{n\in\mathbb{Z}}\dyad{\alpha_n}=\mathds{1}/d$, which is readily satisfied, along with Eq.~\eqref{eq:orthonormalityconstraints} by \begin{equation}
    \ket{\alpha_n}= \frac{1}{\sqrt{d}}e^{2\pi i  n \alpha/d}\ket{v_{ n}}
\end{equation}
for $ n\in \{0, \dots, d-1\}$ (and $\ket{\alpha_n}=0$ for other $ n \in \mathbb{Z}$), where $\{\ket{v_l}\}_{l=0}^{d-1}$ forms an orthonormal basis of $\mathbb{C}^d$. This proves that $1$-HSE is achievable by analytic time-periodic dynamics, in arbitrary dimension.

We now specialize to the case of a single qubit, $d=2$. By Corollary~\ref{cor:singleeigenstateimpliesUE}, to guarantee $k$\nobreakdash-UE we only need to study the components of one quasienergy eigenstate, say $\ket{0_{\bm n}}$. The components of the orthogonal state are determined by $\ket{1_{\bm n}}=\ket{0_{-\bm n}}_{\perp}$, where $\tbinom{a}{b}_\perp=\tbinom{-b^*}{a^*}$. Consequently, it is enough to solve Eq.~\eqref{eq:kHSEconstraints} for $\bm \alpha=(0,0,\dots,0)$, i.e.
$ \sum_{\bm n_j,\bm n'_j\in\mathcal{K}}\bigotimes_{j=1}^k\dyad*{{0}_{\bm n_j}}{{0}_{\bm n'_j}}=\rho_{\mathrm{Haar}}^{(k)}$. We numerically find solutions for $(m=1, k=2)$, and $(m=2, k=3)$, giving rise to single-qubit periodic and two-quasiperiodic analytic Hamiltonians which satisfy $2$-UE and $3$-UE, respectively.

An $(m=1,k=2)$-EL in a qubit is generated by
\begin{equation}
    \begin{array}{ccccc}
    \ket{0_n}=&a_+\ket{\phi_-},\, & -a_-\ket{\phi_+},\, & a_-\ket{\phi_-},\, & -a_+\ket{\phi_+}\\
    \,&(n=0) &(n=1)&(n=2)&(n=3),\\\end{array}
\end{equation}
and $\ket{0_n}=0$ for other $n\in\mathbb{Z}$,
where $a_\pm=\frac{1}{2} \sqrt{1\pm\frac{1}{\sqrt{3}}}$ and $\ket{\phi_\pm}$ are any basis states.  The state $\ket{0(\theta)}=\sum_{n}e^{-i\theta}\ket{(\alpha=0)_n}$ is displayed in Fig~\ref{fig:5}b, with the selection $\ket{\phi_{\pm}}=-\sqrt{\frac{1}{2}\pm\frac{1}{\sqrt{6}}}\ket{0}\pm e^{3 i \pi/4}\sqrt{\frac{1}{2}\mp \frac{1}{\sqrt{6}}}\ket{1}$, which ensures $\ket{\alpha=0(\theta=0)}=\ket{0}$.

An $(m=2,k=3)$-EL in a qubit is generated by
\begin{equation}
    \ket{0_{\bm n}}=\tfrac{1}{2\sqrt{2}}\times\left\{\begin{array}{ll}
    \ket{+}+\ket{v},&{\bm n=(0,0)}, \\ \ket{-}-\ket{v}_\perp,&\bm n=(0,1), \\
    \ket{-}+\ket{v}_\perp,&\bm n=(1,0), \\
    \ket{+}-\ket{v},&\bm n=(1,1),\\
    0, & \text{other }\bm n\in\mathbb{Z}^2\end{array}\right.
\end{equation}
\\
where $\ket{\pm}=(\ket{0}\pm \ket{1})/\sqrt{2}$, $\ket{v}=\tfrac{1}{\sqrt{3}}\ket{-} - \tfrac{1}{\sqrt 6}{\ket{+}}$, $\ket{v}_\perp=\tfrac{1}{\sqrt{3}}\ket{+} + \tfrac{1}{\sqrt 6}{\ket{-}}$.

Finding $(m=1,k)$-ELs for higher $k$ and $d$ would prove our claim that $k$\nobreakdash-HSE is achievable with periodic, time-continuous drives. Nevertheless, we note that the number of terms in Eq.~\eqref{eq:kHSEconstraints} grows exponentially with $k$, which poses an obstacle for numerical solutions. Analytical  understanding of the structure of $(m,k)$-ELs is necessary, and a direction we leave open.

\bibliography{references}

\begin{thebibliography}{100}%
\makeatletter
\providecommand \@ifxundefined [1]{%
 \@ifx{#1\undefined}
}%
\providecommand \@ifnum [1]{%
 \ifnum #1\expandafter \@firstoftwo
 \else \expandafter \@secondoftwo
 \fi
}%
\providecommand \@ifx [1]{%
 \ifx #1\expandafter \@firstoftwo
 \else \expandafter \@secondoftwo
 \fi
}%
\providecommand \natexlab [1]{#1}%
\providecommand \enquote  [1]{``#1''}%
\providecommand \bibnamefont  [1]{#1}%
\providecommand \bibfnamefont [1]{#1}%
\providecommand \citenamefont [1]{#1}%
\providecommand \href@noop [0]{\@secondoftwo}%
\providecommand \href [0]{\begingroup \@sanitize@url \@href}%
\providecommand \@href[1]{\@@startlink{#1}\@@href}%
\providecommand \@@href[1]{\endgroup#1\@@endlink}%
\providecommand \@sanitize@url [0]{\catcode `\\12\catcode `\$12\catcode `\&12\catcode `\#12\catcode `\^12\catcode `\_12\catcode `\%12\relax}%
\providecommand \@@startlink[1]{}%
\providecommand \@@endlink[0]{}%
\providecommand \url  [0]{\begingroup\@sanitize@url \@url }%
\providecommand \@url [1]{\endgroup\@href {#1}{\urlprefix }}%
\providecommand \urlprefix  [0]{URL }%
\providecommand \Eprint [0]{\href }%
\providecommand \doibase [0]{https://doi.org/}%
\providecommand \selectlanguage [0]{\@gobble}%
\providecommand \bibinfo  [0]{\@secondoftwo}%
\providecommand \bibfield  [0]{\@secondoftwo}%
\providecommand \translation [1]{[#1]}%
\providecommand \BibitemOpen [0]{}%
\providecommand \bibitemStop [0]{}%
\providecommand \bibitemNoStop [0]{.\EOS\space}%
\providecommand \EOS [0]{\spacefactor3000\relax}%
\providecommand \BibitemShut  [1]{\csname bibitem#1\endcsname}%
\let\auto@bib@innerbib\@empty
%</preamble>
\bibitem [{\citenamefont {Shnirel'man}(1974)}]{Shnirelman1973}%
  \BibitemOpen
  \bibfield  {author} {\bibinfo {author} {\bibfnamefont {A.~I.}\ \bibnamefont {Shnirel'man}},\ }\bibfield  {title} {\bibinfo {title} {Ergodic properties of eigenfunctions},\ }\href {http://www.mathnet.ru/php/archive.phtml?wshow=paper&jrnid=rm&paperid=4463&option_lang=eng} {\bibfield  {journal} {\bibinfo  {journal} {Uspekhi Mat. Nauk,}\ }\textbf {\bibinfo {volume} {29}},\ \bibinfo {pages} {181} (\bibinfo {year} {1974})}\BibitemShut {NoStop}%
\bibitem [{\citenamefont {Sunada}(1997)}]{Sunada1997}%
  \BibitemOpen
  \bibfield  {author} {\bibinfo {author} {\bibfnamefont {T.}~\bibnamefont {Sunada}},\ }\bibfield  {title} {\bibinfo {title} {Quantum ergodicity},\ }in\ \href@noop {} {\emph {\bibinfo {booktitle} {Progress in Inverse Spectral Geometry}}},\ \bibinfo {editor} {edited by\ \bibinfo {editor} {\bibfnamefont {S.~I.}\ \bibnamefont {Andersson}}\ and\ \bibinfo {editor} {\bibfnamefont {M.~L.}\ \bibnamefont {Lapidus}}}\ (\bibinfo  {publisher} {Birkh{\"a}user Basel},\ \bibinfo {address} {Basel},\ \bibinfo {year} {1997})\ pp.\ \bibinfo {pages} {175--196}\BibitemShut {NoStop}%
\bibitem [{\citenamefont {Berry}(1977)}]{Berry1977}%
  \BibitemOpen
  \bibfield  {author} {\bibinfo {author} {\bibfnamefont {M.~V.}\ \bibnamefont {Berry}},\ }\bibfield  {title} {\bibinfo {title} {Regular and irregular semiclassical wavefunctions},\ }\href {https://doi.org/10.1088/0305-4470/10/12/016} {\bibfield  {journal} {\bibinfo  {journal} {J. Phys. A}\ }\textbf {\bibinfo {volume} {10}},\ \bibinfo {pages} {2083} (\bibinfo {year} {1977})}\BibitemShut {NoStop}%
\bibitem [{\citenamefont {Srednicki}(1994)}]{Srednicki1994}%
  \BibitemOpen
  \bibfield  {author} {\bibinfo {author} {\bibfnamefont {M.}~\bibnamefont {Srednicki}},\ }\bibfield  {title} {\bibinfo {title} {Chaos and quantum thermalization},\ }\href {https://doi.org/10.1103/PhysRevE.50.888} {\bibfield  {journal} {\bibinfo  {journal} {Phys. Rev. E}\ }\textbf {\bibinfo {volume} {50}},\ \bibinfo {pages} {888} (\bibinfo {year} {1994})}\BibitemShut {NoStop}%
\bibitem [{\citenamefont {Bohigas}\ \emph {et~al.}(1984)\citenamefont {Bohigas}, \citenamefont {Giannoni},\ and\ \citenamefont {Schmit}}]{Bohigas1984}%
  \BibitemOpen
  \bibfield  {author} {\bibinfo {author} {\bibfnamefont {O.}~\bibnamefont {Bohigas}}, \bibinfo {author} {\bibfnamefont {M.~J.}\ \bibnamefont {Giannoni}},\ and\ \bibinfo {author} {\bibfnamefont {C.}~\bibnamefont {Schmit}},\ }\bibfield  {title} {\bibinfo {title} {Characterization of chaotic quantum spectra and universality of level fluctuation laws},\ }\href {https://doi.org/10.1103/PhysRevLett.52.1} {\bibfield  {journal} {\bibinfo  {journal} {Phys. Rev. Lett.}\ }\textbf {\bibinfo {volume} {52}},\ \bibinfo {pages} {1} (\bibinfo {year} {1984})}\BibitemShut {NoStop}%
\bibitem [{\citenamefont {Fisher}\ \emph {et~al.}(2023)\citenamefont {Fisher}, \citenamefont {Khemani}, \citenamefont {Nahum},\ and\ \citenamefont {Vijay}}]{Fisher2023}%
  \BibitemOpen
  \bibfield  {author} {\bibinfo {author} {\bibfnamefont {M.~P.}\ \bibnamefont {Fisher}}, \bibinfo {author} {\bibfnamefont {V.}~\bibnamefont {Khemani}}, \bibinfo {author} {\bibfnamefont {A.}~\bibnamefont {Nahum}},\ and\ \bibinfo {author} {\bibfnamefont {S.}~\bibnamefont {Vijay}},\ }\bibfield  {title} {\bibinfo {title} {Random quantum circuits},\ }\href {https://doi.org/10.1146/annurev-conmatphys-031720-030658} {\bibfield  {journal} {\bibinfo  {journal} {Annu. Rev. Condens. Matter Phys}\ }\textbf {\bibinfo {volume} {14}},\ \bibinfo {pages} {335–379} (\bibinfo {year} {2023})}\BibitemShut {NoStop}%
\bibitem [{\citenamefont {Pilatowsky-Cameo}\ \emph {et~al.}(2023)\citenamefont {Pilatowsky-Cameo}, \citenamefont {Dag}, \citenamefont {Ho},\ and\ \citenamefont {Choi}}]{Pilatowsky2023}%
  \BibitemOpen
  \bibfield  {author} {\bibinfo {author} {\bibfnamefont {S.}~\bibnamefont {Pilatowsky-Cameo}}, \bibinfo {author} {\bibfnamefont {C.~B.}\ \bibnamefont {Dag}}, \bibinfo {author} {\bibfnamefont {W.~W.}\ \bibnamefont {Ho}},\ and\ \bibinfo {author} {\bibfnamefont {S.}~\bibnamefont {Choi}},\ }\bibfield  {title} {\bibinfo {title} {Complete {H}ilbert-space ergodicity in quantum dynamics of generalized {F}ibonacci drives},\ }\href {https://doi.org/10.1103/PhysRevLett.131.250401} {\bibfield  {journal} {\bibinfo  {journal} {Phys. Rev. Lett.}\ }\textbf {\bibinfo {volume} {131}},\ \bibinfo {pages} {250401} (\bibinfo {year} {2023})}\BibitemShut {NoStop}%
\bibitem [{\citenamefont {Ho}\ \emph {et~al.}(1983)\citenamefont {Ho}, \citenamefont {Chu},\ and\ \citenamefont {Tietz}}]{Ho1983}%
  \BibitemOpen
  \bibfield  {author} {\bibinfo {author} {\bibfnamefont {T.-S.}\ \bibnamefont {Ho}}, \bibinfo {author} {\bibfnamefont {S.-I.}\ \bibnamefont {Chu}},\ and\ \bibinfo {author} {\bibfnamefont {J.~V.}\ \bibnamefont {Tietz}},\ }\bibfield  {title} {\bibinfo {title} {Semiclassical many-mode {F}loquet theory},\ }\href {https://doi.org/10.1016/0009-2614(83)80732-5} {\bibfield  {journal} {\bibinfo  {journal} {Chemical Physics Letters}\ }\textbf {\bibinfo {volume} {96}},\ \bibinfo {pages} {464} (\bibinfo {year} {1983})}\BibitemShut {NoStop}%
\bibitem [{\citenamefont {Luck}\ \emph {et~al.}(1988)\citenamefont {Luck}, \citenamefont {Orland},\ and\ \citenamefont {Smilansky}}]{Luck1988}%
  \BibitemOpen
  \bibfield  {author} {\bibinfo {author} {\bibfnamefont {J.~M.}\ \bibnamefont {Luck}}, \bibinfo {author} {\bibfnamefont {H.}~\bibnamefont {Orland}},\ and\ \bibinfo {author} {\bibfnamefont {U.}~\bibnamefont {Smilansky}},\ }\bibfield  {title} {\bibinfo {title} {On the response of a two-level quantum system to a class of time-dependent quasiperiodic perturbations},\ }\href {https://doi.org/10.1007/bf01014213} {\bibfield  {journal} {\bibinfo  {journal} {J. Stat. Phys.}\ }\textbf {\bibinfo {volume} {53}},\ \bibinfo {pages} {551} (\bibinfo {year} {1988})}\BibitemShut {NoStop}%
\bibitem [{\citenamefont {Casati}\ \emph {et~al.}(1989)\citenamefont {Casati}, \citenamefont {Guarneri},\ and\ \citenamefont {Shepelyansky}}]{Casati1989}%
  \BibitemOpen
  \bibfield  {author} {\bibinfo {author} {\bibfnamefont {G.}~\bibnamefont {Casati}}, \bibinfo {author} {\bibfnamefont {I.}~\bibnamefont {Guarneri}},\ and\ \bibinfo {author} {\bibfnamefont {D.~L.}\ \bibnamefont {Shepelyansky}},\ }\bibfield  {title} {\bibinfo {title} {{A}nderson transition in a one-dimensional system with three incommensurate frequencies},\ }\href {https://doi.org/10.1103/PhysRevLett.62.345} {\bibfield  {journal} {\bibinfo  {journal} {Phys. Rev. Lett.}\ }\textbf {\bibinfo {volume} {62}},\ \bibinfo {pages} {345} (\bibinfo {year} {1989})}\BibitemShut {NoStop}%
\bibitem [{\citenamefont {Feudel}\ \emph {et~al.}(1995)\citenamefont {Feudel}, \citenamefont {Pikovsky},\ and\ \citenamefont {Zaks}}]{Feudel1995}%
  \BibitemOpen
  \bibfield  {author} {\bibinfo {author} {\bibfnamefont {U.}~\bibnamefont {Feudel}}, \bibinfo {author} {\bibfnamefont {A.~S.}\ \bibnamefont {Pikovsky}},\ and\ \bibinfo {author} {\bibfnamefont {M.~A.}\ \bibnamefont {Zaks}},\ }\bibfield  {title} {\bibinfo {title} {Correlation properties of a quasiperiodically forced two-level system},\ }\href {https://doi.org/10.1103/PhysRevE.51.1762} {\bibfield  {journal} {\bibinfo  {journal} {Phys. Rev. E}\ }\textbf {\bibinfo {volume} {51}},\ \bibinfo {pages} {1762} (\bibinfo {year} {1995})}\BibitemShut {NoStop}%
\bibitem [{\citenamefont {Bambusi}\ and\ \citenamefont {Graffi}(2001)}]{Bambusi2001}%
  \BibitemOpen
  \bibfield  {author} {\bibinfo {author} {\bibfnamefont {D.}~\bibnamefont {Bambusi}}\ and\ \bibinfo {author} {\bibfnamefont {S.}~\bibnamefont {Graffi}},\ }\bibfield  {title} {\bibinfo {title} {Time quasi-periodic unbounded perturbations of {S}chr\"{o}dinger operators and {KAM} methods},\ }\href {https://doi.org/10.1007/s002200100426} {\bibfield  {journal} {\bibinfo  {journal} {Communications in Mathematical Physics}\ }\textbf {\bibinfo {volume} {219}},\ \bibinfo {pages} {465–480} (\bibinfo {year} {2001})}\BibitemShut {NoStop}%
\bibitem [{\citenamefont {Gentile}(2003)}]{Gentile2003}%
  \BibitemOpen
  \bibfield  {author} {\bibinfo {author} {\bibfnamefont {G.}~\bibnamefont {Gentile}},\ }\bibfield  {title} {\bibinfo {title} {Quasi-periodic solutions for two-level systems},\ }\href {https://doi.org/10.1007/s00220-003-0943-0} {\bibfield  {journal} {\bibinfo  {journal} {Communications in Mathematical Physics}\ }\textbf {\bibinfo {volume} {242}},\ \bibinfo {pages} {221–250} (\bibinfo {year} {2003})}\BibitemShut {NoStop}%
\bibitem [{\citenamefont {Chu}\ and\ \citenamefont {Telnov}(2004)}]{Chu2004}%
  \BibitemOpen
  \bibfield  {author} {\bibinfo {author} {\bibfnamefont {S.-I.}\ \bibnamefont {Chu}}\ and\ \bibinfo {author} {\bibfnamefont {D.~A.}\ \bibnamefont {Telnov}},\ }\bibfield  {title} {\bibinfo {title} {Beyond the {F}loquet theorem: generalized {F}loquet formalisms and quasienergy methods for atomic and molecular multiphoton processes in intense laser fields},\ }\href {https://doi.org/10.1016/j.physrep.2003.10.001} {\bibfield  {journal} {\bibinfo  {journal} {Physics Reports}\ }\textbf {\bibinfo {volume} {390}},\ \bibinfo {pages} {1–131} (\bibinfo {year} {2004})}\BibitemShut {NoStop}%
\bibitem [{\citenamefont {Gommers}\ \emph {et~al.}(2006)\citenamefont {Gommers}, \citenamefont {Denisov},\ and\ \citenamefont {Renzoni}}]{Gommers2006}%
  \BibitemOpen
  \bibfield  {author} {\bibinfo {author} {\bibfnamefont {R.}~\bibnamefont {Gommers}}, \bibinfo {author} {\bibfnamefont {S.}~\bibnamefont {Denisov}},\ and\ \bibinfo {author} {\bibfnamefont {F.}~\bibnamefont {Renzoni}},\ }\bibfield  {title} {\bibinfo {title} {Quasiperiodically driven ratchets for cold atoms},\ }\href {https://doi.org/10.1103/PhysRevLett.96.240604} {\bibfield  {journal} {\bibinfo  {journal} {Phys. Rev. Lett.}\ }\textbf {\bibinfo {volume} {96}},\ \bibinfo {pages} {240604} (\bibinfo {year} {2006})}\BibitemShut {NoStop}%
\bibitem [{\citenamefont {Chab\'e}\ \emph {et~al.}(2008)\citenamefont {Chab\'e}, \citenamefont {Lemari\'e}, \citenamefont {Gr\'emaud}, \citenamefont {Delande}, \citenamefont {Szriftgiser},\ and\ \citenamefont {Garreau}}]{Chabe2008}%
  \BibitemOpen
  \bibfield  {author} {\bibinfo {author} {\bibfnamefont {J.}~\bibnamefont {Chab\'e}}, \bibinfo {author} {\bibfnamefont {G.}~\bibnamefont {Lemari\'e}}, \bibinfo {author} {\bibfnamefont {B.}~\bibnamefont {Gr\'emaud}}, \bibinfo {author} {\bibfnamefont {D.}~\bibnamefont {Delande}}, \bibinfo {author} {\bibfnamefont {P.}~\bibnamefont {Szriftgiser}},\ and\ \bibinfo {author} {\bibfnamefont {J.~C.}\ \bibnamefont {Garreau}},\ }\bibfield  {title} {\bibinfo {title} {Experimental observation of the {A}nderson metal-insulator transition with atomic matter waves},\ }\href {https://doi.org/10.1103/PhysRevLett.101.255702} {\bibfield  {journal} {\bibinfo  {journal} {Phys. Rev. Lett.}\ }\textbf {\bibinfo {volume} {101}},\ \bibinfo {pages} {255702} (\bibinfo {year} {2008})}\BibitemShut {NoStop}%
\bibitem [{\citenamefont {Jauslin}\ and\ \citenamefont {Lebowitz}(1991)}]{Jauslin1991}%
  \BibitemOpen
  \bibfield  {author} {\bibinfo {author} {\bibfnamefont {H.~R.}\ \bibnamefont {Jauslin}}\ and\ \bibinfo {author} {\bibfnamefont {J.~L.}\ \bibnamefont {Lebowitz}},\ }\bibfield  {title} {\bibinfo {title} {Spectral and stability aspects of quantum chaos},\ }\href {https://doi.org/10.1063/1.165809} {\bibfield  {journal} {\bibinfo  {journal} {Chaos}\ }\textbf {\bibinfo {volume} {1}},\ \bibinfo {pages} {114} (\bibinfo {year} {1991})}\BibitemShut {NoStop}%
\bibitem [{\citenamefont {Blekher}\ \emph {et~al.}(1992)\citenamefont {Blekher}, \citenamefont {Jauslin},\ and\ \citenamefont {Lebowitz}}]{Blekher1992}%
  \BibitemOpen
  \bibfield  {author} {\bibinfo {author} {\bibfnamefont {P.~M.}\ \bibnamefont {Blekher}}, \bibinfo {author} {\bibfnamefont {H.~R.}\ \bibnamefont {Jauslin}},\ and\ \bibinfo {author} {\bibfnamefont {J.~L.}\ \bibnamefont {Lebowitz}},\ }\bibfield  {title} {\bibinfo {title} {{F}loquet spectrum for two-level systems in quasiperiodic time-dependent fields},\ }\href {https://doi.org/10.1007/bf01048846} {\bibfield  {journal} {\bibinfo  {journal} {J. Stat. Phys.}\ }\textbf {\bibinfo {volume} {68}},\ \bibinfo {pages} {271} (\bibinfo {year} {1992})}\BibitemShut {NoStop}%
\bibitem [{\citenamefont {Roberts}\ and\ \citenamefont {Yoshida}(2017)}]{Roberts2017}%
  \BibitemOpen
  \bibfield  {author} {\bibinfo {author} {\bibfnamefont {D.~A.}\ \bibnamefont {Roberts}}\ and\ \bibinfo {author} {\bibfnamefont {B.}~\bibnamefont {Yoshida}},\ }\bibfield  {title} {\bibinfo {title} {Chaos and complexity by design},\ }\href {https://doi.org/10.1007/jhep04(2017)121} {\bibfield  {journal} {\bibinfo  {journal} {J. High Energ. Phys.}\ }\textbf {\bibinfo {volume} {2017}}\bibinfo  {number} { (4)},\ \bibinfo {pages} {121}}\BibitemShut {NoStop}%
\bibitem [{\citenamefont {Cotler}\ \emph {et~al.}(2017)\citenamefont {Cotler}, \citenamefont {Hunter-Jones}, \citenamefont {Liu},\ and\ \citenamefont {Yoshida}}]{Cotler2017}%
  \BibitemOpen
\bibfield  {number} {  }\bibfield  {author} {\bibinfo {author} {\bibfnamefont {J.}~\bibnamefont {Cotler}}, \bibinfo {author} {\bibfnamefont {N.}~\bibnamefont {Hunter-Jones}}, \bibinfo {author} {\bibfnamefont {J.}~\bibnamefont {Liu}},\ and\ \bibinfo {author} {\bibfnamefont {B.}~\bibnamefont {Yoshida}},\ }\bibfield  {title} {\bibinfo {title} {Chaos, complexity, and random matrices},\ }\bibfield  {journal} {\bibinfo  {journal} {Journal of High Energy Physics}\ }\textbf {\bibinfo {volume} {2017}},\ \href {https://doi.org/10.1007/jhep11(2017)048} {10.1007/jhep11(2017)048} (\bibinfo {year} {2017})\BibitemShut {NoStop}%
\bibitem [{\citenamefont {Vikram}\ and\ \citenamefont {Galitski}(2023)}]{Vikram2023}%
  \BibitemOpen
  \bibfield  {author} {\bibinfo {author} {\bibfnamefont {A.}~\bibnamefont {Vikram}}\ and\ \bibinfo {author} {\bibfnamefont {V.}~\bibnamefont {Galitski}},\ }\bibfield  {title} {\bibinfo {title} {Dynamical quantum ergodicity from energy level statistics},\ }\href {https://doi.org/10.1103/PhysRevResearch.5.033126} {\bibfield  {journal} {\bibinfo  {journal} {Phys. Rev. Res.}\ }\textbf {\bibinfo {volume} {5}},\ \bibinfo {pages} {033126} (\bibinfo {year} {2023})}\BibitemShut {NoStop}%
\bibitem [{\citenamefont {Kaneko}\ \emph {et~al.}(2020)\citenamefont {Kaneko}, \citenamefont {Iyoda},\ and\ \citenamefont {Sagawa}}]{Kaneko2020}%
  \BibitemOpen
  \bibfield  {author} {\bibinfo {author} {\bibfnamefont {K.}~\bibnamefont {Kaneko}}, \bibinfo {author} {\bibfnamefont {E.}~\bibnamefont {Iyoda}},\ and\ \bibinfo {author} {\bibfnamefont {T.}~\bibnamefont {Sagawa}},\ }\bibfield  {title} {\bibinfo {title} {Characterizing complexity of many-body quantum dynamics by higher-order eigenstate thermalization},\ }\href {https://doi.org/10.1103/PhysRevA.101.042126} {\bibfield  {journal} {\bibinfo  {journal} {Phys. Rev. A}\ }\textbf {\bibinfo {volume} {101}},\ \bibinfo {pages} {042126} (\bibinfo {year} {2020})}\BibitemShut {NoStop}%
\bibitem [{\citenamefont {Fava}\ \emph {et~al.}(2023)\citenamefont {Fava}, \citenamefont {Kurchan},\ and\ \citenamefont {Pappalardi}}]{Fava2023}%
  \BibitemOpen
  \bibfield  {author} {\bibinfo {author} {\bibfnamefont {M.}~\bibnamefont {Fava}}, \bibinfo {author} {\bibfnamefont {J.}~\bibnamefont {Kurchan}},\ and\ \bibinfo {author} {\bibfnamefont {S.}~\bibnamefont {Pappalardi}},\ }\href@noop {} {\bibinfo {title} {Designs via free probability}} (\bibinfo {year} {2023}),\ \Eprint {https://arxiv.org/abs/2308.06200} {arXiv:2308.06200 [quant-ph]} \BibitemShut {NoStop}%
\bibitem [{\citenamefont {Shou}\ \emph {et~al.}(2023)\citenamefont {Shou}, \citenamefont {Vikram},\ and\ \citenamefont {Galitski}}]{Shou2023}%
  \BibitemOpen
  \bibfield  {author} {\bibinfo {author} {\bibfnamefont {L.}~\bibnamefont {Shou}}, \bibinfo {author} {\bibfnamefont {A.}~\bibnamefont {Vikram}},\ and\ \bibinfo {author} {\bibfnamefont {V.}~\bibnamefont {Galitski}},\ }\href@noop {} {\bibinfo {title} {Spectral anomalies and broken symmetries in maximally chaotic quantum maps}} (\bibinfo {year} {2023}),\ \Eprint {https://arxiv.org/abs/2312.14067} {arXiv:2312.14067 [quant-ph]} \BibitemShut {NoStop}%
\bibitem [{\citenamefont {Mark}\ \emph {et~al.}(2024)\citenamefont {Mark}, \citenamefont {Surace}, \citenamefont {Elben}, \citenamefont {Shaw}, \citenamefont {Choi}, \citenamefont {Refael}, \citenamefont {Endres},\ and\ \citenamefont {Choi}}]{Mark2024}%
  \BibitemOpen
  \bibfield  {author} {\bibinfo {author} {\bibfnamefont {D.~K.}\ \bibnamefont {Mark}}, \bibinfo {author} {\bibfnamefont {F.}~\bibnamefont {Surace}}, \bibinfo {author} {\bibfnamefont {A.}~\bibnamefont {Elben}}, \bibinfo {author} {\bibfnamefont {A.~L.}\ \bibnamefont {Shaw}}, \bibinfo {author} {\bibfnamefont {J.}~\bibnamefont {Choi}}, \bibinfo {author} {\bibfnamefont {G.}~\bibnamefont {Refael}}, \bibinfo {author} {\bibfnamefont {M.}~\bibnamefont {Endres}},\ and\ \bibinfo {author} {\bibfnamefont {S.}~\bibnamefont {Choi}},\ }\href {https://arxiv.org/abs/2403.11970} {\bibinfo {title} {A maximum entropy principle in deep thermalization and in {H}ilbert-space ergodicity}} (\bibinfo {year} {2024}),\ \Eprint {https://arxiv.org/abs/2403.11970} {arXiv:2403.11970 [quant-ph]} \BibitemShut {NoStop}%
\bibitem [{Note1()}]{Note1}%
  \BibitemOpen
  \bibinfo {note} {It is assumed that there is a well-defined limiting distribution as $t \to \infty $, which may not always hold in certain pathological cases.}\BibitemShut {Stop}%
\bibitem [{\citenamefont {Diestel}\ and\ \citenamefont {Spalsbury}(2014)}]{Diestel2014Haar}%
  \BibitemOpen
  \bibfield  {author} {\bibinfo {author} {\bibfnamefont {J.}~\bibnamefont {Diestel}}\ and\ \bibinfo {author} {\bibfnamefont {A.}~\bibnamefont {Spalsbury}},\ }\href@noop {} {\emph {\bibinfo {title} {The Joys of Haar Measure}}},\ Graduate studies in mathematics\ (\bibinfo  {publisher} {American Mathematical Society},\ \bibinfo {address} {Providence, RI},\ \bibinfo {year} {2014})\BibitemShut {NoStop}%
\bibitem [{\citenamefont {Harrow}(2013)}]{Harrow2013}%
  \BibitemOpen
  \bibfield  {author} {\bibinfo {author} {\bibfnamefont {A.~W.}\ \bibnamefont {Harrow}},\ }\href {https://doi.org/10.48550/ARXIV.1308.6595} {\bibinfo {title} {The church of the symmetric subspace}} (\bibinfo {year} {2013})\BibitemShut {NoStop}%
\bibitem [{\citenamefont {Mele}(2024)}]{Mele2023}%
  \BibitemOpen
  \bibfield  {author} {\bibinfo {author} {\bibfnamefont {A.~A.}\ \bibnamefont {Mele}},\ }\bibfield  {title} {\bibinfo {title} {Introduction to {H}aar measure tools in quantum information: A beginner's tutorial},\ }\href {https://doi.org/10.22331/q-2024-05-08-1340} {\bibfield  {journal} {\bibinfo  {journal} {Quantum}\ }\textbf {\bibinfo {volume} {8}},\ \bibinfo {pages} {1340} (\bibinfo {year} {2024})}\BibitemShut {NoStop}%
\bibitem [{Note2()}]{Note2}%
  \BibitemOpen
  \bibinfo {note} {This is equivalent to Definition~\ref {def:CHSE} because we are considering finite-dimensional quantum systems. Then knowledge of all moments uniquely determines a distribution (this is known as the moment problem in mathematics). This follows from the Weierstrass approximation theorem, which states that polynomials are dense under the uniform norm in the space of continuous functions.}\BibitemShut {Stop}%
\bibitem [{Note3()}]{Note3}%
  \BibitemOpen
  \bibinfo {note} {This can be seen by multiplying by $O^{\otimes k}$ on both sides of Eq.~\protect \eqref {eq:kHSE_cond} and taking the trace.}\BibitemShut {Stop}%
\bibitem [{\citenamefont {Cornfeld}\ \emph {et~al.}(1982)\citenamefont {Cornfeld}, \citenamefont {Fomin},\ and\ \citenamefont {Sinai}}]{Cornfeld1982}%
  \BibitemOpen
  \bibfield  {author} {\bibinfo {author} {\bibfnamefont {I.~P.}\ \bibnamefont {Cornfeld}}, \bibinfo {author} {\bibfnamefont {S.~V.}\ \bibnamefont {Fomin}},\ and\ \bibinfo {author} {\bibfnamefont {Y.~G.}\ \bibnamefont {Sinai}},\ }\href {https://doi.org/10.1007/978-1-4615-6927-5} {\emph {\bibinfo {title} {Ergodic Theory}}}\ (\bibinfo  {publisher} {Springer New York},\ \bibinfo {year} {1982})\BibitemShut {NoStop}%
\bibitem [{Note4()}]{Note4}%
  \BibitemOpen
  \bibinfo {note} {A simple example is the qubit $d=2$ case, where $\protect \operatorname {PU}(2)\protect \cong \protect \text {SO}(3)$ is the group of all rotations of three-dimensional space, which acts by rotating states around the Bloch sphere $\protect \mathbb {P}(\protect \mathbb {C}^2)$.}\BibitemShut {Stop}%
\bibitem [{\citenamefont {Collins}\ and\ \citenamefont {Śniady}(2006)}]{Collins2006}%
  \BibitemOpen
  \bibfield  {author} {\bibinfo {author} {\bibfnamefont {B.}~\bibnamefont {Collins}}\ and\ \bibinfo {author} {\bibfnamefont {P.}~\bibnamefont {Śniady}},\ }\bibfield  {title} {\bibinfo {title} {Integration with respect to the {Haar} measure on unitary, orthogonal and symplectic group},\ }\href {https://doi.org/10.1007/s00220-006-1554-3} {\bibfield  {journal} {\bibinfo  {journal} {Communications in Mathematical Physics}\ }\textbf {\bibinfo {volume} {264}},\ \bibinfo {pages} {773–795} (\bibinfo {year} {2006})}\BibitemShut {NoStop}%
\bibitem [{\citenamefont {Verdeny}\ \emph {et~al.}(2016)\citenamefont {Verdeny}, \citenamefont {Puig},\ and\ \citenamefont {Mintert}}]{Verdeny2016}%
  \BibitemOpen
  \bibfield  {author} {\bibinfo {author} {\bibfnamefont {A.}~\bibnamefont {Verdeny}}, \bibinfo {author} {\bibfnamefont {J.}~\bibnamefont {Puig}},\ and\ \bibinfo {author} {\bibfnamefont {F.}~\bibnamefont {Mintert}},\ }\bibfield  {title} {\bibinfo {title} {Quasi-periodically driven quantum systems},\ }\href {https://doi.org/10.1515/zna-2016-0079} {\bibfield  {journal} {\bibinfo  {journal} {Zeitschrift f\"{u}r Naturforschung A}\ }\textbf {\bibinfo {volume} {71}},\ \bibinfo {pages} {897} (\bibinfo {year} {2016})}\BibitemShut {NoStop}%
\bibitem [{\citenamefont {Nandy}\ \emph {et~al.}(2018)\citenamefont {Nandy}, \citenamefont {Sen},\ and\ \citenamefont {Sen}}]{Nandy2018}%
  \BibitemOpen
  \bibfield  {author} {\bibinfo {author} {\bibfnamefont {S.}~\bibnamefont {Nandy}}, \bibinfo {author} {\bibfnamefont {A.}~\bibnamefont {Sen}},\ and\ \bibinfo {author} {\bibfnamefont {D.}~\bibnamefont {Sen}},\ }\bibfield  {title} {\bibinfo {title} {Steady states of a quasiperiodically driven integrable system},\ }\href {https://doi.org/10.1103/PhysRevB.98.245144} {\bibfield  {journal} {\bibinfo  {journal} {Phys. Rev. B}\ }\textbf {\bibinfo {volume} {98}},\ \bibinfo {pages} {245144} (\bibinfo {year} {2018})}\BibitemShut {NoStop}%
\bibitem [{\citenamefont {Peng}\ and\ \citenamefont {Refael}(2018)}]{Yang2018}%
  \BibitemOpen
  \bibfield  {author} {\bibinfo {author} {\bibfnamefont {Y.}~\bibnamefont {Peng}}\ and\ \bibinfo {author} {\bibfnamefont {G.}~\bibnamefont {Refael}},\ }\bibfield  {title} {\bibinfo {title} {Time-quasiperiodic topological superconductors with {M}ajorana multiplexing},\ }\href {https://doi.org/10.1103/PhysRevB.98.220509} {\bibfield  {journal} {\bibinfo  {journal} {Phys. Rev. B}\ }\textbf {\bibinfo {volume} {98}},\ \bibinfo {pages} {220509(R)} (\bibinfo {year} {2018})}\BibitemShut {NoStop}%
\bibitem [{\citenamefont {Long}\ \emph {et~al.}(2022)\citenamefont {Long}, \citenamefont {Crowley},\ and\ \citenamefont {Chandran}}]{Long2022}%
  \BibitemOpen
  \bibfield  {author} {\bibinfo {author} {\bibfnamefont {D.~M.}\ \bibnamefont {Long}}, \bibinfo {author} {\bibfnamefont {P.~J.~D.}\ \bibnamefont {Crowley}},\ and\ \bibinfo {author} {\bibfnamefont {A.}~\bibnamefont {Chandran}},\ }\bibfield  {title} {\bibinfo {title} {Many-body localization with quasiperiodic driving},\ }\href {https://doi.org/10.1103/PhysRevB.105.144204} {\bibfield  {journal} {\bibinfo  {journal} {Phys. Rev. B}\ }\textbf {\bibinfo {volume} {105}},\ \bibinfo {pages} {144204} (\bibinfo {year} {2022})}\BibitemShut {NoStop}%
\bibitem [{\citenamefont {Martin}\ \emph {et~al.}(2022)\citenamefont {Martin}, \citenamefont {Martin},\ and\ \citenamefont {Agarwal}}]{Tristan2022}%
  \BibitemOpen
  \bibfield  {author} {\bibinfo {author} {\bibfnamefont {T.}~\bibnamefont {Martin}}, \bibinfo {author} {\bibfnamefont {I.}~\bibnamefont {Martin}},\ and\ \bibinfo {author} {\bibfnamefont {K.}~\bibnamefont {Agarwal}},\ }\bibfield  {title} {\bibinfo {title} {Effect of quasiperiodic and random noise on many-body dynamical decoupling protocols},\ }\href {https://doi.org/10.1103/PhysRevB.106.134306} {\bibfield  {journal} {\bibinfo  {journal} {Phys. Rev. B}\ }\textbf {\bibinfo {volume} {106}},\ \bibinfo {pages} {134306} (\bibinfo {year} {2022})}\BibitemShut {NoStop}%
\bibitem [{\citenamefont {Das}\ \emph {et~al.}(2023)\citenamefont {Das}, \citenamefont {Bhakuni}, \citenamefont {Santos},\ and\ \citenamefont {Sharma}}]{Das2023}%
  \BibitemOpen
  \bibfield  {author} {\bibinfo {author} {\bibfnamefont {P.}~\bibnamefont {Das}}, \bibinfo {author} {\bibfnamefont {D.~S.}\ \bibnamefont {Bhakuni}}, \bibinfo {author} {\bibfnamefont {L.~F.}\ \bibnamefont {Santos}},\ and\ \bibinfo {author} {\bibfnamefont {A.}~\bibnamefont {Sharma}},\ }\bibfield  {title} {\bibinfo {title} {Periodically and quasiperiodically driven anisotropic {D}icke model},\ }\href {https://doi.org/10.1103/PhysRevA.108.063716} {\bibfield  {journal} {\bibinfo  {journal} {Phys. Rev. A}\ }\textbf {\bibinfo {volume} {108}},\ \bibinfo {pages} {063716} (\bibinfo {year} {2023})}\BibitemShut {NoStop}%
\bibitem [{\citenamefont {Martin}\ \emph {et~al.}(2017)\citenamefont {Martin}, \citenamefont {Refael},\ and\ \citenamefont {Halperin}}]{Martin2017}%
  \BibitemOpen
  \bibfield  {author} {\bibinfo {author} {\bibfnamefont {I.}~\bibnamefont {Martin}}, \bibinfo {author} {\bibfnamefont {G.}~\bibnamefont {Refael}},\ and\ \bibinfo {author} {\bibfnamefont {B.}~\bibnamefont {Halperin}},\ }\bibfield  {title} {\bibinfo {title} {Topological frequency conversion in strongly driven quantum systems},\ }\href {https://doi.org/10.1103/PhysRevX.7.041008} {\bibfield  {journal} {\bibinfo  {journal} {Phys. Rev. X}\ }\textbf {\bibinfo {volume} {7}},\ \bibinfo {pages} {041008} (\bibinfo {year} {2017})}\BibitemShut {NoStop}%
\bibitem [{\citenamefont {Crowley}\ \emph {et~al.}(2019)\citenamefont {Crowley}, \citenamefont {Martin},\ and\ \citenamefont {Chandran}}]{Crowley2019}%
  \BibitemOpen
  \bibfield  {author} {\bibinfo {author} {\bibfnamefont {P.~J.~D.}\ \bibnamefont {Crowley}}, \bibinfo {author} {\bibfnamefont {I.}~\bibnamefont {Martin}},\ and\ \bibinfo {author} {\bibfnamefont {A.}~\bibnamefont {Chandran}},\ }\bibfield  {title} {\bibinfo {title} {Topological classification of quasiperiodically driven quantum systems},\ }\href {https://doi.org/10.1103/PhysRevB.99.064306} {\bibfield  {journal} {\bibinfo  {journal} {Phys. Rev. B}\ }\textbf {\bibinfo {volume} {99}},\ \bibinfo {pages} {064306} (\bibinfo {year} {2019})}\BibitemShut {NoStop}%
\bibitem [{\citenamefont {Else}\ \emph {et~al.}(2020)\citenamefont {Else}, \citenamefont {Ho},\ and\ \citenamefont {Dumitrescu}}]{Dominic2020}%
  \BibitemOpen
  \bibfield  {author} {\bibinfo {author} {\bibfnamefont {D.~V.}\ \bibnamefont {Else}}, \bibinfo {author} {\bibfnamefont {W.~W.}\ \bibnamefont {Ho}},\ and\ \bibinfo {author} {\bibfnamefont {P.~T.}\ \bibnamefont {Dumitrescu}},\ }\bibfield  {title} {\bibinfo {title} {Long-lived interacting phases of matter protected by multiple time-translation symmetries in quasiperiodically driven systems},\ }\href {https://doi.org/10.1103/PhysRevX.10.021032} {\bibfield  {journal} {\bibinfo  {journal} {Phys. Rev. X}\ }\textbf {\bibinfo {volume} {10}},\ \bibinfo {pages} {021032} (\bibinfo {year} {2020})}\BibitemShut {NoStop}%
\bibitem [{\citenamefont {Dumitrescu}\ \emph {et~al.}(2018)\citenamefont {Dumitrescu}, \citenamefont {Vasseur},\ and\ \citenamefont {Potter}}]{Dumitrescu2018}%
  \BibitemOpen
  \bibfield  {author} {\bibinfo {author} {\bibfnamefont {P.~T.}\ \bibnamefont {Dumitrescu}}, \bibinfo {author} {\bibfnamefont {R.}~\bibnamefont {Vasseur}},\ and\ \bibinfo {author} {\bibfnamefont {A.~C.}\ \bibnamefont {Potter}},\ }\bibfield  {title} {\bibinfo {title} {Logarithmically slow relaxation in quasiperiodically driven random spin chains},\ }\href {https://doi.org/10.1103/PhysRevLett.120.070602} {\bibfield  {journal} {\bibinfo  {journal} {Phys. Rev. Lett.}\ }\textbf {\bibinfo {volume} {120}},\ \bibinfo {pages} {070602} (\bibinfo {year} {2018})}\BibitemShut {NoStop}%
\bibitem [{\citenamefont {Autti}\ \emph {et~al.}(2018)\citenamefont {Autti}, \citenamefont {Eltsov},\ and\ \citenamefont {Volovik}}]{Autti2018}%
  \BibitemOpen
  \bibfield  {author} {\bibinfo {author} {\bibfnamefont {S.}~\bibnamefont {Autti}}, \bibinfo {author} {\bibfnamefont {V.~B.}\ \bibnamefont {Eltsov}},\ and\ \bibinfo {author} {\bibfnamefont {G.~E.}\ \bibnamefont {Volovik}},\ }\bibfield  {title} {\bibinfo {title} {Observation of a time quasicrystal and its transition to a superfluid time crystal},\ }\href {https://doi.org/10.1103/PhysRevLett.120.215301} {\bibfield  {journal} {\bibinfo  {journal} {Phys. Rev. Lett.}\ }\textbf {\bibinfo {volume} {120}},\ \bibinfo {pages} {215301} (\bibinfo {year} {2018})}\BibitemShut {NoStop}%
\bibitem [{\citenamefont {Giergiel}\ \emph {et~al.}(2019)\citenamefont {Giergiel}, \citenamefont {Kuro\ifmmode~\acute{s}\else \'{s}\fi{}},\ and\ \citenamefont {Sacha}}]{Giergiel2019}%
  \BibitemOpen
  \bibfield  {author} {\bibinfo {author} {\bibfnamefont {K.}~\bibnamefont {Giergiel}}, \bibinfo {author} {\bibfnamefont {A.}~\bibnamefont {Kuro\ifmmode~\acute{s}\else \'{s}\fi{}}},\ and\ \bibinfo {author} {\bibfnamefont {K.}~\bibnamefont {Sacha}},\ }\bibfield  {title} {\bibinfo {title} {Discrete time quasicrystals},\ }\href {https://doi.org/10.1103/PhysRevB.99.220303} {\bibfield  {journal} {\bibinfo  {journal} {Phys. Rev. B}\ }\textbf {\bibinfo {volume} {99}},\ \bibinfo {pages} {220303(R)} (\bibinfo {year} {2019})}\BibitemShut {NoStop}%
\bibitem [{\citenamefont {Zhao}\ \emph {et~al.}(2021)\citenamefont {Zhao}, \citenamefont {Mintert}, \citenamefont {Moessner},\ and\ \citenamefont {Knolle}}]{Zhao2021}%
  \BibitemOpen
  \bibfield  {author} {\bibinfo {author} {\bibfnamefont {H.}~\bibnamefont {Zhao}}, \bibinfo {author} {\bibfnamefont {F.}~\bibnamefont {Mintert}}, \bibinfo {author} {\bibfnamefont {R.}~\bibnamefont {Moessner}},\ and\ \bibinfo {author} {\bibfnamefont {J.}~\bibnamefont {Knolle}},\ }\bibfield  {title} {\bibinfo {title} {Random multipolar driving: Tunably slow heating through spectral engineering},\ }\href {https://doi.org/10.1103/PhysRevLett.126.040601} {\bibfield  {journal} {\bibinfo  {journal} {Phys. Rev. Lett.}\ }\textbf {\bibinfo {volume} {126}},\ \bibinfo {pages} {040601} (\bibinfo {year} {2021})}\BibitemShut {NoStop}%
\bibitem [{\citenamefont {Zaletel}\ \emph {et~al.}(2023)\citenamefont {Zaletel}, \citenamefont {Lukin}, \citenamefont {Monroe}, \citenamefont {Nayak}, \citenamefont {Wilczek},\ and\ \citenamefont {Yao}}]{Zalatel2023}%
  \BibitemOpen
  \bibfield  {author} {\bibinfo {author} {\bibfnamefont {M.~P.}\ \bibnamefont {Zaletel}}, \bibinfo {author} {\bibfnamefont {M.}~\bibnamefont {Lukin}}, \bibinfo {author} {\bibfnamefont {C.}~\bibnamefont {Monroe}}, \bibinfo {author} {\bibfnamefont {C.}~\bibnamefont {Nayak}}, \bibinfo {author} {\bibfnamefont {F.}~\bibnamefont {Wilczek}},\ and\ \bibinfo {author} {\bibfnamefont {N.~Y.}\ \bibnamefont {Yao}},\ }\bibfield  {title} {\bibinfo {title} {Colloquium: Quantum and classical discrete time crystals},\ }\href {https://doi.org/10.1103/RevModPhys.95.031001} {\bibfield  {journal} {\bibinfo  {journal} {Rev. Mod. Phys.}\ }\textbf {\bibinfo {volume} {95}},\ \bibinfo {pages} {031001} (\bibinfo {year} {2023})}\BibitemShut {NoStop}%
\bibitem [{Note5()}]{Note5}%
  \BibitemOpen
  \bibinfo {note} {Possibly with countably many pieces.}\BibitemShut {Stop}%
\bibitem [{Note6()}]{Note6}%
  \BibitemOpen
  \bibinfo {note} {This follows from the minimality of $m$. If we had ${\protect \bm {n}\cdot \protect \bm {\omega }}=0$ with some entry $n_j\protect \neq 0$, then we could write $\omega _j=-(\DOTSB \sum@ \slimits@ _{k\protect \neq j}n_k\omega _k)/n_j$, which would allow us to reduce $m$, by writing $\theta _j$ in terms of the other $\theta _k$ in $H(\protect \bm {\theta })$.}\BibitemShut {Stop}%
\bibitem [{\citenamefont {Shirley}(1965)}]{Shirley1965}%
  \BibitemOpen
  \bibfield  {author} {\bibinfo {author} {\bibfnamefont {J.~H.}\ \bibnamefont {Shirley}},\ }\bibfield  {title} {\bibinfo {title} {Solution of the {S}chr\"odinger equation with a {H}amiltonian periodic in time},\ }\href {https://doi.org/10.1103/PhysRev.138.B979} {\bibfield  {journal} {\bibinfo  {journal} {Phys. Rev.}\ }\textbf {\bibinfo {volume} {138}},\ \bibinfo {pages} {B979} (\bibinfo {year} {1965})}\BibitemShut {NoStop}%
\bibitem [{\citenamefont {Murdock}(1978)}]{Murdock1978}%
  \BibitemOpen
  \bibfield  {author} {\bibinfo {author} {\bibfnamefont {J.~A.}\ \bibnamefont {Murdock}},\ }\bibfield  {title} {\bibinfo {title} {On the {F}loquet problem for quasiperiodic systems},\ }\href {https://doi.org/10.1090/s0002-9939-1978-0481275-8} {\bibfield  {journal} {\bibinfo  {journal} {Proc. Am. Math. Soc.}\ }\textbf {\bibinfo {volume} {68}},\ \bibinfo {pages} {179} (\bibinfo {year} {1978})}\BibitemShut {NoStop}%
\bibitem [{\citenamefont {Jorba}\ and\ \citenamefont {Sim\'o}(1992)}]{Jorba1992}%
  \BibitemOpen
  \bibfield  {author} {\bibinfo {author} {\bibfnamefont {A.}~\bibnamefont {Jorba}}\ and\ \bibinfo {author} {\bibfnamefont {C.}~\bibnamefont {Sim\'o}},\ }\bibfield  {title} {\bibinfo {title} {On the reducibility of linear differential equations with quasiperiodic coefficients},\ }\href {https://doi.org/10.1016/0022-0396(92)90107-x} {\bibfield  {journal} {\bibinfo  {journal} {Journal of Differential Equations}\ }\textbf {\bibinfo {volume} {98}},\ \bibinfo {pages} {111–124} (\bibinfo {year} {1992})}\BibitemShut {NoStop}%
\bibitem [{Note7()}]{Note7}%
  \BibitemOpen
  \bibinfo {note} {There is a unitary degree of freedom in this decomposition, as generally one can choose $P(\protect \bm {0})=P_0$ and replace Eq.~\protect \eqref {eq:quasienergydecomp} with $U(t)=P(\protect \bm {\omega }t)e^{-iQ t} P_0^\dagger $. We set $P_0=\protect \mathds {1}$ for simplicity.}\BibitemShut {Stop}%
\bibitem [{Note8()}]{Note8}%
  \BibitemOpen
  \bibinfo {note} {This stems from the (continuous) Kronecker-Weyl theorem (Ref.~\cite {Beck2017}, p. 9), which states that the infinite-time average of any quasiperiodic function always exists, and it is equal to the average of the parent function over the torus, $\mathop {{}\protect \mathbb {E}}_{t\geq 0}[f(t)]=\mathop {{}\protect \mathbb {E}}_{\protect \bm {\theta }\in \protect \mathbb {T}^m}[f(\protect \bm {\theta })]$}\BibitemShut {NoStop}%
\bibitem [{Note9()}]{Note9}%
  \BibitemOpen
  \bibinfo {note} {$\norm {\protect \, \cdot \protect \,}_\infty $ is the usual operator norm, corresponding to the Schatten $\infty $-norm.}\BibitemShut {Stop}%
\bibitem [{Note10()}]{Note10}%
  \BibitemOpen
  \bibinfo {note} {Note that Theorem~\ref {th:01} allows the evolution operator $U(t)$ to change discontinuously in time, thus also encompassing discrete-time dynamics arising, for instance, from a brickwork circuit, in which case the Hamiltonian $H(t)$ is a sequence of Dirac-$\delta $ pulses.}\BibitemShut {Stop}%
\bibitem [{Note11()}]{Note11}%
  \BibitemOpen
  \bibinfo {note} {$E_{\min }(t)$ ($E_{\max }(t)$) is the instantaneous ground (most-excited) state energy.}\BibitemShut {Stop}%
\bibitem [{\citenamefont {Margolus}\ and\ \citenamefont {Levitin}(1998)}]{Margolus1998}%
  \BibitemOpen
  \bibfield  {author} {\bibinfo {author} {\bibfnamefont {N.}~\bibnamefont {Margolus}}\ and\ \bibinfo {author} {\bibfnamefont {L.~B.}\ \bibnamefont {Levitin}},\ }\bibfield  {title} {\bibinfo {title} {The maximum speed of dynamical evolution},\ }\href {https://doi.org/10.1016/s0167-2789(98)00054-2} {\bibfield  {journal} {\bibinfo  {journal} {Physica D: Nonlinear Phenomena}\ }\textbf {\bibinfo {volume} {120}},\ \bibinfo {pages} {188–195} (\bibinfo {year} {1998})}\BibitemShut {NoStop}%
\bibitem [{\citenamefont {Mandelstam}\ and\ \citenamefont {Tamm}(1945)}]{Mandelstam1945}%
  \BibitemOpen
  \bibfield  {author} {\bibinfo {author} {\bibfnamefont {L.}~\bibnamefont {Mandelstam}}\ and\ \bibinfo {author} {\bibfnamefont {I.}~\bibnamefont {Tamm}},\ }\bibfield  {title} {\bibinfo {title} {The uncertainty relation between energy and time in non-relativistic quantum mechanics},\ }\href {https://doi.org/10.1007/978-3-642-74626-0_8} {\bibfield  {journal} {\bibinfo  {journal} {J Phys. USSR}\ }\textbf {\bibinfo {volume} {9}},\ \bibinfo {pages} {249} (\bibinfo {year} {1945})},\ \bibinfo {note} {reprinted in \href{https://doi.org/10.1007/978-3-642-74626-0_8}{ISBN:978-3-642-74626-0 (1991)}}\BibitemShut {NoStop}%
\bibitem [{\citenamefont {Ng}\ \emph {et~al.}(2011)\citenamefont {Ng}, \citenamefont {Lidar},\ and\ \citenamefont {Preskill}}]{Ng2011}%
  \BibitemOpen
  \bibfield  {author} {\bibinfo {author} {\bibfnamefont {H.~K.}\ \bibnamefont {Ng}}, \bibinfo {author} {\bibfnamefont {D.~A.}\ \bibnamefont {Lidar}},\ and\ \bibinfo {author} {\bibfnamefont {J.}~\bibnamefont {Preskill}},\ }\bibfield  {title} {\bibinfo {title} {Combining dynamical decoupling with fault-tolerant quantum computation},\ }\href {https://doi.org/10.1103/PhysRevA.84.012305} {\bibfield  {journal} {\bibinfo  {journal} {Phys. Rev. A}\ }\textbf {\bibinfo {volume} {84}},\ \bibinfo {pages} {012305} (\bibinfo {year} {2011})}\BibitemShut {NoStop}%
\bibitem [{\citenamefont {Marvian}\ and\ \citenamefont {Lidar}(2015)}]{Marvian2015}%
  \BibitemOpen
  \bibfield  {author} {\bibinfo {author} {\bibfnamefont {I.}~\bibnamefont {Marvian}}\ and\ \bibinfo {author} {\bibfnamefont {D.~A.}\ \bibnamefont {Lidar}},\ }\bibfield  {title} {\bibinfo {title} {Quantum speed limits for leakage and decoherence},\ }\href {https://doi.org/10.1103/PhysRevLett.115.210402} {\bibfield  {journal} {\bibinfo  {journal} {Phys. Rev. Lett.}\ }\textbf {\bibinfo {volume} {115}},\ \bibinfo {pages} {210402} (\bibinfo {year} {2015})}\BibitemShut {NoStop}%
\bibitem [{Note12()}]{Note12}%
  \BibitemOpen
  \bibinfo {note} {The smoothness assumption is necessary, as there are examples of nonsmooth, but continuous and surjective maps which increase dimension (e.g. space-filling curves)}\BibitemShut {NoStop}%
\bibitem [{\citenamefont {Dmitri~Burago}(2001)}]{BuragoBook2001}%
  \BibitemOpen
  \bibfield  {author} {\bibinfo {author} {\bibfnamefont {S.~I.}\ \bibnamefont {Dmitri~Burago}, \bibfnamefont {Yuri~Burago}},\ }\href@noop {} {\emph {\bibinfo {title} {A Course in Metric Goemetry}}},\ Gruaduate Studies in Mathematics\ (\bibinfo  {publisher} {American Mathematical Society},\ \bibinfo {address} {Providence, RI},\ \bibinfo {year} {2001})\BibitemShut {NoStop}%
\bibitem [{Note13()}]{Note13}%
  \BibitemOpen
  \bibinfo {note} {Note that average is just a discrete sum because the drives consist purely of kicks at integer times.}\BibitemShut {Stop}%
\bibitem [{Note14()}]{Note14}%
  \BibitemOpen
  \bibinfo {note} {We choose the normalization factor $M$ so to set the thermal fluctuations to unity $(\tr (O^2)-\tr (O)^2)/d=1$}\BibitemShut {NoStop}%
\bibitem [{Note15()}]{Note15}%
  \BibitemOpen
  \bibinfo {note} {Here, we normalize by selecting $M$ such that $(\protect \text {tr}(O^2)-\protect \text {tr}(O)^2)/d^2=1$}\BibitemShut {NoStop}%
\bibitem [{\citenamefont {Elben}\ \emph {et~al.}(2022)\citenamefont {Elben}, \citenamefont {Flammia}, \citenamefont {Huang}, \citenamefont {Kueng}, \citenamefont {Preskill}, \citenamefont {Vermersch},\ and\ \citenamefont {Zoller}}]{Elben2022}%
  \BibitemOpen
  \bibfield  {author} {\bibinfo {author} {\bibfnamefont {A.}~\bibnamefont {Elben}}, \bibinfo {author} {\bibfnamefont {S.~T.}\ \bibnamefont {Flammia}}, \bibinfo {author} {\bibfnamefont {H.-Y.}\ \bibnamefont {Huang}}, \bibinfo {author} {\bibfnamefont {R.}~\bibnamefont {Kueng}}, \bibinfo {author} {\bibfnamefont {J.}~\bibnamefont {Preskill}}, \bibinfo {author} {\bibfnamefont {B.}~\bibnamefont {Vermersch}},\ and\ \bibinfo {author} {\bibfnamefont {P.}~\bibnamefont {Zoller}},\ }\bibfield  {title} {\bibinfo {title} {The randomized measurement toolbox},\ }\href {https://doi.org/10.1038/s42254-022-00535-2} {\bibfield  {journal} {\bibinfo  {journal} {Nat. Rev. Phys.}\ }\textbf {\bibinfo {volume} {5}},\ \bibinfo {pages} {9} (\bibinfo {year} {2022})}\BibitemShut {NoStop}%
\bibitem [{\citenamefont {Hurwitz}(1897)}]{Hurwitz1897}%
  \BibitemOpen
  \bibfield  {author} {\bibinfo {author} {\bibfnamefont {A.}~\bibnamefont {Hurwitz}},\ }\bibfield  {title} {\bibinfo {title} {über die erzeugung der invarianten durch integration},\ }\href {http://eudml.org/doc/58378} {\bibfield  {journal} {\bibinfo  {journal} {Nachrichten von der Gesellschaft der Wissenschaften zu Göttingen, Mathematisch-Physikalische Klasse}\ }\textbf {\bibinfo {volume} {1897}},\ \bibinfo {pages} {71} (\bibinfo {year} {1897})}\BibitemShut {NoStop}%
\bibitem [{\citenamefont {Zyczkowski}\ and\ \citenamefont {Kus}(1994)}]{Zyczkowski1994}%
  \BibitemOpen
  \bibfield  {author} {\bibinfo {author} {\bibfnamefont {K.}~\bibnamefont {Zyczkowski}}\ and\ \bibinfo {author} {\bibfnamefont {M.}~\bibnamefont {Kus}},\ }\bibfield  {title} {\bibinfo {title} {Random unitary matrices},\ }\href {https://doi.org/10.1088/0305-4470/27/12/028} {\bibfield  {journal} {\bibinfo  {journal} {Jour. of Phys. A: Math. and Gen.}\ }\textbf {\bibinfo {volume} {27}},\ \bibinfo {pages} {4235} (\bibinfo {year} {1994})}\BibitemShut {NoStop}%
\bibitem [{\citenamefont {Diaconis}\ and\ \citenamefont {Forrester}(2017)}]{Diaconis2017}%
  \BibitemOpen
  \bibfield  {author} {\bibinfo {author} {\bibfnamefont {P.}~\bibnamefont {Diaconis}}\ and\ \bibinfo {author} {\bibfnamefont {P.~J.}\ \bibnamefont {Forrester}},\ }\bibfield  {title} {\bibinfo {title} {Hurwitz and the origins of random matrix theory in mathematics},\ }\href {https://doi.org/10.1142/s2010326317300017} {\bibfield  {journal} {\bibinfo  {journal} {Random Matrices: Theory and Applications}\ }\textbf {\bibinfo {volume} {06}},\ \bibinfo {pages} {1730001} (\bibinfo {year} {2017})}\BibitemShut {NoStop}%
\bibitem [{\citenamefont {Ramamoorthy}\ \emph {et~al.}(2008)\citenamefont {Ramamoorthy}, \citenamefont {Rajagopal}, \citenamefont {Ruan},\ and\ \citenamefont {Wenzel}}]{Ramamoorthy2008}%
  \BibitemOpen
  \bibfield  {author} {\bibinfo {author} {\bibfnamefont {S.}~\bibnamefont {Ramamoorthy}}, \bibinfo {author} {\bibfnamefont {R.}~\bibnamefont {Rajagopal}}, \bibinfo {author} {\bibfnamefont {Q.}~\bibnamefont {Ruan}},\ and\ \bibinfo {author} {\bibfnamefont {L.}~\bibnamefont {Wenzel}},\ }\bibinfo {title} {Low-discrepancy curves and efficient coverage of space},\ in\ \href {https://doi.org/10.1007/978-3-540-68405-3_13} {\emph {\bibinfo {booktitle} {Algorithmic Foundation of Robotics VII: Selected Contributions of the Seventh International Workshop on the Algorithmic Foundations of Robotics}}},\ \bibinfo {editor} {edited by\ \bibinfo {editor} {\bibfnamefont {S.}~\bibnamefont {Akella}}, \bibinfo {editor} {\bibfnamefont {N.~M.}\ \bibnamefont {Amato}}, \bibinfo {editor} {\bibfnamefont {W.~H.}\ \bibnamefont {Huang}},\ and\ \bibinfo {editor} {\bibfnamefont {B.}~\bibnamefont {Mishra}}}\ (\bibinfo  {publisher} {Springer Berlin Heidelberg},\ \bibinfo {year} {2008})\ pp.\ \bibinfo {pages} {203--218}\BibitemShut {NoStop}%
\bibitem [{\citenamefont {Morokoff}\ and\ \citenamefont {Caflisch}(1995)}]{Morokoff1995}%
  \BibitemOpen
  \bibfield  {author} {\bibinfo {author} {\bibfnamefont {W.~J.}\ \bibnamefont {Morokoff}}\ and\ \bibinfo {author} {\bibfnamefont {R.~E.}\ \bibnamefont {Caflisch}},\ }\bibfield  {title} {\bibinfo {title} {Quasi-monte carlo integration},\ }\href {https://doi.org/10.1006/jcph.1995.1209} {\bibfield  {journal} {\bibinfo  {journal} {Journal of Computational Physics}\ }\textbf {\bibinfo {volume} {122}},\ \bibinfo {pages} {218–230} (\bibinfo {year} {1995})}\BibitemShut {NoStop}%
\bibitem [{\citenamefont {Seymour}\ and\ \citenamefont {Zaslavsky}(1984)}]{Seymour1984}%
  \BibitemOpen
  \bibfield  {author} {\bibinfo {author} {\bibfnamefont {P.}~\bibnamefont {Seymour}}\ and\ \bibinfo {author} {\bibfnamefont {T.}~\bibnamefont {Zaslavsky}},\ }\bibfield  {title} {\bibinfo {title} {Averaging sets: A generalization of mean values and spherical designs},\ }\href {https://doi.org/10.1016/0001-8708(84)90022-7} {\bibfield  {journal} {\bibinfo  {journal} {Advances in Mathematics}\ }\textbf {\bibinfo {volume} {52}},\ \bibinfo {pages} {213} (\bibinfo {year} {1984})}\BibitemShut {NoStop}%
\bibitem [{\citenamefont {Scott}(2008)}]{Scott2008}%
  \BibitemOpen
  \bibfield  {author} {\bibinfo {author} {\bibfnamefont {A.~J.}\ \bibnamefont {Scott}},\ }\bibfield  {title} {\bibinfo {title} {Optimizing quantum process tomography with unitary 2-designs},\ }\href {https://doi.org/10.1088/1751-8113/41/5/055308} {\bibfield  {journal} {\bibinfo  {journal} {J. Phys. A Math. Theor.}\ }\textbf {\bibinfo {volume} {41}},\ \bibinfo {pages} {055308} (\bibinfo {year} {2008})}\BibitemShut {NoStop}%
\bibitem [{Note16()}]{Note16}%
  \BibitemOpen
  \bibinfo {note} {These lower bounds are given by the trace distance traversed by the trajectories in Figs.~\ref {fig:5}a and b (see Corollary \ref {cor:dis_traveled_bound}).}\BibitemShut {Stop}%
\bibitem [{\citenamefont {Cotler}\ \emph {et~al.}(2023)\citenamefont {Cotler}, \citenamefont {Mark}, \citenamefont {Huang}, \citenamefont {Hern\'andez}, \citenamefont {Choi}, \citenamefont {Shaw}, \citenamefont {Endres},\ and\ \citenamefont {Choi}}]{Cotler2023}%
  \BibitemOpen
  \bibfield  {author} {\bibinfo {author} {\bibfnamefont {J.~S.}\ \bibnamefont {Cotler}}, \bibinfo {author} {\bibfnamefont {D.~K.}\ \bibnamefont {Mark}}, \bibinfo {author} {\bibfnamefont {H.-Y.}\ \bibnamefont {Huang}}, \bibinfo {author} {\bibfnamefont {F.}~\bibnamefont {Hern\'andez}}, \bibinfo {author} {\bibfnamefont {J.}~\bibnamefont {Choi}}, \bibinfo {author} {\bibfnamefont {A.~L.}\ \bibnamefont {Shaw}}, \bibinfo {author} {\bibfnamefont {M.}~\bibnamefont {Endres}},\ and\ \bibinfo {author} {\bibfnamefont {S.}~\bibnamefont {Choi}},\ }\bibfield  {title} {\bibinfo {title} {Emergent quantum state designs from individual many-body wave functions},\ }\href {https://doi.org/10.1103/PRXQuantum.4.010311} {\bibfield  {journal} {\bibinfo  {journal} {PRX Quantum}\ }\textbf {\bibinfo {volume} {4}},\ \bibinfo {pages} {010311} (\bibinfo {year} {2023})}\BibitemShut {NoStop}%
\bibitem [{\citenamefont {Claeys}\ and\ \citenamefont {Lamacraft}(2022)}]{Claeys2022}%
  \BibitemOpen
  \bibfield  {author} {\bibinfo {author} {\bibfnamefont {P.~W.}\ \bibnamefont {Claeys}}\ and\ \bibinfo {author} {\bibfnamefont {A.}~\bibnamefont {Lamacraft}},\ }\bibfield  {title} {\bibinfo {title} {Emergent quantum state designs and biunitarity in dual-unitary circuit dynamics},\ }\href {https://doi.org/10.22331/q-2022-06-15-738} {\bibfield  {journal} {\bibinfo  {journal} {Quantum}\ }\textbf {\bibinfo {volume} {6}},\ \bibinfo {pages} {738} (\bibinfo {year} {2022})}\BibitemShut {NoStop}%
\bibitem [{\citenamefont {Wilming}\ and\ \citenamefont {Roth}(2022)}]{Wilming2022}%
  \BibitemOpen
  \bibfield  {author} {\bibinfo {author} {\bibfnamefont {H.}~\bibnamefont {Wilming}}\ and\ \bibinfo {author} {\bibfnamefont {I.}~\bibnamefont {Roth}},\ }\href@noop {} {\bibinfo {title} {High-temperature thermalization implies the emergence of quantum state designs}} (\bibinfo {year} {2022}),\ \Eprint {https://arxiv.org/abs/2202.01669} {arXiv:2202.01669 [quant-ph]} \BibitemShut {NoStop}%
\bibitem [{\citenamefont {Ho}\ and\ \citenamefont {Choi}(2022)}]{Ho2022}%
  \BibitemOpen
  \bibfield  {author} {\bibinfo {author} {\bibfnamefont {W.~W.}\ \bibnamefont {Ho}}\ and\ \bibinfo {author} {\bibfnamefont {S.}~\bibnamefont {Choi}},\ }\bibfield  {title} {\bibinfo {title} {Exact emergent quantum state designs from quantum chaotic dynamics},\ }\href {https://doi.org/10.1103/PhysRevLett.128.060601} {\bibfield  {journal} {\bibinfo  {journal} {Phys. Rev. Lett.}\ }\textbf {\bibinfo {volume} {128}},\ \bibinfo {pages} {060601} (\bibinfo {year} {2022})}\BibitemShut {NoStop}%
\bibitem [{\citenamefont {Ippoliti}\ and\ \citenamefont {Ho}(2022)}]{Ippoliti2022}%
  \BibitemOpen
  \bibfield  {author} {\bibinfo {author} {\bibfnamefont {M.}~\bibnamefont {Ippoliti}}\ and\ \bibinfo {author} {\bibfnamefont {W.~W.}\ \bibnamefont {Ho}},\ }\bibfield  {title} {\bibinfo {title} {Solvable model of deep thermalization with distinct design times},\ }\href {https://doi.org/10.22331/q-2022-12-29-886} {\bibfield  {journal} {\bibinfo  {journal} {Quantum}\ }\textbf {\bibinfo {volume} {6}},\ \bibinfo {pages} {886} (\bibinfo {year} {2022})}\BibitemShut {NoStop}%
\bibitem [{\citenamefont {Choi}\ \emph {et~al.}(2023)\citenamefont {Choi}, \citenamefont {Shaw}, \citenamefont {Madjarov}, \citenamefont {Xie}, \citenamefont {Finkelstein}, \citenamefont {Covey}, \citenamefont {Cotler}, \citenamefont {Mark}, \citenamefont {Huang}, \citenamefont {Kale}, \citenamefont {Pichler}, \citenamefont {Brand{\~{a}}o}, \citenamefont {Choi},\ and\ \citenamefont {Endres}}]{Choi2023}%
  \BibitemOpen
  \bibfield  {author} {\bibinfo {author} {\bibfnamefont {J.}~\bibnamefont {Choi}}, \bibinfo {author} {\bibfnamefont {A.~L.}\ \bibnamefont {Shaw}}, \bibinfo {author} {\bibfnamefont {I.~S.}\ \bibnamefont {Madjarov}}, \bibinfo {author} {\bibfnamefont {X.}~\bibnamefont {Xie}}, \bibinfo {author} {\bibfnamefont {R.}~\bibnamefont {Finkelstein}}, \bibinfo {author} {\bibfnamefont {J.~P.}\ \bibnamefont {Covey}}, \bibinfo {author} {\bibfnamefont {J.~S.}\ \bibnamefont {Cotler}}, \bibinfo {author} {\bibfnamefont {D.~K.}\ \bibnamefont {Mark}}, \bibinfo {author} {\bibfnamefont {H.-Y.}\ \bibnamefont {Huang}}, \bibinfo {author} {\bibfnamefont {A.}~\bibnamefont {Kale}}, \bibinfo {author} {\bibfnamefont {H.}~\bibnamefont {Pichler}}, \bibinfo {author} {\bibfnamefont {F.~G. S.~L.}\ \bibnamefont {Brand{\~{a}}o}}, \bibinfo {author} {\bibfnamefont {S.}~\bibnamefont {Choi}},\ and\ \bibinfo {author} {\bibfnamefont {M.}~\bibnamefont {Endres}},\ }\bibfield  {title} {\bibinfo {title} {Preparing random states and benchmarking with
  many-body quantum chaos},\ }\href {https://doi.org/10.1038/s41586-022-05442-1} {\bibfield  {journal} {\bibinfo  {journal} {Nature}\ }\textbf {\bibinfo {volume} {613}},\ \bibinfo {pages} {468} (\bibinfo {year} {2023})}\BibitemShut {NoStop}%
\bibitem [{\citenamefont {Ippoliti}\ and\ \citenamefont {Ho}(2023)}]{Ippoliti2023}%
  \BibitemOpen
  \bibfield  {author} {\bibinfo {author} {\bibfnamefont {M.}~\bibnamefont {Ippoliti}}\ and\ \bibinfo {author} {\bibfnamefont {W.~W.}\ \bibnamefont {Ho}},\ }\bibfield  {title} {\bibinfo {title} {Dynamical purification and the emergence of quantum state designs from the projected ensemble},\ }\href {https://doi.org/10.1103/PRXQuantum.4.030322} {\bibfield  {journal} {\bibinfo  {journal} {PRX Quantum}\ }\textbf {\bibinfo {volume} {4}},\ \bibinfo {pages} {030322} (\bibinfo {year} {2023})}\BibitemShut {NoStop}%
\bibitem [{\citenamefont {Lucas}\ \emph {et~al.}(2023)\citenamefont {Lucas}, \citenamefont {Piroli}, \citenamefont {De~Nardis},\ and\ \citenamefont {De~Luca}}]{Lucas2023}%
  \BibitemOpen
  \bibfield  {author} {\bibinfo {author} {\bibfnamefont {M.}~\bibnamefont {Lucas}}, \bibinfo {author} {\bibfnamefont {L.}~\bibnamefont {Piroli}}, \bibinfo {author} {\bibfnamefont {J.}~\bibnamefont {De~Nardis}},\ and\ \bibinfo {author} {\bibfnamefont {A.}~\bibnamefont {De~Luca}},\ }\bibfield  {title} {\bibinfo {title} {Generalized deep thermalization for free fermions},\ }\href {https://doi.org/10.1103/PhysRevA.107.032215} {\bibfield  {journal} {\bibinfo  {journal} {Phys. Rev. A}\ }\textbf {\bibinfo {volume} {107}},\ \bibinfo {pages} {032215} (\bibinfo {year} {2023})}\BibitemShut {NoStop}%
\bibitem [{\citenamefont {Shrotriya}\ and\ \citenamefont {Ho}(2023)}]{Shrotriya2023}%
  \BibitemOpen
  \bibfield  {author} {\bibinfo {author} {\bibfnamefont {H.}~\bibnamefont {Shrotriya}}\ and\ \bibinfo {author} {\bibfnamefont {W.~W.}\ \bibnamefont {Ho}},\ }\href@noop {} {\bibinfo {title} {Nonlocality of deep thermalization}} (\bibinfo {year} {2023}),\ \Eprint {https://arxiv.org/abs/2305.08437} {arXiv:2305.08437 [quant-ph]} \BibitemShut {NoStop}%
\bibitem [{\citenamefont {Sambe}(1973)}]{Sambe1973}%
  \BibitemOpen
  \bibfield  {author} {\bibinfo {author} {\bibfnamefont {H.}~\bibnamefont {Sambe}},\ }\bibfield  {title} {\bibinfo {title} {Steady states and quasienergies of a quantum-mechanical system in an oscillating field},\ }\href {https://doi.org/10.1103/PhysRevA.7.2203} {\bibfield  {journal} {\bibinfo  {journal} {Phys. Rev. A}\ }\textbf {\bibinfo {volume} {7}},\ \bibinfo {pages} {2203} (\bibinfo {year} {1973})}\BibitemShut {NoStop}%
\bibitem [{\citenamefont {Parker}\ \emph {et~al.}(2019)\citenamefont {Parker}, \citenamefont {Cao}, \citenamefont {Avdoshkin}, \citenamefont {Scaffidi},\ and\ \citenamefont {Altman}}]{Parker2019}%
  \BibitemOpen
  \bibfield  {author} {\bibinfo {author} {\bibfnamefont {D.~E.}\ \bibnamefont {Parker}}, \bibinfo {author} {\bibfnamefont {X.}~\bibnamefont {Cao}}, \bibinfo {author} {\bibfnamefont {A.}~\bibnamefont {Avdoshkin}}, \bibinfo {author} {\bibfnamefont {T.}~\bibnamefont {Scaffidi}},\ and\ \bibinfo {author} {\bibfnamefont {E.}~\bibnamefont {Altman}},\ }\bibfield  {title} {\bibinfo {title} {A universal operator growth hypothesis},\ }\href {https://doi.org/10.1103/PhysRevX.9.041017} {\bibfield  {journal} {\bibinfo  {journal} {Phys. Rev. X}\ }\textbf {\bibinfo {volume} {9}},\ \bibinfo {pages} {041017} (\bibinfo {year} {2019})}\BibitemShut {NoStop}%
\bibitem [{\citenamefont {Brown}\ and\ \citenamefont {Susskind}(2018)}]{Brown2018}%
  \BibitemOpen
  \bibfield  {author} {\bibinfo {author} {\bibfnamefont {A.~R.}\ \bibnamefont {Brown}}\ and\ \bibinfo {author} {\bibfnamefont {L.}~\bibnamefont {Susskind}},\ }\bibfield  {title} {\bibinfo {title} {Second law of quantum complexity},\ }\href {https://doi.org/10.1103/physrevd.97.086015} {\bibfield  {journal} {\bibinfo  {journal} {Phys. Rev. D}\ }\textbf {\bibinfo {volume} {97}},\ \bibinfo {pages} {086015} (\bibinfo {year} {2018})}\BibitemShut {NoStop}%
\bibitem [{\citenamefont {Haferkamp}\ \emph {et~al.}(2022)\citenamefont {Haferkamp}, \citenamefont {Faist}, \citenamefont {Kothakonda}, \citenamefont {Eisert},\ and\ \citenamefont {Yunger~Halpern}}]{Haferkamp2022}%
  \BibitemOpen
  \bibfield  {author} {\bibinfo {author} {\bibfnamefont {J.}~\bibnamefont {Haferkamp}}, \bibinfo {author} {\bibfnamefont {P.}~\bibnamefont {Faist}}, \bibinfo {author} {\bibfnamefont {N.~B.~T.}\ \bibnamefont {Kothakonda}}, \bibinfo {author} {\bibfnamefont {J.}~\bibnamefont {Eisert}},\ and\ \bibinfo {author} {\bibfnamefont {N.}~\bibnamefont {Yunger~Halpern}},\ }\bibfield  {title} {\bibinfo {title} {Linear growth of quantum circuit complexity},\ }\href {https://doi.org/10.1038/s41567-022-01539-6} {\bibfield  {journal} {\bibinfo  {journal} {Nature Physics}\ }\textbf {\bibinfo {volume} {18}},\ \bibinfo {pages} {528–532} (\bibinfo {year} {2022})}\BibitemShut {NoStop}%
\bibitem [{\citenamefont {Camilo}\ and\ \citenamefont {Teixeira}(2020)}]{Giancarlo2020}%
  \BibitemOpen
  \bibfield  {author} {\bibinfo {author} {\bibfnamefont {G.}~\bibnamefont {Camilo}}\ and\ \bibinfo {author} {\bibfnamefont {D.}~\bibnamefont {Teixeira}},\ }\bibfield  {title} {\bibinfo {title} {Complexity and floquet dynamics: Nonequilibrium ising phase transitions},\ }\href {https://doi.org/10.1103/PhysRevB.102.174304} {\bibfield  {journal} {\bibinfo  {journal} {Phys. Rev. B}\ }\textbf {\bibinfo {volume} {102}},\ \bibinfo {pages} {174304} (\bibinfo {year} {2020})}\BibitemShut {NoStop}%
\bibitem [{\citenamefont {Suchsland}\ \emph {et~al.}(2023)\citenamefont {Suchsland}, \citenamefont {Moessner},\ and\ \citenamefont {Claeys}}]{Suchsland2023}%
  \BibitemOpen
  \bibfield  {author} {\bibinfo {author} {\bibfnamefont {P.}~\bibnamefont {Suchsland}}, \bibinfo {author} {\bibfnamefont {R.}~\bibnamefont {Moessner}},\ and\ \bibinfo {author} {\bibfnamefont {P.~W.}\ \bibnamefont {Claeys}},\ }\href@noop {} {\bibinfo {title} {Krylov complexity and trotter transitions in unitary circuit dynamics}} (\bibinfo {year} {2023}),\ \Eprint {https://arxiv.org/abs/2308.03851} {arXiv:2308.03851 [quant-ph]} \BibitemShut {NoStop}%
\bibitem [{\citenamefont {Rains}(1998)}]{Rains1998}%
  \BibitemOpen
  \bibfield  {author} {\bibinfo {author} {\bibfnamefont {E.~M.}\ \bibnamefont {Rains}},\ }\bibfield  {title} {\bibinfo {title} {Increasing subsequences and the classical groups},\ }\bibfield  {journal} {\bibinfo  {journal} {The Electronic Journal of Combinatorics}\ }\textbf {\bibinfo {volume} {5}},\ \href {https://doi.org/10.37236/1350} {10.37236/1350} (\bibinfo {year} {1998})\BibitemShut {NoStop}%
\bibitem [{\citenamefont {Regev}(1981)}]{Regev1981}%
  \BibitemOpen
  \bibfield  {author} {\bibinfo {author} {\bibfnamefont {A.}~\bibnamefont {Regev}},\ }\bibfield  {title} {\bibinfo {title} {Asymptotic values for degrees associated with strips of young diagrams},\ }\href {https://doi.org/10.1016/0001-8708(81)90012-8} {\bibfield  {journal} {\bibinfo  {journal} {Advances in Mathematics}\ }\textbf {\bibinfo {volume} {41}},\ \bibinfo {pages} {115} (\bibinfo {year} {1981})}\BibitemShut {NoStop}%
\bibitem [{\citenamefont {Stanley}(2006)}]{Stanley2006}%
  \BibitemOpen
  \bibfield  {author} {\bibinfo {author} {\bibfnamefont {R.~P.}\ \bibnamefont {Stanley}},\ }\bibfield  {title} {\bibinfo {title} {Increasing and decreasing subsequences and their variants},\ }in\ \href {https://www.mathunion.org/fileadmin/ICM/Proceedings/ICM2006.1/ICM2006.1.ocr.pdf} {\emph {\bibinfo {booktitle} {{Proceedings of the International Congress of Mathematicians}}}},\ Vol.~\bibinfo {volume} {1}\ (\bibinfo {year} {2006})\ pp.\ \bibinfo {pages} {545--577}\BibitemShut {NoStop}%
\bibitem [{\citenamefont {Oszmaniec}\ \emph {et~al.}(2020)\citenamefont {Oszmaniec}, \citenamefont {Sawicki},\ and\ \citenamefont {Horodecki}}]{Oszmaniec2020}%
  \BibitemOpen
  \bibfield  {author} {\bibinfo {author} {\bibfnamefont {M.}~\bibnamefont {Oszmaniec}}, \bibinfo {author} {\bibfnamefont {A.}~\bibnamefont {Sawicki}},\ and\ \bibinfo {author} {\bibfnamefont {M.}~\bibnamefont {Horodecki}},\ }\href@noop {} {\bibinfo {title} {Epsilon-nets, unitary designs and random quantum circuits}} (\bibinfo {year} {2020}),\ \Eprint {https://arxiv.org/abs/2007.10885} {arXiv:2007.10885 [quant-ph]} \BibitemShut {NoStop}%
\bibitem [{Note17()}]{Note17}%
  \BibitemOpen
  \bibinfo {note} {Proposition \ref {prop:speed_limit_prop} is a special case of Lemma~2 of Ref.~\cite {Ng2011}, taking one of the Hamiltonians to be zero}\BibitemShut {NoStop}%
\bibitem [{Note18()}]{Note18}%
  \BibitemOpen
  \bibinfo {note} {Balls are taken with respect to the trace distance, i.e., $B({\phi ,r})=\{\psi \in \protect \mathbb {P}(\protect \mathbb {C}^d) | D(\psi ,\phi )<r\}$.}\BibitemShut {Stop}%
\bibitem [{\citenamefont {Kus}\ \emph {et~al.}(1988)\citenamefont {Kus}, \citenamefont {Mostowski},\ and\ \citenamefont {Haake}}]{Kus1988}%
  \BibitemOpen
  \bibfield  {author} {\bibinfo {author} {\bibfnamefont {M.}~\bibnamefont {Kus}}, \bibinfo {author} {\bibfnamefont {J.}~\bibnamefont {Mostowski}},\ and\ \bibinfo {author} {\bibfnamefont {F.}~\bibnamefont {Haake}},\ }\bibfield  {title} {\bibinfo {title} {Universality of eigenvector statistics of kicked tops of different symmetries},\ }\href {https://doi.org/10.1088/0305-4470/21/22/006} {\bibfield  {journal} {\bibinfo  {journal} {Jour. Phys. A: Math. and Gen.}\ }\textbf {\bibinfo {volume} {21}},\ \bibinfo {pages} {L1073} (\bibinfo {year} {1988})}\BibitemShut {NoStop}%
\bibitem [{\citenamefont {Oszmaniec}\ \emph {et~al.}(2024)\citenamefont {Oszmaniec}, \citenamefont {Kotowski}, \citenamefont {Horodecki},\ and\ \citenamefont {Hunter-Jones}}]{Oszmaniec2024}%
  \BibitemOpen
  \bibfield  {author} {\bibinfo {author} {\bibfnamefont {M.}~\bibnamefont {Oszmaniec}}, \bibinfo {author} {\bibfnamefont {M.}~\bibnamefont {Kotowski}}, \bibinfo {author} {\bibfnamefont {M.}~\bibnamefont {Horodecki}},\ and\ \bibinfo {author} {\bibfnamefont {N.}~\bibnamefont {Hunter-Jones}},\ }\href {https://arxiv.org/abs/2205.09734} {\bibinfo {title} {Saturation and recurrence of quantum complexity in random local quantum dynamics}} (\bibinfo {year} {2024}),\ \Eprint {https://arxiv.org/abs/2205.09734} {arXiv:2205.09734 [quant-ph]} \BibitemShut {NoStop}%
\bibitem [{\citenamefont {Beck}(2017)}]{Beck2017}%
  \BibitemOpen
  \bibfield  {author} {\bibinfo {author} {\bibfnamefont {J.}~\bibnamefont {Beck}},\ }\href@noop {} {\emph {\bibinfo {title} {Strong Uniformity and Large Dynamical Systems}}}\ (\bibinfo  {publisher} {World Scientific},\ \bibinfo {year} {2017})\ p.\ \bibinfo {pages} {458}\BibitemShut {NoStop}%
\end{thebibliography}%
\end{document}